\documentclass[aps,prd,10pt,notitlepage,nofootinbib,superscriptaddress,showkeys,showpacs]{revtex4-1}

 \pdfoutput=1

\usepackage{amsmath,amssymb,amsthm,latexsym,bbm}
\usepackage[pdftex]{hyperref}
\usepackage{color}
\usepackage{graphicx}

\newcommand{\SL}{\mathsf{SL}}

\newcommand{\sign}{\operatorname{sign}}

\newcommand{\C}{\ensuremath{\mathbb{C}} }

\newcommand{\R}{\ensuremath{\mathbb{R}} }
\newcommand{\N}{\ensuremath{\mathbb{N}} }

\newcommand{\be}{\begin{equation}}
\newcommand{\ee}{\end{equation}}
\newcommand{\bes}{\begin{eqnarray}}
\newcommand{\ees}{\end{eqnarray}}
\newcommand{\beq}{\begin{eqnarray}}
\newcommand{\eeq}{\end{eqnarray}}
\def\f{\frac}
\def\nn{\nonumber}
\def\bw{\bar{w}}
\def\bz{\bar{z}}
\def\bpsi{\bar{\psi}}
\def\ceta{\bar{\eta}}
\def\eps{\epsilon}

\def\te{\tilde{e}}
\def\tZ{\tilde{Z}}
\def\trho{\tilde{\rho}}
\def\la{\langle}
\def\ra{\rangle}
\newcommand{\id}{\mathbb{I}}
\def\tP{\tilde{P}}
\def\pp{\partial}
\def\Om{\Omega}
\def\vV{\vec{V}}
\def\vsigma{\vec{\sigma}}

\newcommand{\mat} [2] {\left ( \begin{array}{#1}#2\end{array} \right ) }

\newcommand{\cL}{{\mathcal L}}

\newcommand{\cM}{{\mathcal M}}

\newcommand{\cC}{{\mathcal C}}
\newcommand{\cW}{{\mathcal W}}
\newcommand{\cZ}{{\mathcal Z}}

\newcommand{\cP}{{\mathcal P}}
\newcommand{\cU}{{\mathcal U}}

\newcommand{\SU}{\mathrm{SU}}

\newcommand{\G}{\Gamma}
\newcommand{\Edges}{\mathsf{E}}
\newcommand{\Ver}{\mathsf{V}}
\newcommand{\EdgesO}{\vec{\mathsf{E}}}
\newcommand{\Angles}{\mathsf{A}}
\newcommand{\Pf}{\mathsf{Pf}}

\newcommand{\col}{col}
\newcommand{\e}{Y}
\newcommand{\an}{X}

\newtheorem{theo}{Theorem}[section]
\newtheorem{lemma}[theo]{Lemma}
\newtheorem{prop}[theo]{Proposition}

\theoremstyle{definition}
\newtheorem{defi}[theo]{Definition}
\newtheorem{rem}[theo]{Remark}

\theoremstyle{remark}

\begin{document}
\title[]{Duality between Spin networks and the 2D Ising model}

%\author[V. Bonzom]{Valentin Bonzom}
%\address{?}
%\email{?}
%\author[F. Costantino]{Francesco Costantino}
%\address{Institut de Recherche Math\'ematique Avanc\'ee\\
%  Rue Ren\'e Descartes 7\\
% 67084 Strasbourg, France}
%\email{costanti@math.unistra.fr}
%\author[E. Livine]{Etera Livine}
%\address{?}
%\email{?}
%
%\thanks{}

\author{Valentin Bonzom}\email{bonzom@lipn.univ-paris13.fr}
\affiliation{LIPN, UMR CNRS 7030, Institut Galil\'ee, Universit\'e Paris 13, Sorbonne Paris Cit\'e, 99, avenue Jean-Baptiste Cl\'ement, 93430 Villetaneuse, France, EU}
\author{Francesco Costantino}\email{francesco.costantino@math.univ-toulouse.fr}
\affiliation{Institut de Math\'ematiques de Toulouse,  Universit\'e de Toulouse III Paul Sabatier, 118 route de Narbonne, Toulouse 31062 France}
\author{Etera R. Livine}\email{etera.livine@ens-lyon.fr}
\affiliation{Laboratoire de Physique, ENS Lyon, CNRS UMR 5672, 46 All\'ee d'Italie, 69007 Lyon, France}
\affiliation{Korea Institute for Advanced Study, Seoul 130-722, Korea}

\date{\small \today}

\begin{abstract}

The goal of this paper is to exhibit a deep relation between the partition function of the Ising model on a planar trivalent graph and the generating series of the spin network evaluations on the same graph. We provide respectively a fermionic and a bosonic Gaussian integral formulation for each of these functions and we show that they are the inverse of each other (up to some explicit constants) by exhibiting a supersymmetry relating the two formulations. 

We investigate three aspects and applications of this duality. First, we propose higher order supersymmetric theories which couple the geometry of the spin networks to the Ising model and for which supersymmetric localization still holds. Secondly, after interpreting the generating function of spin network evaluations as the projection of a coherent state of loop quantum gravity onto the flat connection state, we find the probability distribution induced by that coherent state on the edge spins and study its stationary phase approximation. It is found that the stationary points correspond to the critical values of the couplings of the 2D Ising model, at least for isoradial graphs. Third, we analyze the mapping of the correlations of the Ising model to spin network observables, and describe the phase transition on those observables on the hexagonal lattice.

This opens the door to many new possibilities, especially for the study of the coarse-graining and continuum limit of spin networks in the context of quantum gravity.

\end{abstract}

\maketitle
\setcounter{tocdepth}{3}

\tableofcontents
%%%%%%%%%%%%%%%%%%%%%%%

%%%%%%
\section{Introduction}
%%%%%%

This paper is devoted to studying the relationship between two fundamental objects in physics, associated to finite graphs: the two-dimensional Ising model and the spin network evaluations on planar graphs. 

Given a graph $\G$ and a \emph{coloring} of the edges of $\Gamma$, i.e. a map $c:\Edges\to \mathbb{N}$ where $\Edges$ is the set of edges, the \emph{classical spin network evaluation} $\la \Gamma,c\ra$ is a rational number which is the result of  contracting some tensors over irreducible representations of $\SU(2)$. Spin networks arise in many areas, in particular related to physics. Since they come from the representation theory of $\SU(2)$, they are objects of prime interest in the theory of quantum angular momentum \cite{Ed} where they are often called Wigner symbols. As such, they have applications in atomic/molecular physics, chemistry, quantum information and so on \cite{Marzuoli}. More recent applications stem from quantum gravity as spin network equipped with holonomies are the states of loop quantum gravity \cite{SpinNetworksBaez, Thiemann} while their evaluations provide quantum gravity amplitudes, known as spin foams, \cite{PR, NouiPerez3D, SpinFoamBaez}. The latter are intimately related to lattice topological invariants, such as the Reidemeister torsion \cite{BarrettNaish, TwistedCohomology, CellularQuant} and the Turaev-Viro invariant of 3-manifolds.

Spin network evaluations can be computed in many different ways. The most famous is certainly the combinatorial definition due to R.~Penrose \cite{Pe}. In quantum gravity one often uses contractions of $\SU(2)$ intertwiners (notice that it requires an orientation on the edges while Penrose's definition does not -- we will prove the equivalence between those evaluations in the main text). In the last years it has been understood \cite{We, CoMa, Spin1/2Hamiltonian, bonzom, laurent}
that a good way to study spin network evaluations on a fixed graph $\Gamma$ is to organize it in a single generating series $Z^{spin}(\Gamma):=\sum_{c} \la \Gamma,c\ra {\bf Y}^c$ where the symbols ${\bf Y}^c$ stand for a suitable multivariate monomial in formal variables $Y_e, e\in \Edges$ (full details will be provided later).  

On a seemingly different side of physics (and mathematical physics), the 2D Ising model can be defined on the same graph $\G$. It probably is the most famous statistical model, based on the configuration space of maps from the set $\Ver$ of vertices of $\Gamma$ into $\{\pm 1\}$. The Hamiltonian (energy function) of the model is a sum of interactions between nearest-neighboring sites of $\G$, hence associated to the set $\Edges$ of edges of $\Gamma$, and weighted by couplings $y_e, e\in \Edges$. The partition function $Z^{Ising}(\Gamma)$ was proved by van den Waerden to be proportional to a sum over even subgraphs of $\Gamma$ weighted by some monomials in $\tanh(y_e), e\in \Edges(\Gamma)$ (see \cite{BaxterBook} for instance).

The present paper is motivated by the following observation made in \cite{co};
if $\Gamma$ is a planar trivalent graph and for each edge $e\in \Edges$ we set $Y_e:=\tanh(y_e)$ then the following equality holds:
\begin{equation} \label{FundamentalEq}
\left(Z^{Ising}(\Gamma)\right)^2Z^{spin}(\Gamma)=\left(\prod_{e\in \Edges}\cosh(y_e) \right)^22^{\# 2\Ver}.
\end{equation}
Such a  relation was independently noticed and shown to hold for the square 2D lattice with homogeneous couplings  in \cite{bianca}.
A straightforward proof of that equality can be given using van der Waerden high temperature expansion for the Ising model and Westbury's formula for the generating series of spin networks. However, this approach sheds no light on the intimate reason why the equality holds.

In the present paper we explain this phenomenon by identifying a supersymmetry which relates the Ising model on a trivalent planar graph $\Gamma$ to the generating series of spin networks on it. In order to achieve this, we:
\begin{enumerate}
\item represent $Z^{Ising}$ and its square via fermionic Gaussian integrals on a space of fermions indexed by the half-edges of $\Gamma$ (Propositions \ref{prop:GrassmannIsing} and \ref{prop:GrassmannIsingC}),
\item compute $Z^{spin}$ as a standard (bosonic) Gaussian integral on $\C^{2\#\Edges}$ (Theorem \ref{teo:complexgaussian}),
\item provide a supersymmetry relating the fermions and bosons associated to each half-edge (Section \ref{sec:supersymmetry}).
\end{enumerate}
Formulas for $Z^{Ising}$ based on fermionic (or Grassmanian) integrals have existed since the beginning of the 80's (see \cite{isingfermion1}, \cite{isingfermion2}, \cite{isingfermion3}). Our contribution here is to provide a new formula in which the integrand is a Gaussian whose bilinear form is defined using a Kasteleyn orientation on $\Gamma$ (which always exists by \cite{CR1}, \cite{CR2} as $\Gamma$ has an even number of vertices). Our formula is inspired by the integral representation of \cite{Sportiello} on the hexagonal lattice (where Kasteleyn orientations are not required though).

On the spin network side, a first Gaussian integral representation for $Z^{spin}$ was found in \cite{CoMa} (see also \cite{laurent, bonzom} for another Gaussian integral approach applicable for non-trivalent graphs), where a regularization process was needed to compute the integral. Here we provide another (very similar) formula via a convergent integral but based on a space whose dimension is twice as large. Another observation here is that the definition of $Z^{spin}$ we use depends on an orientation on $\Gamma$ while standard spin networks do not. In Theorem \ref{teo:comparison} we prove that our evaluation still coincides with the standard one, provided the orientation is a Kasteleyn orientation. 

The main point in the above formulas is that they really ``look similar''. This is formalized in Section  \ref{sec:supersymmetry} by means of a supersymmetry between the fermionic and bosonic degrees of freedom. We show that a Berezin integral whose argument is the product of the above two integrands equals a constant, by a simple supersymmetric localization argument. 

This unveils a deep relation between two precedently  unrelated objects. We start exploring aspects and consequences of this in Subsection \ref{sec:nontrivialcoupled} and then in Section \ref{sec:isingcorrelations}.
In Subsection \ref{sec:nontrivialcoupled} we introduce a supersymmetry-preserving generalization (whose extensive study will be pursued elsewhere) where the Ising model and the spin networks on $\Gamma$ are coupled non-trivially. In particular, the Berezin integral is not Gaussian anymore. We argue that the supersymmetric localization we have found will be useful to compute such modified versions of the generating function of the spin networks evaluations.

In Section \ref{sec:isingcorrelations}, we start by recalling the origin of the generating series $Z^{spin}$ in loop quantum gravity, as a coherent state (see also \cite{bonzom} and \cite{Spin1/2Hamiltonian}) whose background geometry is set by the couplings $Y_e$ and where the spins on the edges are quantum numbers of length. The expectations of products of length operators in this coherent state are moments of a probability distribution on the spins which depends on the couplings $Y_e$, $e\in\Edges$. The stationary phase approximation of this distribution, at large spins, leads to relations between the couplings $Y_e$ and the values of the spins at the stationary points. Using the physical meaning of the spins as quantum numbers of length, those relations can be written in geometric terms. Remarkably, they are found to be the same as those which define the critical couplings of the 2D Ising model on isoradial graphs \cite{isoradial, BoutillierDeTiliereSurvey}. In particular, the spins only determine the geometrical shape at the stationary points and can be arbitrary rescaled.

%expectation of the coupling with the empty state of the state represented by a large spin coloring of a planar grpah $\Gamma$ as a contribution coming from an euclidean geometry on the plane whose geometry is related to the values of the coupling constants $Y_e$ and we relate this geometry with the critical values of the Ising models on the same graph, at least when $\Gamma$ is an isoradial graph.  

We further use, in Section \ref{sec:correlations}, the fundamental equality \eqref{FundamentalEq} between $Z^{spin}$ and $Z^{Ising}$ to relate the observables of the Ising model, i.e. Ising spin correlations, to length observables in loop quantum gravity, in particular to matrix elements of products of length operators between the coherent state and the physical state (i.e. the flat connection state).
%compute the expectations of the measurements of the spin of an edge of $\Gamma$ given the correlation of the endpoints of the edge in the Ising model (see Section \ref{sec:correlations}). 
In Section \ref{sec:distributionspin} we push this further to compute the generating series of the matrix elements of the length operator of a single edge to an arbitrary power: it takes a simple closed form as a function of the nearest-neighbor correlation function of the Ising model.

%expectations of the observations of a single spin's edge (see Theorem \ref{teo:genserexpect}). 

Finally in Section \ref{vertexint}, another integral representation for coherent spin network states is introduced. It is based on integrals over spinors (variables on $\C^2$) associated to the vertices of the graph. This way, we show that one can define other generating functions which differ from $Z^{spin}$ by choices of combinatorial factors whilst still admitting Gaussian integral representations. We find in particular an example of such generating functions for which the stationary points of the distribution on the spins for the matrix elements of length operators between that coherent state and the flat connection state depend on the scale of the geometry, in contrast to the result obtained for the ordinary generating function $Z^{spin}$. This generalizes the comparison between different choices of generating functions of spin network evaluations started in \cite{bonzom} on 2-vertex graphs.

%One of the evidences of the above computations is that if one considers the stationary points for the coupling of a large spin network over a planar trivalent graph with the empty state, then the critical values come in lines (i.e. the geometries are defined only up to a global scale), and not in isolated points. To face this difficulty, we define another renormalization of spin network evaluations (see Equation \eqref{eq:newspin}),  which is natural from the point of view of Loop Quantum Gravity and whose critical points are associated to geometry not up to a scale.

The questions opened by the present work are multiple and we hope that the supersymmetry we here introduce between the 2D Ising model and the spin network evaluations will allow to exchange and cross-fertilize results in Statistical Mechanics and Loop Quantum Gravity. We outline some of those perspectives in Section \ref{sec:conclusions}.

%%%%%%
%\section{Planar Graphs, Embeddings and Orientations}
%%%%%%

\medskip
{\bf Acknowledgements}
The three authors profited of a PEPS funding from CNRS for the project ``Spin-Ising''. 
{\bf Notations.}
In all the paper we will let $\G$ be a finite graph, $\Ver$ be the set of its vertices, $\Edges$ the set of its edges, $\EdgesO$ the set of its half-edges or, equivalently of its oriented edges ($\EdgesO$ has a natural $2\to1$ map to $\Edges$), and $\Angles$ the set of its \emph{angles} i.e. pairs of distinct half edges having the same endpoint.
For each oriented edge  ${e}:s(e)\to t(e)$, we write $s(e),t(e)\in \Ver$ for its \emph{source} and \emph{target} vertices.
%We will also write an oriented edge as $\vec{e}:s(e)\to t(e)$ were $s(e),t(e)\in \Ver$ are its \emph{source} and \emph{target}, and an unoriented edge as $e: u\leftrightarrow v$, where, again $u,v\in \Ver$ are its endpoints.

Unless explicitly stated the contrary we will assume that $\G$ is planar, i.e. embedded (up to isotopy) in $S^2$. This automatically equips $\G$ with the datum of a cyclic, counter-clockwise ordering of the edges around each vertex. For each angle $\alpha$ around a vertex, following the cyclic ordering around that vertex, we call $s(\alpha)$ the source half-edge of the angle and $t(\alpha)$ its target half-edge.

Reciprocally, a cyclic ordering of the edges around each vertex allows to canonically thicken $\G$ to an oriented surface with boundary which we will denote $S(\G)$: since $\G$ is assumed to be planar, $S(\G)$ is homeomorphic to the complement of a collection open discs in $S^2$. A connected component of the complement of $S(\G)$ is called a \emph{face} of $\Gamma$ and it is naturally equipped with the counter-clockwise orientation.

$\G$ will further be assumed to be connected and bridgeless (i.e. 1-particle irreducible, or, equivalently no edge disconnects $\G$). Taking bridges into account is quite simple but requires to extend some definitions (like the equivalence class of Kasteleyn orientations) and it does not bring much to the theory. Recall indeed that from recoupling theory the spin of a bridge in a spin network has to vanish, which thus factorizes the spin network evaluation into two parts associated to disjoint graphs. Note that being bridgeless is equivalent to saying that each edge is incident to exactly two faces.

In the following we will associate functions to $\G$ which depend on one of two kinds of parameters: the $\e_{e}$ variables, indexed by $e\in \Edges$ and the $\an_{\alpha}$ variables, indexed by $\alpha\in \Angles$.

%%%%%%
\section{2D Ising Model}
%%%%%%

%%%
\subsection{Loop Expansion of the Ising Model}
%%%

\begin{defi}[Ising Model]
 An \emph{Ising spin configuration} on $\G$ is a map $\sigma:\Ver\to \{-1,+1\}$ associating $\pm1$ to each vertex of the graph. 
 The partition function of the Ising model on $\G$ is a function of couplings $y_{e}$ along each edge:
 $$Z^{Ising}(\G,\{y_e\})=\sum_{\sigma} \exp\left(\sum_{e\in\Edges } y_{e}\sigma_{s(e)}\sigma_{t(e)}\right).$$
%$$Z^{Ising}(\G,\{Y_e\})=\sum_{\sigma} \exp\left(\sum_{\Edges\ni e:u\leftrightarrow v } \e_{e}\sigma_u\sigma_v\right).$$
%To solve the Ising model is to compute explicitly the function $Z^{Ising}(\G)$ for $\G$.
 \end{defi}

The van der Waerden identity $\exp(y\sigma_u\sigma_v)=\cosh(\e)(1+\tanh(y)\sigma_u\sigma_v)$, for all $\sigma_{u,v}=\pm1$, allows to re-express the partition function as
\begin{equation*}
\begin{aligned}
Z^{Ising}(\G,\{y_e\})
&=
\big(\prod_{e\in \Edges}\cosh(y_e) \big)\sum_{\sigma}\prod_{e} (1+\tanh(y_e)\sigma_{s(e)}\sigma_{t(e)}) \\
&=
2^{\# \Ver} \big( \prod_{e\in\Edges} \cosh(y_e) \big)\sum_{\gamma \in \mathcal{G}}\prod_{e\in \gamma} \tanh(y_e)\nn
=
2^{\# \Ver} \big( \prod_{e\in\Edges} \cosh(y_e) \big)\sum_{\gamma \in \mathcal{G}}\prod_{e\in \gamma}\e_{e},
\end{aligned}
\end{equation*}
%\begin{multline}
%Z^{Ising}(\G,\{Y_e\})=\big(\prod_{e\in \Edges}\cosh(\e_e) \big)\sum_{\sigma}\prod_{e:u\leftrightarrow v} (1+\tanh(\e_e)\sigma_u\sigma_v)=\\
%\big( \prod_{e\in\Edges} \cosh(\e_e) \big)2^{\# \Ver} \sum_{\gamma \in \mathcal{G}}\prod_{e\in \gamma} \tanh(\e_e)
%\end{multline}
where we write $\e_{e}= \tanh(y_e)$ and  we sum over the set $\mathcal{G}$ of \emph{even subgraphs} of $\G$ (also known as \emph{Eulerian} subgraphs) i.e. subgraphs $\gamma\subset \G$ such that every vertex of $G$ is incident to an even number of edges of $\gamma$.

In this paper, we will focus on 3-valent graphs, i.e such that every vertex has exactly 3 edges attached to it. If $\Gamma$ is a 3-valent graph, then all even subgraphs are unions of disjoint loops. This is the high-temperature loop expansion of the Ising model, which we will match against the evaluation of spin networks.

%There is actually another loop expansion of the Ising partition function in terms of even subgraphs on the dual graph (cluster expansion). We will use this in the present work.

%%%
\subsection{Grassmannians for the 2D Planar Ising Model}
%%%

\begin{defi}[Kasteleyn orientation]
A Kasteleyn orientation on $\G$ is an orientation of the edges such that each face has an odd number of edges whose orientations do not match the one induced by the face.
\end{defi}

When drawing $\G$ on the plane, it means each face has an odd number of clockwise edges, see Fig. \ref{fig:Kasteleyn}. If $\G$ is embedded in the plane (and therefore the outerface, i.e. the connected component of its complement that is infinite, is not considered as a face), then Kasteleyn orientations exist (Kasteleyn's theorem). If $\G$ is embedded in $S^2$, then the notion of outerface is meaningless as which face is the outerface in the drawing depends on the choice of projection on the plane. If one insists on drawing $\G$ in the plane, the outerface then has to have an odd number of counter-clockwise edges. It is known that Kasteleyn orientations on cellular decompositions of oriented compact surfaces exist if and and only if the number of vertices is even \cite{CR1}, \cite{CR2}. A regular graph of degree 3 has an even number of vertices.

%We will further assume that the planar graph $\Gamma$ is 3-valent, because this is the relevant case to compare to spin recoupling and 3j-symbols. We nevertheless believe that the proofs below can be straightforwardly adapted to the case of arbitrary valence.

\smallskip

We introduce a set of Grassmann variables attached to half-edges, $\{\psi_{s(e)},\psi_{t(e)}\}_e$, which all anti-commute with one another. We define an edge action for each edge $e$ and a corner action for each angle $\alpha$:
\begin{equation}
I_e(\psi_{s(e)},\psi_{t(e)}) = \psi_{s(e)}\,\psi_{t(e)},
\qquad
I_\alpha(\psi_{s(\alpha)},\psi_{t(\alpha)}) = \psi_{s(\alpha)}\,\psi_{t(\alpha)}.
\end{equation}
In what follows we shall also use the notation $\psi^v_e$ to denote the Grassmann variable associated to the half-edge contained in $e$ and incident to $v$ (if it exists it is unique as $\Gamma$ is bridgeless). 

\begin{prop}[Grassmannian expression of the Ising model] \label{prop:GrassmannIsing}
The partition function of the Ising model on  a 3-valent planar $\Gamma$ (equipped with a Kasteleyn orientation) reads
\be
Z^{Ising}(\Gamma,\{y_e\})
= \left(\prod_{e\in \Edges}\cosh(y_e) \right)2^{\# \Ver}\,
Z_{f}(\Gamma,\{X_{\alpha}\})\,,
\ee
%\nn\\
$$
Z_{f}(\Gamma,\{X_{\alpha}\})\,=\,
\int \prod_e d\psi_{t(e)} d\psi_{s(e)}\ \exp\left(\sum_e I_e(\psi_{s(e)},\psi_{t(e)}) +\sum_\alpha X_\alpha\,I_\alpha(\psi_{s(\alpha)},\psi_{t(\alpha)}) \right),
$$
%\ee
%\begin{multline}
%Z^{Ising}(\Gamma,\{Y_e\}) = \left(\prod_{e\in \Edges}\cosh(\e_e) \right)2^{\# \Ver} \\
%\int \prod_e d\psi_{s(e)} d\psi_{t(e)}\ \exp\left(\sum_e I_e(\psi_{s(e)},\psi_{t(e)}) + \sum_\alpha X_\alpha\,I_\alpha(\psi_{s(\alpha)},\psi_{t(\alpha)}) \right),
%\end{multline}
where $X_\alpha = (\e_{s(\alpha)}\,\e_{t(\alpha)})^{1/2}$.
%where $X_\alpha = (\tanh y_{s(\alpha)}\,\tanh y_{t(\alpha)})^{1/2}$.
%
The Grassmannian integral is normalized as $\int \prod_e d\psi_{t(e)} d\psi_{s(e)} \prod_e \psi_{t(e)} \psi_{s(e)} = 1$.
\end{prop}

The proof is based on the following lemma on Kasteleyn orientations. Let $c$ be a cycle of $\Gamma$. When $\G$ and $c$ are drawn in $\R^2$, there is a well-defined inside and outside of the cycle. The disc inside $c$ induces a counter-clockwise orientation on $c$. A \emph{large angle} of $c$ is a pair of half-edges of the cycle both incident to the same vertex and such that the third half-edge incident to this vertex (and not contained in $c$) lies on the inside, as illustrated in fig.\ref{fig:Kasteleyn}.

\begin{figure}
\includegraphics[width=7cm]{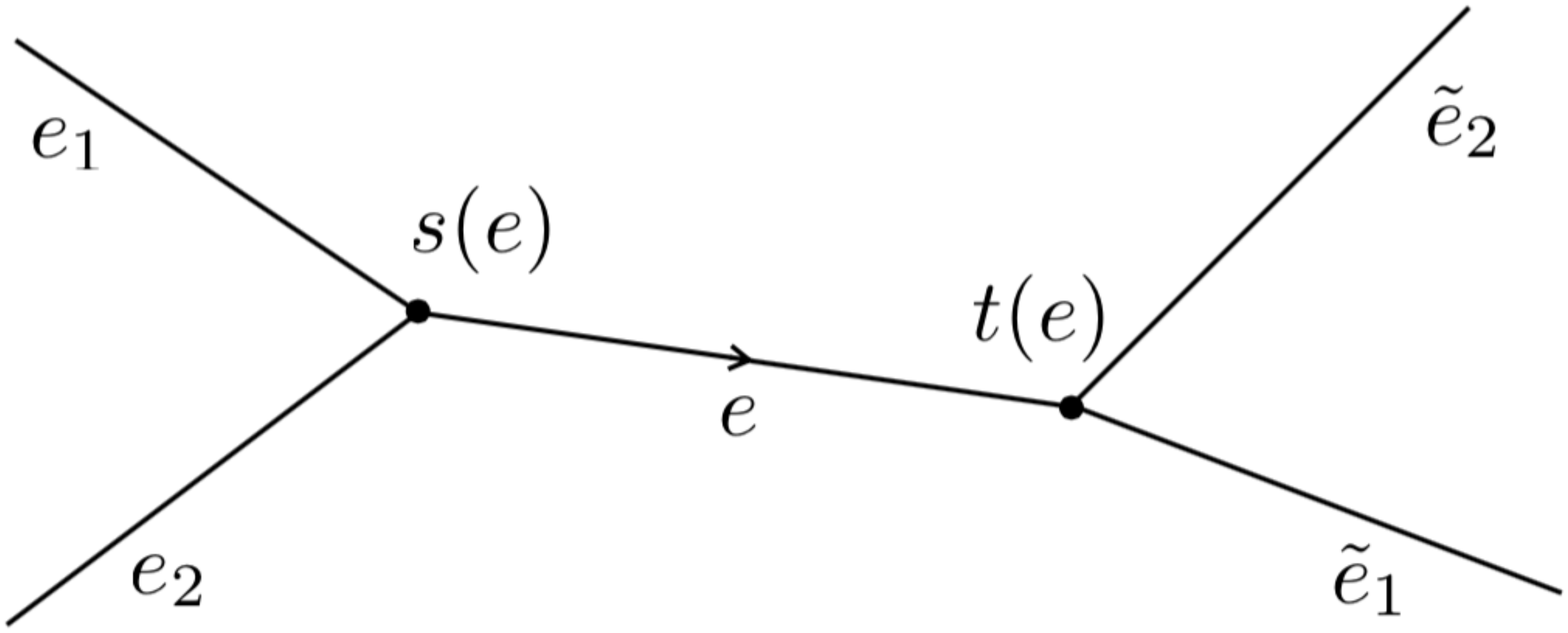}
\includegraphics[width=7cm]{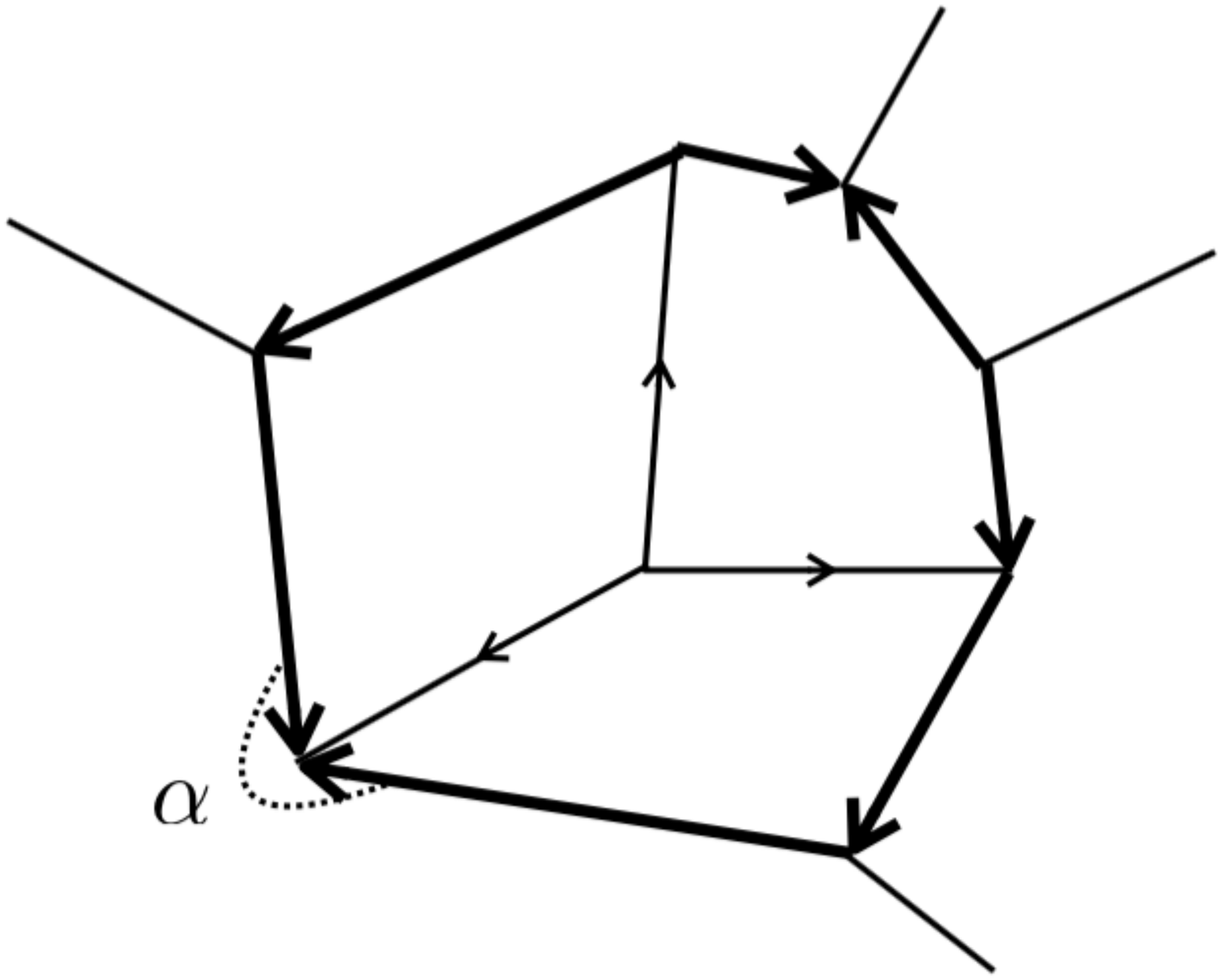}
\caption{Example of a Kasteleyn orientation around and inside a cycle $c$ (in
bold): each face of the planar graph contains a odd number of
clockwise oriented edges on its boundary, which implies that the
number of clockwise oriented edges around an arbitrary cycle is related
to the number of interior vertices; here $F=3$, $V_{int}=1$,
$E_{int}=3$, $a(c)=3$ ($\alpha$ is a large angle).}\label{fig:Kasteleyn}
\end{figure}

\begin{lemma} \label{lemma:KasteleynCycle}
Let $c$ be a cycle, $a(c)$ its number of large angles, $\#\Edges_{cl}(c)$ its number of clockwise edges, and $\#\Ver_{int}(c)$ the number of vertices on the inside. Then
\beq
%\begin{align}
&&(-1)^{a(c)} = (-1)^{\#\Ver_{int}(c)}, \label{LargeAngles}
\\
&&(-1)^{\#\Edges_{cl}(c)} = (-1)^{\Ver_{int}(c)+1}. \label{CCEdges}
%\end{align}
\eeq
\end{lemma}

\begin{proof}

We begin with the equation \eqref{LargeAngles}. Let $F(c)$ be the number of faces composing the disc bounded by $c$ and $\#\Edges_{int}$ be the number of internal edges in such disc. The subgraph made of the region inside the cycle $c$ and bounded by it (erasing all vertices in $c$ where $c$ does not form a large angle) is planar. Computing the Euler characteristic of $S^2$ cellularized via this planar subgraph we get:
\be
F(c)+1+(\#\Ver_{int}(c)+a(c))-(\#\Edges_{int}(c)+a(c))=2\,.
\label{EulerOnCycle}
\ee
The extra $+1$ counts the external face to the cycle $c$. Moreover, the number of vertices on the boundary cycle $c$ obviously equals the number of boundary edges, thus we are left with counting only internal vertices and internal edges.
On the other hand, we translate the fact that the graph is 3-valent into an equation relating the number of internal vertices and edges:
\be
3\#(\Ver_{int}(c)+a(c))=2(\#\Edges_{int}(c)+a(c))\,.
\ee
Combining these two equations yields:
\be
a(c)=2(F(c)-1)-\#\Ver_{int}(c)\,,
\ee
which implies the desired result \eqref{LargeAngles}.

We now prove the equation \eqref{CCEdges}. The quantity on the left hand side is obtained by assigning the weight $-1$ to clockwise edges and $1$ to counter-clockwise edges of $c$ and taking the product of those weights. We describe another way to do that. One proceeds the same way on each of the $F(c)$ faces in the interior, to get weights $\epsilon_f$. The product of those weights is not exactly what we want: each edge comes once with a $+1$ and once with a $-1$. Thus, it has to be corrected via a factor $(-1)^{\#\Edges_{int}(c)}$. Thus we get 
\begin{equation}
(-1)^{\#\Edges_{cl}(c)} = \prod_f \epsilon_f \times (-1)^{\#\Edges_{int}(c)}.
\end{equation}
Moreover, since the Kasteleyn orientation ensures an odd number of clockwise edges to each face, $\prod_f \epsilon_f = (-1)^{F(c)}$. Using again Euler's relation \eqref{EulerOnCycle}, one gets
\begin{equation}
(-1)^{\#\Edges_{cl}(c)} = (-1)^{F(c)}\ (-1)^{F(c)-1 -\#\Ver_{int}(c)},
\end{equation}
which simplifies to \eqref{CCEdges}.
\end{proof}

\emph{Proof of the Proposition \ref{prop:GrassmannIsing}.} We just have to show that the fermionic integral equals the expansion onto even subgraphs of $\Gamma$,
\begin{equation}
\int \prod_e d\psi_{t(e)} d\psi_{s(e)}\ \exp\left(\sum_e \psi_{s(e)}\psi_{t(e)} + \sum_\alpha X_\alpha \psi_{s(\alpha)}\psi_{t(\alpha)} \right) = \sum_{\gamma \in \mathcal{G}}\prod_{e\in \gamma} \e_e
\end{equation}
for $X_\alpha = (\e_{s(\alpha)}\,\e_{t(\alpha)})^{1/2}$. Notice that $\gamma\in\mathcal{G}$ being even, and $\Gamma$ being regular of degree 3, $\gamma$ is a disjoint union of cycles.

We simply expand the exponential of each edge and corner actions within the integral and commute the sums with the integral. Each exponential terminates at the linear order and only terms which saturate the number of $\psi$s survive.

If a term comes from $I_\alpha$ for some corner $\alpha$ made of the half-edges $h,g$, then $I_e$ cannot appear in this contribution if $e$ is an edge containing $h$ or $g$. Therefore the Grassmannian variable incident to $g$ (or $h$) and contained in the same edge as $g$ (or $h$) must come from a corner action and so on, until it closes to a cycle $c$. A typical contribution is thus labeled by a collection of disjoint cycles, i.e. an even subgraph $\gamma\in\mathcal{G}$, as expected. All half-edges not contained in $\gamma$ must come with their corresponding half-edges to form full edges and they get from $I_e$ the weight 1. Moreover each cycle receives the weight $\prod_{\alpha\in c} X_\alpha = \prod_{e\in c} \e_e$, up to a sign which is crucial to determine.

Consider a cycle $c$ and let $\#\Edges(c), \#\Edges_{cl}(c)$ be respectively the number of edges $c$ is formed of and the number of edges in $c$ which are oriented in the clockwise direction (here use use the fact that $c$ bounds a unique disc in $\R^2$). Since all the terms $I_\alpha$ commute with each other, we can reorder their product freely. We label the vertices of $c$ from 1 to $\ell$ following the counter-clockwise orientation and order the $I_{\alpha}$ accordingly; let us denote by $\psi^i_{j}$ the Grassmannian $\psi_h$ where $i$ denotes the vertex the half-edge $h$ is incident to and $j$ the vertex such that $h$ belongs to $e: i\leftrightarrow j$. The product of Grassmannians we get from expanding the exponentials is then 
\begin{equation}
(-1)^{\#E(c)-a(c)}\,\psi^1_{\ell} \psi^1_{2} \psi^2_{1} \psi^2_{3} \dotsm \psi^\ell_{\ell-1} \psi^\ell_{1} = (-1)^{\#E(c)-a(c)+1}\, \psi^\ell_{ 1} \psi^1_{\ell} \psi^1_{2} \psi^2_{1} \psi^2_{3} \dotsm \psi^\ell _{\ell-1}
\end{equation}
where we reordered the variables $\psi_{s(\alpha)}\psi_{t(\alpha)}$ exactly for the small angles i.e. where their order did not match that induced by the orientation of $c$ and  we have brought $\psi^\ell_{ 1}$ to the beginning so that the Grassmannians which lie on the same edge are next to each other, in the counter-clockwise order. Next, in order to match the measure $\prod_{e} d\psi_{t(e)}\psi_{s(e)}$ we flip the order of those Grassmannian variables associated to half-edges which belong to an edge oriented counter-clockwise. The remaining integral evaluates to 1. Therefore the sign of the cycle $c$ is
\begin{equation}
(-1)^{\#\Edges(c)-a(c)+1} \times (-1)^{\#\Edges(c)-\#\Edges_{cl}(c)} = 1,
\end{equation}
by Lemma \ref{lemma:KasteleynCycle}. As a conclusion, all cycles have a positive sign.
\qed

We can further push this path integral reformulation of the Ising model by using complex fermions instead of real fermions. This is a non-essential mathematical trick, doubling the number of variables associated to each half-edge of the graph (thus having other anticommuting variables $\bpsi_h$), but it will later allow to write explicitly the supersymmetry relating the Ising model to the spin network evaluations:

\begin{prop}[Complex Grassmannian expression of the Ising model] \label{prop:GrassmannIsingC}
The partition function of the Ising model on  a 3-valent planar $\Gamma$ (equipped with a Kasteleyn orientation) reads
\be
Z^{Ising}(\Gamma,\{y_e\})
%= \left(\prod_{e\in \Edges}\cosh(y_e) \right)2^{\# \Ver}\,
%Z_{f}(\Gamma,\{X_{\alpha}\})\,,
= \left(\prod_{e\in \Edges}\cosh(y_e) \right)2^{\# \Ver}\,
Z^{\C}_{f}(\Gamma,\{X_{\alpha}\})\,,
\ee
%\nn\\
$$
Z^{\C}_{f}(\Gamma,\{X_{\alpha}\})\,=\,
\int \prod_e d\psi_{t(e)}  d\psi_{s(e)} d\bpsi_{t(e)} d\bpsi_{s(e)}
\, \exp\left(\sum_{e,v} \psi_{e}^{v} \bpsi_{e}^{v}-\sum_e \bpsi_{s(e)}\bpsi_{t(e)}
+ \sum_\alpha X_\alpha\,\psi_{s(\alpha)}\psi_{t(\alpha)} \right)\,,
$$
where the Grassmannian integral is normalized as $\int \prod_e d\psi_{t(e)} d\psi_{s(e)} d\bpsi_{t(e)} d\bpsi_{s(e)} \prod_e \psi_{t(e)} \psi_{s(e)}\bpsi_{t(e)} \bpsi_{s(e)} = 1$, and we denote $\psi^v_e$ or $\bpsi^v_e$ a variable associated to the half-edge contained in an edge $e$ and incident to the vertex $v$. 

Furthermore it holds:
\be\label{eq:complexising}
Z_{f}^{\C}(\Gamma, \{X_{\alpha}\})^{2}
\,=\,
\int \prod_{h}
d\psi_h d\eta_h d\bpsi_h d\ceta_h\,
e^{  \sum_{e,v}\bigl(\psi_{e}^{v}\ceta_{e}^{v}+\bpsi_{e}^{v}\eta_{e}^{v}\bigr)- \sum_e \bigl(\bpsi_{s(e)}\bpsi_{t(e)} + \ceta_{s(e)}\ceta_{t(e)}\bigr)}
e^{\sum_\alpha X_\alpha \bigl(\psi_{s(\alpha)}\psi_{t(\alpha)} + \eta_{s(\alpha)}\eta_{t(\alpha)}\bigr)}\,.
\ee

\end{prop}
This can be proven by either repeating the same steps as above with the detail of orientations and signs, or by directly performing the Grassmannian integrals over the variables $\bpsi_{e}^{v}$. We detail below the explicit computation of the integral tracking all the signs.

\begin{proof} %Both equalities are The first equality is proved as that of Proposition \ref{prop:GrassmannIsing}.
The proof of the formula for $Z_f^{\C}$ is similar to that of Proposition \ref{prop:GrassmannIsing}: we need to find all the ways to obtain the top monomials $\prod_{e,v} \psi_{e}^v\bpsi_{e}^v$ by taking products of monomials of the forms $\psi_{e}^v\bpsi_{e}^v$, $(-\bpsi_{s(e)}\bpsi_{t(e)})$ and $\psi_{s(\alpha)}\psi_{t(\alpha)}$ and for each such way we need to check that the overall sign of the integral is $1$. 

Observe that in each such way of factorizing the top monomial, each time a factor $\psi_{s(\alpha)}\psi_{t(\alpha)} $ is present then necessarily also a factor $(-\bpsi_{s(e)}\bpsi_{t(e)})$ is present with $e$ being the edge containing $t(\alpha)$. Similarly, each time  a factor $(-\bpsi_{s(e)}\bpsi_{t(e)})$ is present then so is a factor of the form $\psi_{s(\alpha)}\psi_{t(\alpha)}$ where $\alpha$ is one of the two angles at the end of $e$.  This implies that each factorization of the top monomial corresponds to a disjoint union of cycles $c$ (corresponding to closed loops of factors $\prod_{\alpha\in c}\psi_{s(\alpha)}\psi_{t(\alpha)} \prod_{e\in c} \bpsi_{s(e)}\bpsi_{t(e)}$) and a monomial of the form $\prod_{e\notin c} \psi_{s(e)}\bpsi_{s(e)} \psi_{t(e)}\bpsi_{t(e)}$. 

Now orient each component of $c$ as induced by the disc it bounds in $\R^2$ and if $e_1,\ldots e_k$ are the edges encountered while circulating in $c$ and $v_i=e_{i-1}\cap e_i$ ($v_1=e_k\cap e_1$)  reorder the terms $ \psi_{e}^v$ and $\bpsi_{e}^v$ to get 
$$\pm \psi_{e_1}^{v_1}\psi_{e_1}^{v_2}\bpsi_{e_1}^{v_1}\bpsi_{e_1}^{v_2}\cdots \psi_{e_k}^{v_k}\psi_{e_k}^{v_1}\bpsi_{e_k}^{v_k} \bpsi_{e_k}^{v_1} \prod_{e\notin c} \psi_{s(e)}\bpsi_{s(e)} \psi_{t(e)}\bpsi_{t(e)}=\pm (\psi_{e_1}^{v_1}\psi_{e_1}^{v_2}\cdots \psi_{e_k}^{v_k}\psi_{e_k}^{v_1})(\bpsi_{e_1}^{v_1}\bpsi_{e_1}^{v_2}\cdots\bpsi_{e_k}^{v_k} \bpsi_{e_k}^{v_1}) \prod_{e\notin c} \psi_{s(e)}\bpsi_{s(e)} \psi_{t(e)}\bpsi_{t(e)}.$$

Let us analyze the sign we get from the right hand side. We permute the variable $\psi_h$ exactly as in the proof of Proposition \ref{prop:GrassmannIsing}: we first reorder internally all the degree two monomials $\psi_{s(\alpha)}\psi_{t(\alpha)}$ according to the order induced by $c$, thereby acquiring a $-1$ sign if and only if $\alpha$ is a small angle.  Then we permute $\psi^1_k$ with all the other terms $\psi^i_j$, thus acquiring a $-1$ sign. At this stage we get the monomial:
\begin{multline*}
(-1)^{\#\Edges(c)-a(c)+1}(\psi^1_1\psi^2_1\cdots \psi^k_{k}\psi^1_k)(\prod_{e\in c} -\bpsi_{s(e)}\bpsi_{t(e)})(\prod_{h\notin c}\psi_h\bpsi_h)= (-1)^{\#\Edges(c)-a(c)+1}(\psi^1_1\psi^2_1\cdots \psi^k_{k}\psi^1_k)(\prod_{e\in c} \bpsi_{t(e)}\bpsi_{s(e)})(\prod_{h\notin c}\psi_h\bpsi_h)
\\
=(-1)^{2\#\Edges(c)-\#\Edges_{cl}+a(c)+1}(\prod_{e\in c} \psi_{t(e)}\psi_{s(e)})(\prod_{e\in c} \bpsi_{t(e)}\bpsi_{s(e)})(\prod_{h\notin c}\psi_h\bpsi_h)=(\prod_{e\in c} \psi_{t(e)}\psi_{s(e)}\bpsi_{t(e)}\bpsi_{s(e)})(\prod_{h\notin c}\psi_h\bpsi_h)
\end{multline*}
where in the second equality we reordered the terms $\psi^i_i\psi^{i+1}_i$ if and only if the edge was oriented counter-clockwise (i.e. if and only if $e_i:v_i\to v_{i+1}$), and in the third we used Lemma \ref{lemma:KasteleynCycle}.
So the Grassmannian integral of this term is $+\prod_{\alpha\subset c} X_{\alpha}=\prod_{e\in c} Y_e$ as claimed. 

%
%For each half edge $(e_i,v_{i})$ (or $(e_i,v_{i+1})$) in $c$ we then need to permute a factor $\bpsi^{v_i}_{e_i}$ (resp. $\bpsi^{v_{i+1}}_{e_i}$) coming from the monomial $-\bpsi_{s(e)} \bpsi_{t(e)}$ associated to the edge containing $(e_i,v_i)$ with the corresponding factor $\psi^{v_i}_{e_i}$ (resp. $\bpsi^{v_{i+1}}_{e_i}$) coming from the factor $ \psi_{t(\alpha)}\psi_{s(\alpha)}$ of the angle containing it: we then recover a factor $(-1)$ for each edge in $c$ which compensates the factor $-1$ present in the exponential in front of $\bpsi_{s(e)} \bpsi_{t(e)}$.   
%The terms corresponding to the edges not in $c$ do not need to be reordered,  thus the overall sign is $(-1)^{a(c)+\#E_{cl}(c)}=-1$ by Lemma \ref{lemma:KasteleynCycle} ; finally in the above re-arrangement we have to permute the variable $\psi^1_{\ell}$ with all the remaining ones and this compensates the $-1$ sign. 

Before moving to the proof of the second statement, let us first define antisymmetric matrices $A$ and $B$ of size $2\#\Edges\times 2\#\Edges$ as follows:
$$\sum_e -\bpsi_{s(e)}\bpsi_{t(e)} =\frac{1}{2} \vec{\bpsi}^tA\vec{\bpsi},\qquad \sum_\alpha X_\alpha \psi_{s(\alpha)}\psi_{t(\alpha)}= \frac{1}{2} \vec{\psi}^tB\vec{\psi}.$$
Then the first statement implies the following:
$$
Z^{\C}_{f}(\Gamma,\{X_{\alpha}\})\,=\pm \Pf \left(\begin{array}{c|c} B & Id \\ \hline -Id & A\end{array}\right)=\pm\sqrt{\det(I+AB)}
$$
where $\Pf(X)=\sqrt{\det(X)}$ is the Pfaffian of the antisymmetric matrix $X$, in the last passage we used the fact that the determinant of a $2\times 2$-block matrix $\left( \begin{smallmatrix} M_1 & M_2\\ M_3 & M_4\end{smallmatrix}\right)$ is $\det(M_1-M_2M_4^{-1}M_3)\det(M_4)$ if $M_4$ is invertible and here $M_4=A$ is such that $A^2=-Id$. 

We now remark that the integral considered in the second statement equals the Pfaffian of the antisymmetric matrix $M$ of size $8\#\Edges\times 8\#\Edges$ which in the basis given by $\{\psi_e\},\{\eta_e\},\{\bpsi_e\},\{\overline{\eta}_e\}$ is:
\begin{equation}
M=\left(\begin{array}{c|c|c|c}
B & 0 & 0 & -Id\\
\hline
0 & B & -Id & 0\\
\hline
0 & Id & A & 0\\
\hline
Id & 0 & 0 & A
\end{array}\right)\end{equation}
Observing that $A^2=-Id$ and using  again the fact that the determinant of a $2\times 2$-block matrix $\left( \begin{smallmatrix} M_1 & M_2\\ M_3 & M_4\end{smallmatrix}\right)$ is $\det(M_1-M_2M_4^{-1}M_3)\det(M_4)$ if $M_4$ is invertible, we get:
\begin{equation}
\det(M)=\det(A)^2\det(A+B)^2\implies \Pf(M)=\det(A+B)=Z^{\C}_{f}(\Gamma,\{X_{\alpha}\})^2.
\end{equation}
 \end{proof}

%%%
\subsection{From Edge Variables to Angle Variables and Back}\label{sec:mappings}
%%%

Although we started with the Ising partition function $Z^{Ising}(\Gamma,\{y_e\})$ as a function of the edge variables $y_{e}$ (or $Y_{e}=\tanh y_{e}$), it was convenient to switch to angle variables $X_{\alpha}$ to define and compute its reformulation as an odd-Grassmannian integral. This relies on the following fact:
\begin{lemma}
Given variables $Y_{e}$ on the edges of a graph $\Gamma$, assumed to be planar, connected and bridgeless, if we define angle variables $X_{\alpha}=(Y_{s(\alpha)}Y_{t(\alpha)})^{\f12}$, then we have the following equality for any closed loop $\cL$:
\be
\prod_{\alpha\in\cL}X_{\alpha}=\prod_{e\in\cL}Y_{e}\,.
\ee
\end{lemma}
It is interesting to invert this mapping and check if the equality between the Ising partition function and the fermionic path integral holds for arbitrary angle couplings $X_{\alpha}$. Of course, the number of angle variables on a 3-valent graph is twice the number of edge variables. It is nevertheless possible to reverse the above mapping:
\begin{lemma}
Given variables $X_{\alpha}$ on the angles of a graph $\Gamma$, assumed to be 3-valent, planar, connected and bridgeless, if we define edge variables following the convention of the left part of Figure \ref{fig:Kasteleyn}.
\be
\label{angletoedge}
Y_{e}=
\left(
\f{X_{ee_{1}}X_{ee_{2}}}{X_{e_{1}e_{2}}}\f{X_{e\te_{1}}X_{e\te_{2}}}{X_{\te_{1}\te_{2}}}
\right)^{\f12}\,,
\ee
where $e_{1,2}$ are the two other edges incident to the source vertex $s(e)$ while $\te_{1,2}$ are the two other edges incident to the target vertex $t(e)$,
then we have the following equality that holds for any closed loop $\cL$:
$$
\prod_{e\in\cL}Y_{e}=\prod_{\alpha\in\cL}X_{\alpha}\,.
$$
\end{lemma}
These two mappings allow to consider as primary variables either the edge variables $Y_{e}$ or the angle variables $X_{\alpha}$.%, as most convenient depending on the question at hand.

We further notice that if we start with edge variables $Y_{e}$ and we apply the relation $X_{\alpha}=(Y_{s(\alpha)}Y_{t(\alpha)})^{\f12}$ to the mapping \eqref{angletoedge}, then we recover the initial variables $Y_{e}$. Clearly the reverse is not true, due to the fact that these mappings are not one-to-one.

%%%%%%
\section{The Generating Function of Spin Network Evaluations}

We now turn to the realm of quantum geometry and 3d (Euclidean) quantum gravity, where the quantum states of 2D geometry are constructed from the 3nj-symbols of spin recoupling \cite{3DHamiltonian, Spin1/2Hamiltonian}, and so are transition amplitudes for instance in the Ponzano-Regge model (see \cite{PR}, and \cite{BarrettNaish, TwistedCohomology, SortingOut, CellularQuant} for more recent developments).

%%%
\subsection{The Hilbert Space of Spin Networks and their Evaluations}
%%%

We consider a planar, 3-valent  and connected graph $\Gamma$.
We equip the graph with a counter-clockwise cyclic ordering of the edges incident to each vertex.
Let us also choose an orientation $o$ of the edges of $\G$.
For the correspondence with the Ising model (which requires to get a sign $+$ for all loops), this orientation will be chosen to be of the Kasteleyn type.
%
%Recall that $\G$ is equipped with with a counter-clockwise cyclic orientation of the edges incident to each vertex. 

We will discuss the (in)dependence of the spin network evaluation with respect to the chosen orientation and the relation between our (tensorial) definition of evaluation and the other standard definitions of 3nj-symbols in the next section.

\smallskip

In the context of loop quantum gravity and spinfoam models \cite{SpinNetworksBaez, Thiemann}, we define the Hilbert space of wave-functions for the 2D quantum geometry on the graph $\Gamma$ as the space of gauge-invariant functions of $\SU(2)$ group elements living along the oriented edges of the graph:
\be
f(\{g_{e}\}_{e\in\Gamma})
\,=\,
f(\{h_{t(e)}^{-1}g_{e}h_{s(e)}\}_{e\in\Gamma}),
\qquad \forall
\{h_{v}\}_{v\in\Ver}\in\SU(2)^{\# \Ver}\,,
\ee
whose scalar products are defined by the Haar measure on $\SU(2)$:
\be
\la \psi | \phi\ra
\,=\,
\int_{\SU(2)^{\# \Edges}}
[dg_{e}]\, \overline{\psi(\{g_{e}\})}\,\phi(\{g_{e}\})\,.
\ee

\smallskip

A basis of this Hilbert space $L^{2}(\SU(2)^{\#\Edges}/\SU(2)^{\#\Ver})$ of square integrable functions over $\SU(2)^{\#\Edges}$ is provided by the spin network functions. 
For an oriented 3-valent graph $\Gamma$, spin network states are labeled by a graph coloring.
%
%A coloring of the graph consists of (possibly vanishing) integers labeling the edges, $2j_e\,\in\N$ for the edge $e$, where $j_e\in \mathbb{N}/2$ is a half-integer which characterizes a spin, that is an irreducible representation of $\SU(2)$ of dimension $(2j_e+1)$.
%
\begin{defi}
A \emph{coloring} of $\G$ is a map $\col:E\to \N$; we will denote the value of the coloring on an edge $e$ by $2j_e$, and we will call $j_e$ the \emph{spin} of the edge for the given coloring. 
\end{defi}
The spin characterizes an irreducible representation of $\SU(2)$ of dimension $(2j_e+1)$ whose basis vectors are indexed as $e^j_m$ (also denoted $|j,m\rangle$), for $m\in \{-j_e,-j_e+1,\ldots ,j_e\}$; these vectors form the magnetic number basis. 
The spin network function is then defined as:
\be
\label{spinbasis}
\varphi_{\{j_e\}}^{\Gamma}(g_{e})
\,\equiv\,
\sum_{\{m_e^{v}\}}
\prod_e (-1)^{j_e-m_e^{t}}D^{j_{e}}_{-m_{e}^{t},m_{e}^{s}}(g_{e})\,
\prod_v \begin{pmatrix} j_{e_{1}^{v}} & j_{e_{2}^{v}} & j_{e_{3}^{v}}\\ m_{e_{1}}^{v} &  m_{e_{2}}^{v}&  m_{e_{3}}^{v}\end{pmatrix}\,\,,
\ee
where the edges around the vertex $v$ are cyclically ordered  as $j_{e_{1}^{v}}, j_{e_{2}^{v}}, j_{e_{3}^{v}}$ with the magnetic moment labels $m_{e}^{v}$ living on every half-edge around every vertex. Finally $D^{j}_{m,m'}(g)=\la j,m|g|j,m'\ra$ is the Wigner matrix representing the $\SU(2)$ group element $g$ in the irreducible representation of spin $j$.
The symbol $\left(\begin{smallmatrix} j_1 &j_2 &j_3\\ m_1 &m_2 &m_3 \end{smallmatrix}\right)$ is the Wigner 3j-symbol (see \cite{Ed}, Chapter 3.7). It is invariant under rotations, meaning that it is an intertwiner between the tensor product of the three representations and the trivial one. In other words, the vector $\sum_{m_1,m_2,m_3} \left(\begin{smallmatrix} j_1 &j_2 &j_3\\ m_1 &m_2 &m_3 \end{smallmatrix}\right)\, e^{j_1}_{m_1}\otimes e^{j_2}_{m_2}\otimes e^{j_3}_{m_3}$ is a basis of the (1-dimensional) space of $\SU(2)$-invariant vectors in the tensor product of the three representations. The entries of the projector onto the invariant space can be written as
\begin{equation*}
\int_{\SU(2)} dg\ D^{(j_1)}_{m_1 n_1}(g)\,D^{(j_2)}_{m_2 n_2}(g)\,D^{(j_3)}_{m_3 n_3}(g) = \begin{pmatrix} j_1 &j_2 &j_3\\ m_1 &m_2 &m_3 \end{pmatrix}\,\begin{pmatrix} j_1 &j_2 &j_3\\ n_1 &n_2 &n_3 \end{pmatrix}.
\end{equation*}
The orthogonality and completeness relations read
\begin{equation*}
\sum_{m_{1},m_{2}}
\begin{pmatrix} j_{1} & j_{2} & j \\ m_1& m_2 & m \end{pmatrix}
\begin{pmatrix} j_{1} & j_{2} & j' \\ m_1& m_2 & m' \end{pmatrix}
\,=\,
\f1{2j+1}\,\delta_{jj'}\,\delta_{mm'},
\qquad
\sum_{j,m} (2j+1)\,
\begin{pmatrix} j_{1} & j_{2} & j \\ m_1& m_2 & m \end{pmatrix}
\begin{pmatrix} j_{1} & j_{2} & j \\ m'_1& m'_2 & m \end{pmatrix}
\,=\,
\delta_{m_{1}m'_{1}}\,\delta_{m_{2}m'_{2}}.
\end{equation*}
Wigner 3j-symbols are invariant under cyclic permutations of their three entries and pick up a sign upon reversing the cyclic orientation as well as upon reversing the orientations of the incident edges:
\begin{equation*}
\begin{pmatrix} j_1 &j_2 &j_3\\ m_1 &m_2 &m_3 \end{pmatrix} = (-1)^{j_1+j_2+j_3} \begin{pmatrix} j_1 &j_3 &j_2\\ m_1 &m_3 &m_2 \end{pmatrix} = (-1)^{j_1+j_2+j_3} \begin{pmatrix} j_1 &j_2 &j_3\\ -m_1 &-m_2 &-m_3 \end{pmatrix} = \delta_{m_1+m_2+m_3,0} \begin{pmatrix} j_1 &j_2 &j_3\\ m_1 &m_2 &m_3 \end{pmatrix}.
\end{equation*}

The central object of our paper is the spin network evaluation, defined as the evaluation of the spin network function at the identity on all edges, $g_{e}=\id$.

\begin{defi}
For a colored 3-valent  graph $\Gamma$, equipped with an edge orientation $o$, we define its  ``unitary (up to sign) spin network evaluation" as
\begin{equation}
\label{eq:tensorsn}
s^{\Gamma}(\{j_e\},o)
\,\equiv\,
\varphi_{\{j_e\}}^{\Gamma}(\id)
\,=\,
\sum_{\{m_e\}} \prod_e (-1)^{j_e-m_e} \prod_v \begin{pmatrix}  j_{e_{1}^{v}} & j_{e_{2}^{v}} & j_{e_{3}^{v}}\\ \epsilon_{e_{1}} ^{v}m_{e_{1}^{v}} & \epsilon_{e_{2}}^{v} m_{e_{2}^{v}}& \epsilon_{e_{3}}^{v} m_{e_{3}^{v}} \end{pmatrix},
\end{equation}
with the orientation sign $\epsilon_{e_{i}}^{v}=-1$ resp. $1$ recording if the edge is oriented inwards ($v=t(e)$) resp. outwards ($v=s(e)$).
\end{defi}

For instance, Wigner's 6j-symbol can be defined in this way, on the (appropriately oriented) tetrahedral graph,
\begin{equation} \label{6jDef}
\begin{Bmatrix} j_1 &j_2 &j_3 \\ j_4 &j_5 &j_6\end{Bmatrix} = \sum_{\substack{m_1, m_2, m_3,\\ m_4, m_5, m_6}} (-1)^{\sum_{i=1}^6 j_i-m_i} \begin{pmatrix} j_1 &j_2 &j_3\\ m_1 &m_2 &m_3\end{pmatrix} \begin{pmatrix} j_1 &j_5 &j_6\\ -m_1 &-m_5 &m_6\end{pmatrix} \begin{pmatrix} j_3 &j_4 &j_5\\ -m_3 &-m_4 &m_5\end{pmatrix} \begin{pmatrix} j_2 &j_6 &j_4\\ -m_2 &-m_6 &m_4\end{pmatrix}.
\end{equation}

From the definition, it can be seen that flipping the orientation of an edge with color $2j_e$ changes the spin network evaluation by a sign $(-1)^{2j_e}$. Our definition thus really seems to depend on the choice of an orientation $o$. However, we will show in the next section that for a fixed graph and a fixed coloring, the evaluation is constant on the class of Kasteleyn orientations (and can be explicitly related to other definitions of spin network evaluations). Then for the sake of simplicity, we will drop the $o$ from the notation and simply write $s^{\Gamma}(\{j_e\})$ when there can not be any confusion.

%In practice, the conventions for this evaluation work as follows. Let us consider three representations meeting at a vertex, say $j_1, j_2, j_3$ in the counter-clockwise order. If all edges are outwards, we associate to the vertex the invariant tensor $i_{j_1 j_2 j_3}$ whose components in the magnetic number basis are the $3j$-symbols $\left(\begin{smallmatrix} j_1 &j_2 &j_3\\ m_1 &m_2 &m_3 \end{smallmatrix}\right)$, with $-j_e\leq m_e\leq j_e$ for $e=1,2,3$. This tensor is invariant under cyclic permutations of its three entries.
%%
%If an edge is oriented outwards, say the edge $e=1$, we have to work with the representation dual to $j_1$, which in practice means that the components of the invariant tensor read (in the same basis) $(-1)^{j_1-m_1} \left(\begin{smallmatrix} j_1 &j_2 &j_3\\ -m_1 &m_2 &m_3 \end{smallmatrix}\right)$.

\smallskip

In the context of quantum gravity, the spin network evaluation is interpreted as the scalar product of the spin network wave-function with the flat connection state. This actually gives the physical solution for 3d quantum gravity (on $S^2\times [0,1]$) and more generally for topological BF theory.
%\footnotemark.
%
%\footnotetext{
Considering an arbitrary gauge-invariant function $\psi$, we can decompose it on the spin network basis states $\varphi_{\{j_e\}}^{\Gamma}$ upon integrating it against the $\SU(2)$ characters:
$$
\psi(\{g_{e}\})
\,=\,
\sum_{\{j_{e}\}}\psi_{\{j_{e}\}}\,\varphi_{\{j_e\}}^{\Gamma}(g_{e})\,,
\quad
\psi(\id)
\,=\,
\sum_{\{j_{e}\}}\psi_{\{j_{e}\}}s^{\Gamma}(\{j_e\})\,,
\quad
\int 
\prod_{e}[dg_{e}]\,\chi_{j_{e}}(g_{e})\,\psi(\{g_{e}\})
\,=\,
\f1{\prod_{e}(2j_{e}+1)}\,\psi_{\{j_{e}\}}\,s^{\Gamma}(\{j_e\})\,.
$$
Due to the orthogonality of the 3j-symbols, the scalar product between two states $\psi$ and $\phi$ simply reads in this basis as:
$$
\la \psi | \phi\ra
\,=\,
\sum_{\{j_{e}\}}\f1{\prod_{e}(2j_{e}+1)}\,
\overline{\psi_{\{j_{e}\}}}\,\phi_{\{j_{e}\}}\,.
$$
We now introduce the gauge-invariant flat state, physical solution for the topological BF theory:
\begin{equation} \label{FlatState}
\Omega(\{g_{e}\})
\,\equiv\,
\int [dh_{v}]\,
\prod_{e}\delta(h_{s(e)}^{-1}g_{e}h_{t(e)})\,.
\end{equation}
Carefully tracking the inverse group elements and the signs by the following identity on Wigner matrices:
$$
D^{j}_{m,n}(h^{-1})
=\left(D^{j}(h)\right)^{\dagger}_{m,n}
=\overline{D^{j}_{n,m}(h)}
=(-1)^{(n-m)}D^{j}_{-n,-m}(h)\,,
$$
we get the simple decomposition of the flat state and scalar product property:
$$
\Omega(\{g_{e}\})
\,=\,
\sum_{\{j_{e}\}}\prod_{e}(2j_{e}+1)\,s^{\Gamma}(\{j_e\})\,\varphi_{\{j_e\}}^{\Gamma}(g_{e})\,,
\quad
\la \Omega | \Omega\ra
\,=\,
\sum_{\{j_{e}\}}\prod_{e}(2j_{e}+1)\,s^{\Gamma}(\{j_e\})^{2}\,,
\quad
\la \Omega | \phi\ra
\,=\,
\sum_{\{j_{e}\}} s^{\Gamma}(\{j_e\})\,\,\phi_{\{j_{e}\}}
\,=\,
\phi(\id)\,.
$$
%}
%
Spinfoam transition amplitudes between  spin network states are then typically constructed from such spin network evaluations or projections on the flat state.

\smallskip

%Finally,  our main object of interest here is the following generating function of these spin network evaluations in terms of the edge variables $\{Y_e\}_{e\in\Gamma}$:
%\be
%Z^{Spin}(\Gamma,\{Y_e\})
%\,\equiv\,
%\sum_{\{j_e\}} \sqrt{\frac{\prod_v (J_v+1)!}{\prod_{ev} (J_v-2j_e)!}} s(\{j_e\}) \prod_e Y_e^{2j_e}\,.
%\ee

%%%
\subsection{Choice of Normalizations and Orientations}
%%%

In order to relate our definition of ``unitary up to a sign'' evaluations to other standard definitions, a couple of recoupling identities are required.

The following equality is equivalent to the definition \eqref{6jDef} of 6j-symbols (the equivalence can be proved using the properties of 3j-symbols listed in the previous section -- see also formula 6.2.6 in  \cite{Ed}),
\begin{multline} \label{eq:whitehead}
\sum_{m_{12}} (-1)^{j_{12}-m_{12}}\begin{pmatrix} j_1& j_{12}& j_2\\ m_1& m_{12}& m_2\end{pmatrix} \begin{pmatrix} j_{12}& j& j_3\\ -m_{12}& m& m_3\end{pmatrix}\\
= (-1)^{2j_1} \sum_{j_{23}}
(2j_{23}+1) \begin{Bmatrix} j_1& j_2& j_{12}\\ j_3& j& j_{23}\end{Bmatrix} (-1)^{j_1+j_2+j_3+j} \sum_{m_{23}} \begin{pmatrix} j_1& j& j_{23}\\ m_1& m& -m_{23}\end{pmatrix} \begin{pmatrix} j_2& j_{23}& j_3\\ m_2& m_{23}& m_3\end{pmatrix} (-1)^{j_{23}-m_{23}}.
\end{multline}
Note that the sums on both sides are trivial (the only non-vanishing contributions come from $m_{12} = -m_1 - m_2$ and $m_{23} = m_1 + m$). One can also flip the signs of $m_1, m_2, m_3$ and/or $m$ and add factors like $(-1)^{j_1-m_1}$ on both sides of the equation without any more changes. We will also need the well known orthogonality relation,
\begin{align}\label{eq:orthogonality}
\sum_{m_1,m_2}(-1)^{j_1-m_1+j_2-m_2+j-m}(-1)^{j_1+j_2+j'}\left(\begin{matrix} j_2 & j_1 &j\\ -m_2 & -m_1 & -m\end{matrix}\right)\left(\begin{matrix} j_1 & j_2 &j'\\ m_1 & m_2 & m'\end{matrix}\right)=\frac{\delta_{j,j'}\delta_{m,m'}}{2j+1}.
\end{align}

The above equalities can be used to modify the summand in the evaluations \eqref{eq:tensorsn}. Equation \eqref{eq:whitehead} applies to any edge (with spin $j_{12}$ in the notation of \eqref{eq:whitehead}), while \eqref{eq:orthogonality} gets rid of faces of degree two. They translate graphically to
\begin{equation}\label{eq:whiteheadgraphical}
\raisebox{-1cm}{\put(-10,50){$j_1$}\put(24,50){$j_2$}\put(52,50){$j_3$}\put(32,8){$j$}\put(12,27){$j_{12}$}\includegraphics[width=2cm]{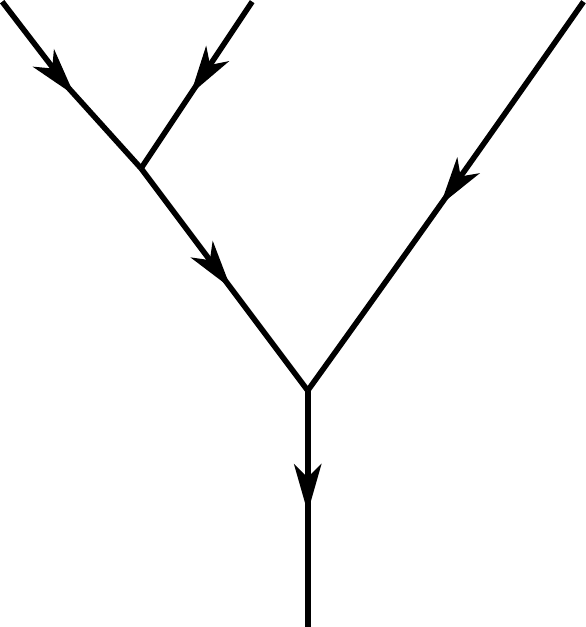}}=(-1)^{2j_1}\sum_{j_{23}} (2j_{23}+1) (-1)^{j_1+j_2+j_3+j}\left\{\begin{matrix} j_1 & j_2 &j_{12}\\ j_3 & j & j_{23} \end{matrix}\right\}\raisebox{-1cm}{\put(-7,50){$j_1$}\put(27,50){$j_2$}\put(52,50){$j_3$}\put(32,8){$j$}\put(34,27){$j_{23}$}\includegraphics[width=2cm]{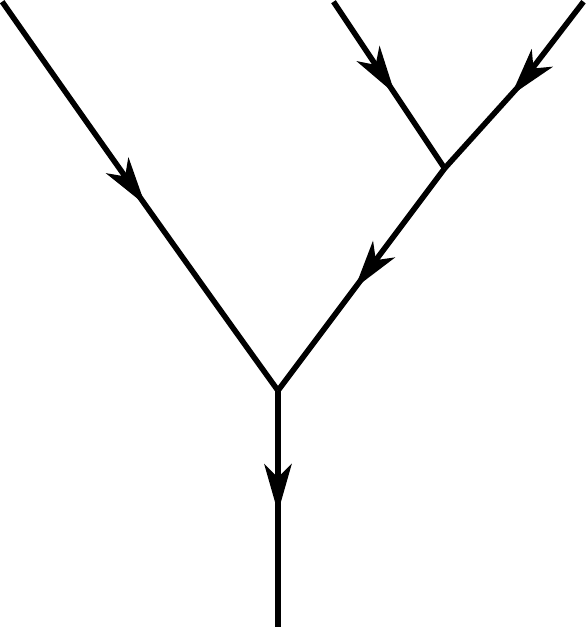}}
\end{equation}
\begin{equation}\label{eq:unzip}
\raisebox{-1cm}{\put(20,50){$j$}\put(30,6){$j'$}\put(8,27){$j_{2}$}\put(40,27){$j_{1}$}\includegraphics[width=2cm]{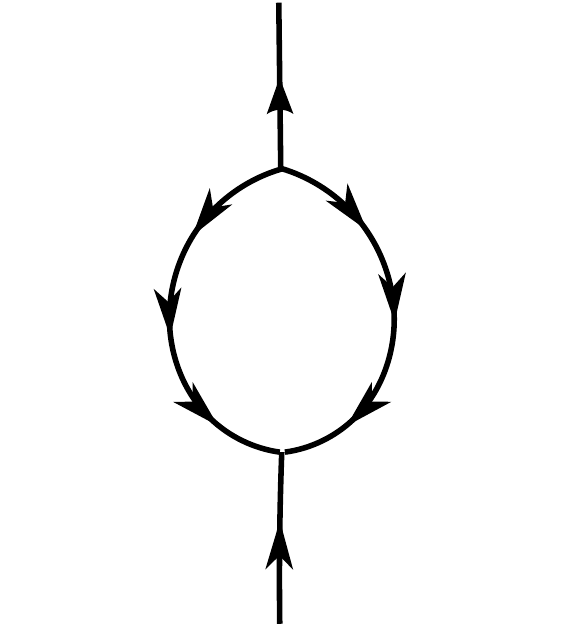}}=\frac{(-1)^{j_1+j_2+j}\delta_{j,j'} }{2j+1}\raisebox{-1cm}{\put(32,8){$j$}\includegraphics[width=2cm]{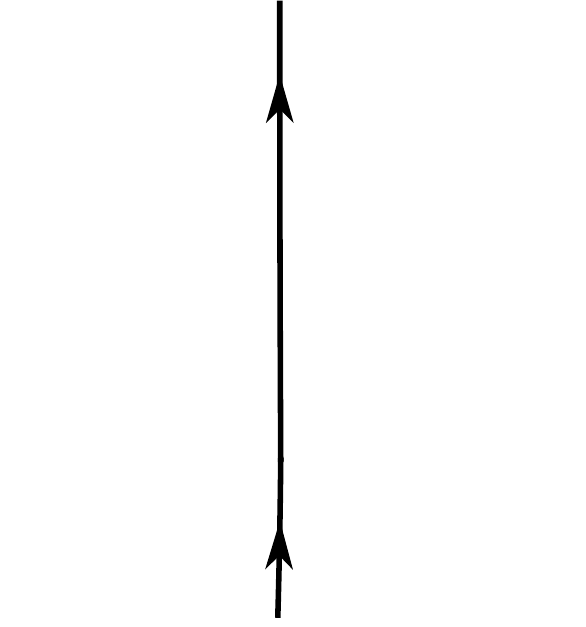}}.
\end{equation}
These equations are unchanged if one flips the orientation of a half-edge on both sides.

Notice that we have singled out the sign $(-1)^{2j_1}$ on purpose in the right hand side of \eqref{eq:whitehead} and \eqref{eq:whiteheadgraphical}, as this is a manifestation of the reason why we call our definition \eqref{eq:tensorsn} of spin network evaluations ``unitary up to a sign''.

The standard definition of spin network evaluations, which we call the \emph{integral normalization} and write $\langle \G,\{j_e\}\rangle^{Int}$, is due to Penrose \cite{Pe} and is based on associating to each colored graph a linear combination of union of curves lying in $S(\G)$ and then evaluating each union of $k$ curves to $(-2)^k$. Some renormalizations of Penrose's definition known as \emph{skein theoretical} evaluation and \emph{unitary} evaluation also exist. As we will prove in Theorem \ref{teo:comparison}, the relation of the above ``unitary up to a sign'' evaluation with the integral spin network evaluation is the following, \emph{if $o$ is a Kasteleyn orientation on $\G$}: 
\begin{equation}\label{eq:integralsn}
\langle \Gamma,\{j_e\}\rangle^{Int} =s(\G,\{j_e\},o)  \sqrt{\frac{\prod_v(J_v+1)!}{\prod_{ev} (J_v-2j_e)!}}=\langle \Gamma,\{j_e\}\rangle^{skein} \frac{\prod_{\alpha} j_\alpha!}{\prod_e (2j_e)!}.
\end{equation}
$J_v$ denotes the sum of the three spins incident to $v$. Furthermore, the unitary evaluation is such that the value of any theta graph is $1$; the normalization we use turns out to coincide with the unitary normalization up to a sign which, if $o$ is a Kasteleyn orientation on $\G$ is computed explicitly:
\begin{equation*} 
s(\G,\{j_e\},o)=\langle \Gamma,\{j_e\}\rangle^U (-1)^{\frac{1}{2}\sum_v J_v}.
\end{equation*}

Our main object of interest in the present article is the generating function of integral evaluations of spin networks with edge variables $\{Y_e\}_{e\in\Gamma}$. Provided $o$ is a Kasteleyn orientation so that \eqref{eq:integralsn} holds, it can be written
\be
Z^{Spin}(\Gamma,\{Y_e\})=\sum_{\{j_e\}} \langle \Gamma,\{j_e\}\rangle^{Int}  \prod_e Y_e^{2j_e}
=
\sum_{\{j_e\}} \sqrt{\frac{\prod_v (J_v+1)!}{\prod_{ev} (J_v-2j_e)!}} s(\{j_e\},o) \prod_e Y_e^{2j_e}\,.
\ee

%%%
\subsection{About the Tensorial Definition of Unitary Spin Networks}
%%%

Observe that if $o$ and $o'$ are orientations on $\G$ and $o'$ differs from $o$ on a single edge $e$ then $s(\G,\{j_e\},o')=s(\G,\{j_e\},o)(-1)^{2j_e}$. In contrast, the standard spin network evaluation does not depend on any orientation of the edges of the graph. So the tensorial definition \eqref{eq:tensorsn} does not coincide in general with the standard one (or it does only up to a sign) if the orientation of $\Gamma$ is chosen randomly.

In \cite{CoMa} (section 2.1) a different tensorial definition was provided based on supersymmetric vector spaces (instead of even ones) which allowed for a definition independent on the choice of an orientation of the edges of $\Gamma$. 
Roughly speaking, one stipulates that the vectors of the representation with spin $j_e$ have parity $2j_e$. 
Then to each vertex one associates the (even) tensor: $i_{j_1 j_2 j_3}$ (as defined above) and to each edge colored by $j_e$ the (even) tensor 
$\omega_e:= \sum_{m_e=-j_e}^{j_e} (-1)^{j_e-m_e}\delta^{j_e}_{m_e}\otimes \delta^{j_e}_{-m_e}$ 
(where $\delta^{j_e}_{m_e}(e^{j_e}_k)=1$ if $k=m_e$ and $0$ else). Observe that because of the parity assumption in truth, in order to define $\omega_e$ we did not use an orientation: indeed applying the flip (exchanging the two tensors in $\omega$) gives $\omega'_e=\sum_{m_e}(-1)^{2j_e} (-1)^{j_e-m_e}\delta^{j_e}_{-m_e}\otimes \delta^{j_e}_{m_e}=\sum_{m'_e} (-1)^{2j_e+2m'_e+j_e-m'_e}\delta^{j_e}_{m'_e}\otimes \delta^{j_e}_{-m'_e}=\omega_e$ (where in the second equality we just reindexed the sum and observed that $2j_e-2m'_e$ is even). 
Then one defines the spin network evaluation as the total supersymmetric contraction of the tensor $\bigotimes_e \omega_{j_e}\otimes \bigotimes_v i_{j_1^v  j_2^v  j_3^v}$ where, again, in the contraction one should take care of introducing signs $(-1)^{ab}$ whenever permuting tensors of degrees $a$ and $b$. In particular, observe that $i_{j_1 j_2 j_3}=(-1)^{2j_1(2j_2+2j_3)}i_{j_2 j_3 j_1}$ but this apparent asymmetry is compensated during the contractions.

In our case we deal with the orientation problem by choosing a Kasteleyn orientation on $\G$ (which is supposed to be planar).  
Observe that by the admissibility condition of the spins (i.e. the fact that for each vertex $v$ the sum $J_v$ of all the spins of the edges surrounding $v$ is integer) and by the above observation on the behavior under the switch of the orientation of an edge, formula \eqref{eq:tensorsn} provides a well defined function on the space of orientation classes which are defined as follows:

\begin{defi}[Orientation class]\label{def:orient}
Two orientations on $\G$ are \emph{equivalent} if they can be obtained from one another by a finite sequence of moves consisting in switching all the orientations of the edges incident to a vertex. We will denote the equivalence class of an orientation $o$ by $[o]$. 
\end{defi}

In our case $\G$ is supposed to be planar and we can equip it with a Kasteleyn orientation; there may be more than one such orientation but it is easy to see that all of them are equivalent in the sense of Definition \ref{def:orient} hence the spin network evaluation on $\G$ is well defined. %Furthermore as proved by Cimasoni and Reshetikhin\nota{citer}, since $\G$ contains an even number of vertices, a Kasteleyn orientation on $\G$ does exist. 
%Furthermore, the set of orientation classes is naturally equipped with the structure of affine space over $H^1(\G;\Z_2)$. 
%Thanks to the results of Cimasoni-Reshetikhin (\cite{CR}), if one fixes a dimer configuration on $\G$, an orientation class is equivalent to the datum of a spin structure on the surface $S(\G)$ obtained by thickening $\G$ in the unique way which is coherent with the cyclic orientations around the vertices of $\G$. 

\begin{theo}\label{teo:comparison}
Let $\G\subset \R^2$ be a planar connected trivalent graph and $o$ be an orientation on $\G$. Then, the following holds for each coloring $\{j_e\}$ of the edges of $\G$ if and only if $o$ is a Kasteleyn orientation:
\begin{equation*}
\langle \Gamma,\{j_e\}\rangle^{Int} = \sqrt{\frac{\prod_v(J_v+1)!}{\prod_{ev} (J_v-2j_e)!}}\ s(\G,\{j_e\},o).
\end{equation*}
%\begin{enumerate}
%\item There is a $H_1(\Sigma,\partial \Sigma;\Z_2)$-affine bijection between the equivalence classes of orientations of $\G$ and the set ${\mathcal Spin}(S(\G))$ of spin-structures on $S(\G)$. The bijection depends on the choice of the dimer configuration.
%\item For each coloring $\{j_e\}$ and for each orientation $o$ on $\G$ the formula $s(\{j_e\}, o)$  given in Equation \eqref{eq:tensorsn} is a well defined function $s:{\mathcal Spin}(S(\G))\to \Q$ (the function depends on the preceding bijection and hence on the chosen dimer configuration). 
%\item If $\Gamma$ is planar and $s_0\in {\mathcal Spin}(S(\G))$ is the spin structure induced by the embedding of $\G$ in $\R^2$ then for each coloring $\{j_e\}$ on $\G$ it holds $s(\{j_e\})=\langle \G,\{j_e\}\rangle^{Int}\sqrt{\frac{\prod_{ev} (J_v-2j_e)!}{\prod_v (J_v+1)!}}$, i.e. up to a normalization the evaluation coincides with the standard evaluation of the unoriented graph $\G$.  
%\end{enumerate}
\end{theo}
\begin{proof}
We will first prove that among all the orientation classes on $\G$ there is at most one, the Kasteleyn orientation class, which can have the required property. The idea is to compare Equation \eqref{eq:integralsn} with the standard evaluations only on ``curve colorings'' i.e. colorings with values in $\{0,\frac{1}{2}\}$; these colorings are naturally identified with the set of curves embedded in $\G$ (the curve associated to a coloring being the set of edges with spin $\frac{1}{2}$). Recall that the standard evaluation of such a coloring is $(-2)^k$, where $k$ is the number of connected components of the curve associated to the coloring. 

The proof is actually similar to the (constructive) proof of the existence of Kasteleyn orientations. Pick up a spanning tree $T\subset \G$, whose orientation is the restriction of $o$ to the edges of $T$. Then we show that requiring the curve coloring evaluation to be $-2$ around each face of $\G$ enforces an orientation on the edges of $\G\setminus T$, and it is such that one gets a Kasteleyn orientation on $\G$.

We denote $\G^*$ the dual graph to $\G$, whose vertices represent the faces of $\G$ and an edge connects two vertices dual to the faces $f_1, f_2$ whenever there is an edge of $\G$ which belongs to the boundary of both faces. Therefore an edge of $\G^*$ uniquely identifies an edge in $\G$ and the other way around. We denote $(\G\setminus T)^*$ the dual to $\G\setminus T$. It is a tree: it is connected because $T$ has no cycle and it has no cycle since $T$ is connected. Since it goes through every face of $\G$, $(\G\setminus T)^*$ is a spanning tree of $\G^*$. We choose a root vertex (for instance dual to the outerface) and consider a leaf $f^*$, i.e. a vertex of degree 1. It is dual to a face $f$ whose boundary edges all belong in $T$ but one, say $e_f$ (which is dual to the edge incident to the leaf). We denote the boundary edges of $f$ by 
\begin{equation*}
\partial f = (\partial f \cap T) \cup e_f,
\end{equation*}
which corresponds to the partition of the edges $\G = (\G\cap T)\cup (\G\setminus T)$ restricted to $\partial f$.
We want to check that imposing $s(\partial f, o)=-2$ implies for $e_f$ an orientation such that there is an odd number of clockwise edges in $\partial f$. By switching the orientations as in Definition \ref{def:orient}, we may assume that all edges of $\partial f\cap T$ are counter-clockwise around $f$. 

Let us start observing that 
$$i_{\frac{1}{2}\frac{1}{2}0}=\frac{1}{\sqrt{2}}(e^\frac{1}{2}_\frac{1}{2}\otimes e^\frac{1}{2}_{-\frac{1}{2}}\otimes e^0_0-e^\frac{1}{2}_{-\frac{1}{2}}\otimes e^\frac{1}{2}_{\frac{1}{2}}\otimes e^0_0), \ {\rm and}\ i_{\frac{1}{2}0\frac{1}{2}}=\frac{-1}{\sqrt{2}}(e^\frac{1}{2}_\frac{1}{2}\otimes e^0_0\otimes e^\frac{1}{2}_{-\frac{1}{2}}-e^\frac{1}{2}_{-\frac{1}{2}}\otimes e^0_0\otimes  e^\frac{1}{2}_{\frac{1}{2}});$$ 
then if $n$ is the number of vertices (or edges) contained in $\partial f$ then the evaluation is
\begin{equation*} 
s(\partial f,o) = \pm \left(\frac{1}{\sqrt{2}}\right) ^{n}(1+1).
\end{equation*}
Indeed there are exactly two non-zero summands in Formula \eqref{eq:tensorsn} and they can easily be checked to have the same sign. Since the normalization factor $\sqrt{\frac{\prod_v(J_v+1)!}{\prod_{ev} (J_v-2j_e)!}}$ equals $\sqrt{2}^{n}$, it is clear that we get $-2$ exactly for one of the two possible orientations on $e_f$ and $2$ for the other. Then a direct inspection in Formula \eqref{eq:tensorsn} and in the tensors $i_{\frac{1}{2}\frac{1}{2}0}$ and $i_{\frac{1}{2}0\frac{1}{2}}$ shows that 
\begin{equation*} 
s(\partial f,o)=2\,\sign(e_f)\,(-1)^{a(\partial f)}
\end{equation*}
where $\sign(e_f)$ is $1$ if $e_f$ is oriented positively by $\partial f$ and $-1$ else and $a(\partial f)$ is the number of ``large angles'' in $\partial f$ as defined in Lemma \ref{lemma:KasteleynCycle}. By Lemma \ref{lemma:KasteleynCycle}, since $\partial f$ contains no internal vertices, we have $(-1)^{a(\partial f)}=1$, hence $s(\partial f,o)=-2$ if and only if $e_f$ is oriented negatively with respect to $f$. Switching back the orientations on $\partial f\cap T$ to $o$ in the sense of Definition \ref{def:orient}, we preserve an odd number of clockwise edges, as claimed. A straightforward induction from the leaves of $(\G\setminus T)^*$ to its root concludes this part: only the Kasteleyn class can provide the standard evaluation for all colorings of $\G$. 

To prove that it actually does, for every coloring $\{j_e\}$ on $\G$, we start by observing that the claim is true if $\G$ is the theta-graph. Moreover if $\G$ is generic and equipped with a Kasteleyn orientation $o$ and a coloring $\{j_e\}$, then we can transform it into a graph which is formed by a connected sum of theta-graphs. This is done by means of a finite sequence of \emph{Whitehead moves} similar to Equation \eqref{eq:whitehead}. To see this one argues by induction on the number of vertices of $\G$: choose a region of $\R^2\setminus \G$  with, say, $n$ edges and pick one of them up, say $e$. Applying $n-2$ Whitehead moves one can ``slide one endpoint of $e$'' along the boundary of the face until one gets a new graph containing two arcs (one of which is the image of $e$ under the sequence of moves) having the same endpoints. Then the so-obtained graph is a connected sum of a theta-graph with a ``simpler'' one. Figure \ref{fig:sequencewhitehead} is an exemplification of the case $n=4$ of this process.
\begin{figure}
$\raisebox{-1.5cm}{\put(-10,50){$$}\includegraphics[width=3cm]{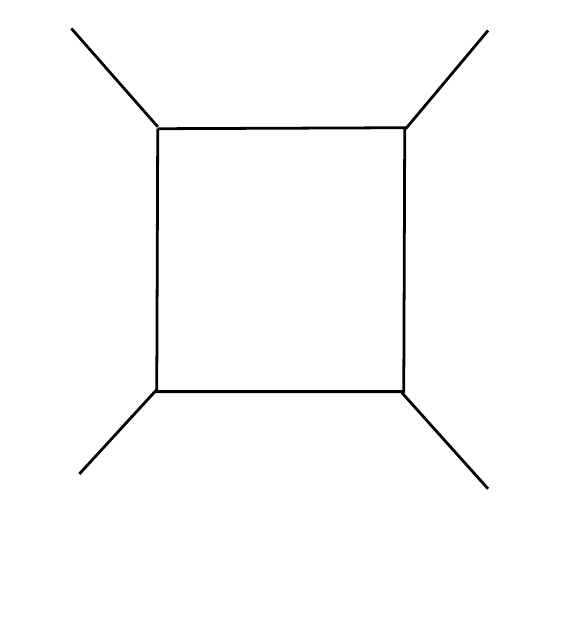}}\to \raisebox{-1.5cm}{\put(-10,50){$$}\includegraphics[width=3cm]{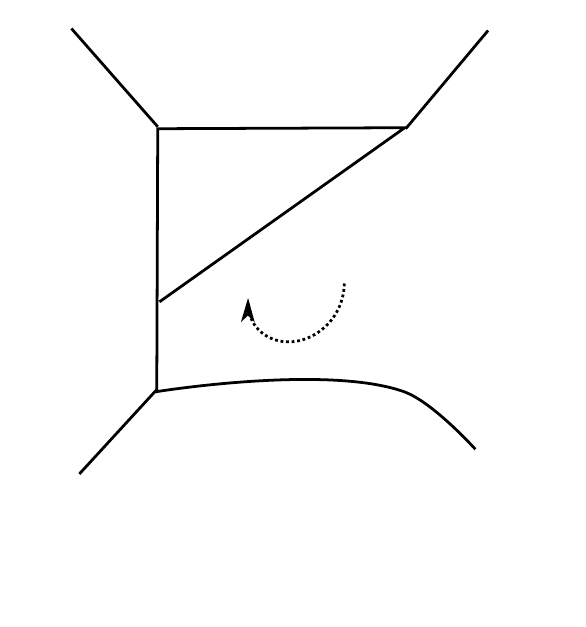}}\to \raisebox{-1.5cm}{\put(-10,50){$$}\includegraphics[width=3cm]{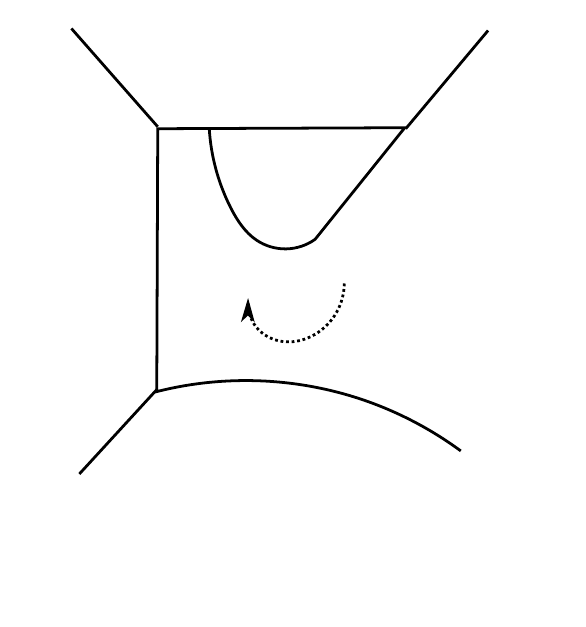}}$
\caption{An example of how to reduce $\G$ to a connected sum of theta-graphs ($n=4$).} \label{fig:sequencewhitehead}
\end{figure}

% (if the orientation of an edge $e$ is not as depicted, we can switch it and multiply the evaluations by $(-1)^{2j_e}$ to reduce to the case of the figure)

The algebraic translation of a Whitehead move in terms of spin networks, for our unitary up to a sign definition, is precisely Equation \eqref{eq:whiteheadgraphical}. Observe that indeed, one can always orient the edge carrying $j_{12}$ as in the left hand side (l.h.s.) of \eqref{eq:whiteheadgraphical}, by switching the orientations of the edges incident to one of its endpoints, if necessary. However, the Whitehead moves does not preserve the Kasteleyn orientation: the drawing on the right hand side (r.h.s.) of \eqref{eq:whiteheadgraphical} is not equipped with a Kasteleyn orientation anymore. The face going along $(j, j_{12}, j_1)$ on the l.h.s. loses a clockwise edge in the process, while the face along $(j_1, j_2)$ gains one. Therefore this can be corrected by switching the orientation of the edge carrying $j_1$ on the r.h.s. This brings up an extra factor $(-1)^{2j_1}$ which cancels the one we already had in \eqref{eq:whiteheadgraphical}. The Whitehead move respecting the Kasteleyn orientation class is thus
\begin{equation}
\left.\raisebox{-1cm}{\put(-10,50){$j_1$}\put(24,50){$j_2$}\put(52,50){$j_3$}\put(32,8){$j$}\put(12,27){$j_{12}$}\includegraphics[width=2cm]{leftclebschgordan.pdf}}\right.%_{\text{Kasteleyn}}
= \sum_{j_{23}} (2j_{23}+1) (-1)^{j_1+j_2+j_{3}+j} \begin{Bmatrix} j_1 & j_2 &j_{12}\\ j_3 & j & j_{23} \end{Bmatrix} \left.\raisebox{-1cm}{\put(-7,50){$j_1$}\put(27,50){$j_2$}\put(52,50){$j_3$}\put(32,8){$j$}\put(34,27){$j_{23}$}\includegraphics[width=2cm]{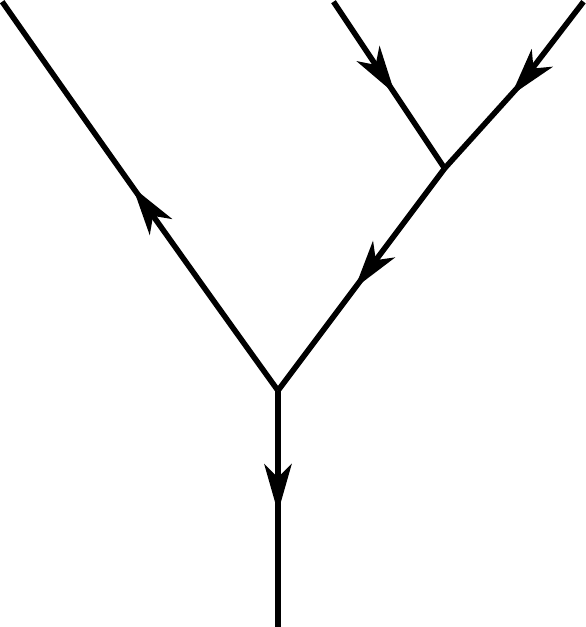}}\right.%_{\text{Kasteleyn}}
\end{equation}
In this equation, the orientation of the half-edges can be flipped, as long as it is done on both sides.

Similarly observe that if $o$ in the left hand side of Equation \eqref{eq:unzip} is a Kasteleyn orientation then $o'$, on the r.h.s. is not. To correct this, we need to switch the orientation of the only edge (colored by $j$) thus introducing an additional factor $(-1)^{2j}$.

Now the crucial observation is that the standard integral evaluation satisfies the same equations, up to the factors $\Delta(a,b,c)=\frac{(a+b+c+1)!}{(a+b-c)!(a+c-b)!(b+c-a)!}$ which we have not taken into account yet. Indeed letting $\Delta(a,b,c)$ re-normalize $s(\G,\{j_e\},o)$ to the standard normalization of $\langle \G,\{j_e\}\rangle^{Int}$ we have to multiply the left hand side by $\sqrt{\Delta(j_1,j_2,j_{12})\Delta(j_3,j,j_{12})}$ and the r.h.s. by $\sqrt{\Delta(j_1,j,j_{23})\Delta(j_2,j_3,j_{23})}$. But, taking into account Racah's formula for Wigner 6j-symbols and the fact that the standard evaluation of a theta-graph colored say by $j_1,j,j_{23}$ is $(-1)^{j_1+j+j_{23}}\Delta(j_1,j,j_{23})$ we see that the standard evaluation $\langle \G,\{j_e\}\rangle^{Int}$ satisfies the following recoupling identities:
\begin{equation}
\raisebox{-1cm}{\put(-10,50){$j_1$}\put(24,50){$j_2$}\put(52,50){$j_3$}\put(32,8){$j$}\put(12,27){$j_{12}$}\includegraphics[width=2cm]{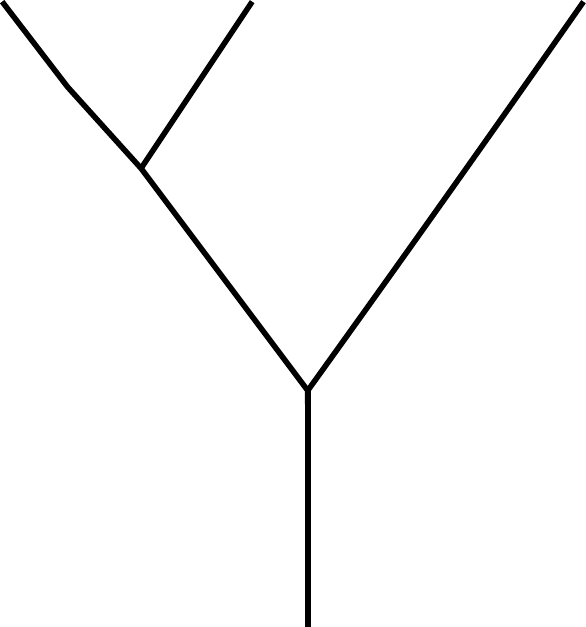}}=\sum_{j_{23}} (2j_{23}+1) (-1)^{j_1+j_2+j_{3}+j} \sqrt{\frac{\Delta(j_1,j_2,j_{12})\Delta(j_{12},j_3,j)}{\Delta(j_2,j_3,j_{23})\Delta(j_1,j_{23},j)}}\left\{\begin{matrix} j_1 & j_2 &j_{12}\\ j_3 & j & j_{23} \end{matrix}\right\}\raisebox{-1cm}{\put(-7,50){$j_1$}\put(27,50){$j_2$}\put(52,50){$j_3$}\put(32,8){$j$}\put(34,27){$j_{23}$}\includegraphics[width=2cm]{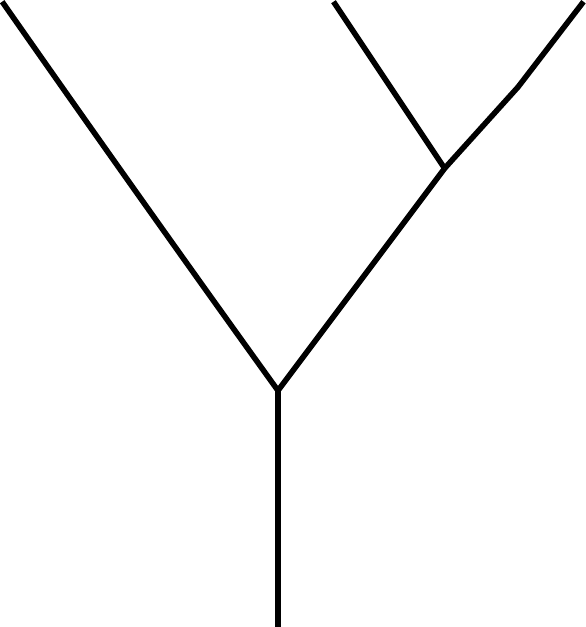}}
\end{equation}
\begin{equation}
\raisebox{-1cm}{\put(20,50){$j$}\put(30,6){$j'$}\put(8,27){$j_{2}$}\put(40,27){$j_{1}$}\includegraphics[width=2cm]{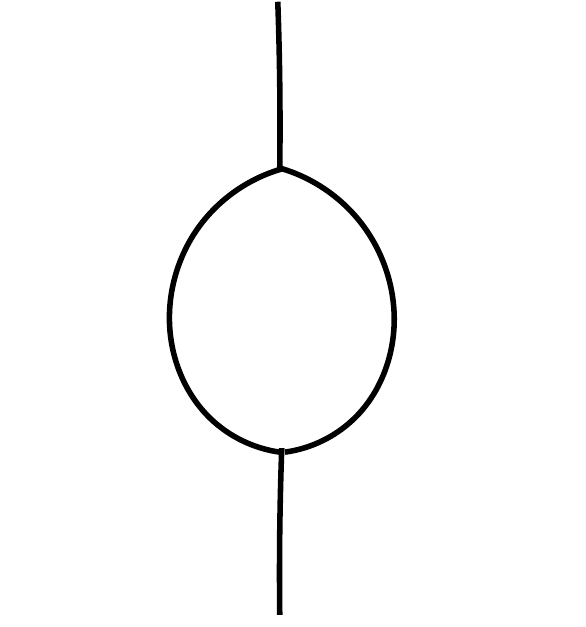}}=\frac{(-1)^{j_1+j_2+j}(-1)^{2j}\delta_{j,j'} \Delta(j,j_1,j_2)}{2j+1}\raisebox{-1cm}{\put(32,8){$j$}\includegraphics[width=2cm]{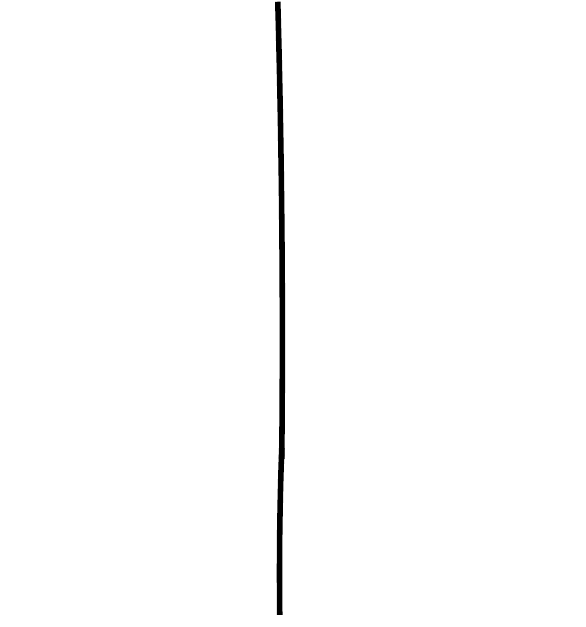}}
\end{equation}
 
Since the values of $\sqrt{\prod_v(J_v+1)!/\prod_{ev} (J_v-2j_e)!} s(\G,\{j_e\},o)$ with $o$ a Kasteleyn orientation satisfy the same recoupling identities as the values of the standard evaluation $\langle \G,\{j_e\}\rangle^{Int}$, the theorem is proved.
\end{proof}

%%%
\subsection{Gaussian Representation of the Generating Function and its Loop Expansion}
%%%

Similarly to the Ising model, we will provide a representation of this generating function for spin network evaluations as a Gaussian integral, but this time in terms of usual even-Grassmannian integration variables.

Let $\G$ be a planar graph, $o$ a Kasteleyn orientation on $\G$ and 
\begin{equation}
\label{eq:genser}
Z^{Spin}(\Gamma,\{Y_e\})=\sum_{\{j_e\}} \sqrt{\frac{\prod_v (J_v+1)!}{\prod_{ev} (J_v-2j_e)!}} s(\{j_e\},o) \prod_e Y_e^{2j_e}=\sum_{\{j_e\}} \langle \G,\{j_e\}\rangle^{Int} \prod_e Y_e^{2j_e}.
\end{equation}
By Theorem \ref{teo:comparison} we can omit the notation $o$.
In this section we provide a Gaussian integral whose value is $Z^{Spin}(\Gamma,\{Y_e\})$.

A similar integral was provided first in \cite{CoMa}. The difference between the approach in the present paper and that in \cite{CoMa} is that here we use complex-valued integrals and we deal with a Gaussian integral which is convergent as it is, while in \cite{CoMa} a suitable regularization procedure was used to make sense of the integrals. On the downside the integral considered here is on a space whose dimension is twice that of the space used in \cite{CoMa}.
Another Gaussian integral expressions was given in \cite{laurent,bonzom}, but it involves integrating variables living on the vertices of $\G$ instead of its half-edges as done here. The relation between these two formulas and their equivalence will be examined later in Section \ref{vertexint}.

One of the interesting points of having such Gaussian integration formulas is that they allow to re-prove Westbury's theorem \cite{We} asserting that for a planar graph $\G$,
$$
Z^{Spin}(\G,\{Y_e\})=P_\G^{-2}
\qquad\qquad\textrm{where}\qquad
P_\G=\sum_{c \ {\rm curves\ in \ } \G} \,\,\,\prod_{e\in c} Y_e\,,
$$
where ``curve embedded in $\G$" means a (possibly empty) union of edges of $\G$ homeomorphic to a union of circles. It is worth noting that Westbury's initial proof was based on a totally different approach, namely computing spin networks via chromatic evaluations.
 
Let us first fix some notation. Recall that for an angle $\alpha\in A$ of a planar graph $\G$, we let $s(\alpha)$ and $t(\alpha)$ be the half edges at the source and target of $\alpha$ (using the cyclic ordering of the edges around the vertices) and $X_\alpha=\sqrt{Y_{s(\alpha)}Y_{t(\alpha)}}$. For each half-edge $h$ contained in an edge $e$ and incident to a vertex $v$, let $z_{ev},w_{ev}$ be complex variables and  $\overline{z}_{ev},\overline{w}_{ev}$ be their complex conjugates. 

\begin{theo} \label{teo:complexgaussian}
With the above notation it holds:
$$ Z^{Spin}(\Gamma,\{Y_e\})=\int_{(\mathbb{C}^2)^{2\#\Edges}} \prod_{ev} \f{dz_{ev}d\overline{z}_{ev}dw_{ev}d\overline{w}_{ev}}{\pi^{2}}\,
e^{-\sum_{ev}(|z_{ev}|^{2}+|w_{ev}|^{2})}\,
e^{-\sum_e (\bz_{s(e)} \bw_{t(e)} - \bw_{s(e)} \bz_{t(e)}) + \sum_\alpha X_\alpha (z_{s(\alpha)} w_{t(\alpha)} - w_{s(\alpha)} z_{t(\alpha)})}$$
where the r.h.s. is a well defined integral on $\C^{4\#E}$. Computing the value of the integral one also gets $Z^{Spin}(\G,\{Y_e\})=P_\G^{-2}$.
\end{theo}
\begin{proof}
Let us prove the first statement.
We start with the generating function of Wigner 3j-symbols, which is known \cite{Bargmann} to be given by:
\begin{equation}\label{eq:genserangle}
\sum_{\substack{j_1, j_2, j_3\\ m_1,m_2,m_3}} \begin{pmatrix} j_1 &j_2 &j_3\\ m_1 &m_2 &m_3 \end{pmatrix} \sqrt{(J+1)!} \prod_{i=1}^3 \frac{Y_i^{j_i} z_i^{j_i+m_i} w_i^{j_i-m_i}}{\sqrt{(J-2j_i)!(j_i-m_i)!(j_i+m_i)!}} 
= \exp \sum_\alpha X_\alpha (z_{s(\alpha)} w_{t(\alpha)} - w_{s(\alpha)} z_{t(\alpha)}).
\end{equation}
where $J=j_1+j_2+j_3$ and $\alpha$ denotes an angle. In the notation if $\alpha$ is the angle between the edges $i$ and $i+1$, then $z_{s(\alpha)} w_{t(\alpha)} - w_{s(\alpha)} z_{t(\alpha)} = z_{i} w_{i+1} - w_{i} z_{i+1}$, for $i=1,2,3$.
%Moreover, we have set $X_\alpha = \sqrt{Y_{s(\alpha)} Y_{t(\alpha)}}$.

\smallskip

Taking the product over all the vertices of $\G$ of the above generating sums we do not yet get the desired generating sum for two reasons:
\begin{enumerate}
\item\label{Point1} We need to consider only products where the spins and magnetic numbers of half-edges contained in the same edge match;
\item Even after solving point \ref{Point1}, the total weight of the state associated to a spin network with a choice of magnetic numbers per each edge would be off by a factor $\prod_{e} (-1)^{j_e-m_e} ((j_e-m_e)!(j_e+m_e)!)^{-1}$ (compare with Formula \eqref{eq:tensorsn}). 
\end{enumerate}

We solve both preceding problems by considering for each edge $e\in \Edges$ the following (well-defined) Gaussian integral:
\beq
&&\int_{\C^{4}} \f{dz_{t}d\overline{z}_tdz_{s}d\overline{z}_s dw_{t}d\overline{w}_t dw_{s}d\overline{w}_s}{\pi^{4}}\,
e^{-(|z_{t}|^{2}+|z_{s}|^{2}+|w_{t}|^{2}+|w_{s}|^{2})}
e^{-(\bw_{t} \bz_{s}-\bz_{t} \bw_{s})} z_{t}^{j_{t}+m_{t}} w_{s}^{j_{s}-m_{s}} z_{s}^{j_{s}+m_{s}} w_{t}^{j_{t}-m_{t}} \nn\\
&=&
\delta_{j_{t},j_{s}}\delta_{m_{t},-m_{s}}\,(-1)^{j_{t}-m_{t}}\,(j_{t}+m_{t})!(j_{t}-m_{t})!
\eeq

Thus taking the integral over all $\C^{4\#E}$ of the product over all the vertices of the generating series \eqref{eq:genserangle} and the product over all edges of the above Gaussian factors we get the desired result. 
Observe that the so-obtained Gaussian integral is convergent when $\{Y_e\}$ are near $0$.

\medskip

One could then try to switch back from complex variables to real variables, and avoid the doubling of the number of  integrals, by using the Segal-Bargmann transform, which maps complex holomorphic monomial $z^{n}$ onto the Hermite polynomials $H_{n}(x)$. This works and leads to a real Gaussian integral but finally produces a non-trivial and non-linear action in terms of the couplings $X_{\alpha}$. %One can see details in appendix \ref{SB} {\bf where did the appendix go?}. 

We now prove the last statement. The quadratic form associated to our complex Gaussian,
%$$
%Z^{spin}(\G,\{Y_e\})
%\,=\,
%\int_{(\mathbb{C}^2)^{2\#\Edges}} \prod_{ev} \f{dz_{ev}d\overline{z}_{ev}dw_{ev}d\overline{w}_{ev}}{\pi^{2}}\,
%e^{-\sum_{ev}(|z_{ev}|^{2}+|w_{ev}|^{2})}\,
%e^{-\sum_e (\bz_{s(e)} \bw_{t(e)} - \bw_{s(e)} \bz_{t(e)}) + \sum_\alpha X_\alpha (z_{s(\alpha)} w_{t(\alpha)} - w_{s(\alpha)} z_{t(\alpha)})}\,,
%$$
 is given by a $(8\#\Edges\times8\#\Edges)$ square matrix. In terms of the real parts of the $z$ and $w$ variables and their imaginary parts it reads as:
$$
\cM^{\Gamma}
\,\equiv\,
\mat{c|c}{ -A^{\G}+B^{\G}+2\id & -i(A^{\G}+B^{\G}) \\ \hline -i(A^{\G}+B^{\G})& A^{\G}-B^{\G}+2\id}\,
$$
where $A^{\G}$ and $B^{\G}$ are the symmetric real matrices of size $4\#\Edges$ defined as follows:
\be
{^{t}}ZA^{\Gamma}Z=-2\sum_e (z_{s(e)} w_{t(e)} - w_{s(e)} z_{t(e)}),
\quad
{^{t}}ZB^{\Gamma}Z=-2\sum_\alpha X_\alpha (z_{s(\alpha)} w_{t(\alpha)} - w_{s(\alpha)} z_{t(\alpha)})\,.
\ee

Using the facts that $(A^{\Gamma})^{2}=\id$ and $A^{\Gamma}B^{\Gamma}=-B^{\Gamma}A^{\Gamma}$ (which are easily checked by the definitions of $A^\G$ and $B^\G$), and summing first $-i$ times the last columns to the first ones and then $i$-times the new last lines to the first ones we compute:
$$
\det \cM^{\Gamma}
\,=\,
\det \mat{c|c}{ 4\id & 2i(\id-B^{\G}) \\ \hline -2i(\id+A^{\G})& A^{\G}-B^{\G}+2\id}
\,=\,
4^{4\#\Edges}\det(\id + A^{\G}B^{\G})
\,=\,
2^{8\#\Edges}\det(A^{\G}+B^{\G})
%\,=\,
%2^{8\#\Edges}\det M^{\G}\,.
$$
where the second equality holds from the fact that the determinant of a $2\times 2$-block matrix $\left( \begin{smallmatrix} M_1 & M_2\\ M_3 & M_4\end{smallmatrix}\right)$ is $\det(M_4-M_3M_1^{-1}M_2)\det(M_1)$ if $M_1$ is invertible. 

To conclude we now follow the lines of the approach used in \cite{CoMa}: observe that in terms of the real parts of the $z$ variables and of the real parts of the $w$ variables both $A^\G$ and $B^\G$ are symmetric and of the form 
$A=\left( \begin{smallmatrix} 0 & A'^t\\ A' & 0\end{smallmatrix}\right)$ and $B=\left( \begin{smallmatrix} 0 & B'^t\\ B' & 0\end{smallmatrix}\right)$  with $A'^t=-A$ and $B'^t=-B$. Thus $\det(A^\G+B^\G)=\Pf(A'+B')^4$ where $\Pf$ denotes the Pfaffian. We claim that $\Pf(A'+B')=P_\G$. 
To interpret it we use Kasteleyn's approach \cite{Kas} by interpreting the matrix $A'+B'$ as the adjacency matrix of the oriented graph $\G'$ obtained by replacing each vertex of $\G$ with a triangle and orienting the new edges counter-clockwise, while keeping the edges coming from those of $\G$ oriented as in $\G$. It is straightforward to check that $\G'$ inherits then a Kasteleyn orientation.

Remark that the edges of $\G'$ are of two types: those corresponding to edges of $\G$, say of ``edge type'' and weighted $1$, and those corresponding to angles of $\G$, say of  ``angle type'' and weighted $X_\alpha$. 
By Kasteleyn's results the value of the Pfaffian is the generating sum of the dimer configurations on $\G'$ weighted with the weights of the edges contained in the dimers.
Dimer configurations on $\G'$ correspond bijectively to curves in $\G$: consider the angle-type edges contained in the dimer configurations; they can be completed in a unique way to a curve in $\G'$ by adding some edges of edge type; then retracting the triangles of $\G'$ to points we get the desired curve. 
Reciprocally each curve in $\G$ meets a vertex either $0$ times or $2$ times: in the latter case it identifies an edge of angle type in $\G'$, in the former we associate to it the three  edges of edge type of $\G'$ surrounding the triangle corresponding to the vertex.

Since the orientation of $\G'$ is a Kasteleyn orientation, we conclude by the results of \cite{Kas} that $\Pf(A'+B')$ is the generating series of dimer configurations on $\G'$ weighted as above, which in turn equals $\sum_{c\subset \G} \prod_{\alpha\subset c} X_\alpha$. Recalling that $X_\alpha=\sqrt{Y_{s(\alpha)}Y_{t(\alpha)}}$, we are done. 
\end{proof}

This theorem implies the following remark, central to our present work, on the duality between the 2D Ising model and spin network evaluations:

\begin{rem}
Let $\G$ be a planar, connected and 3-valent graph. Then if we match the Ising couplings $y_{e}$ to the parameters of the spin network generating function by $Y_{e}=\tanh y_{e}$, the following equalities hold exactly:
\be
\label{point}
Z_{f}^{\C}(\G,\{X_{\alpha}\})^2\,
Z^{Spin}(\G,Y_{e})
\,=\,
1\,,
\qquad
Z^{Ising}(\G,\{y_e\})^2\, Z^{Spin}(\Gamma,\{Y_e\})
\,=\,
\left(2^{2\#\Ver}\prod_{e\in \Edges} \cosh(y_e)^2\right) 
\,,
\ee
where we keep the matching $X_{\alpha}=\sqrt{Y_{s(\alpha)}Y_{t(\alpha)}}$ between angle and edge variables ensuring that $\prod_{\alpha\in\cL}X_{\alpha}=\prod_{e\in\cL}Y_{e}$ around any loop $\cL$ in the graph $\G$.
\end{rem}

%%%%%%
\section{Duality through Supersymmetry}\label{subsec:supersymmetry}
%%%%%%

In order to understand the relation between the Ising model and spin network evaluations, one can play a little game. On the one hand, the Ising partition function on a graph $\Gamma$ has a simple loop expansion as $P_{\Gamma}=1+\tP_{\G}=1+\sum_{\gamma\subset\Gamma} \cW_{\gamma}$ (up to pre-factors) where we sum over all non-empty even subgraphs $\gamma$. These even subgraphs are simply identified as unions of disjoint loops when $\Gamma$ is 3-valent. The amplitudes $\cW_{\gamma}$ are the weights $\prod_{e\in\gamma} Y_{e}$. On the other hand, the generating function of the spin network evaluation is related to the squared inverse of $P_{\gamma}$, which can be expanded as a power series:
%\footnotemark
\be
P_{\Gamma}^{-2}
=
\f1{(1+\tP_{\G})^{2}}
=
\sum_{n\in\N} (-1)^{n}n\,\tP_{\G}^{n}
=
\sum_{n\in\N} (-1)^{n}n\,\sum_{\gamma_{1}..\gamma_{n}}  \cW_{\gamma_{1}}.. \cW_{\gamma_{n}} .
\ee
%
%\footnotetext{
%Expanding the square would be produce a very different answer:
%$P_{\Gamma}^{-1}
%=
%{(1+\tP_{\G})}^{-1}
%=
%\sum_{n\in\N} (-1)^{n}n\tP_{\G}^{n}
%$.}
%
Such union of even subgraphs can be considered as an arbitrary subgraph $G$ with extra integer labels $c_{e}$ on each edge indicating how many times an edge $e$ belongs to one of the even subgraphs $\gamma_{i}$. We have the obvious constraint that the sum of $c_{e}$'s around each vertex $v$ is even, plus the less obvious constraints that they must satisfy the triangular inequalities.  These colors $c_{e}$ are actually to  be thought of as twice the spins $j_{e}$.
%(although we are only considering $P_{\Gamma}^{-1}$ and not its square $P_{\Gamma}^{-2}$, but this would be only affect the combinatorial factors and nothing more).
This allows to write the power series above as:
\be
P_{\Gamma}^{-2}
=
1+\sum_{G\subset\Gamma}\sum_{\{c_{e}\}_{e\in G}} \cC_{G}(c_{e})\,\cW_{G}^{\{c_{e}\}}\,, 
\qquad
\cW_{G}^{\{c_{e}\}}
\,=\,
\prod_{e\in G}Y_{e}^{c_{e}}
\,,
\ee
where the coefficients $\cC_{G}(c_{e})$ count, up to the pre-factor $(-1)^{n}n$, the number of ways we can decompose the subgraph $G$ with admissible coloring $\{c_{e}\}$ as a union of even subgraphs $\gamma_{i}$. This is actually the expansion in terms of spin network evaluations. But it explains, in some sense, how while the Ising model sums over configurations with the occupation number on its edge being 0 or 1, and can be thought as fermionic, the inverse of its partition function sums over configurations with arbitrary  integral occupation numbers (with some non-trivial combinatorial factors, which are actually the spin network evaluations) and can be thought as bosonic. We formulate in a rigorous way below defining a supersymmetry relating the Ising partition function and the generating function for spin network evaluations.

%%%
\subsection{A Supersymmetric Theory}\label{sec:supersymmetry}
%%%

We showed in Proposition \ref{prop:GrassmannIsing} that the Ising partition function (on the graph $\Gamma$) can be formulated as a Gaussian fermionic  integral. It directly evaluates the determinant of the quadratic form, which is simply the polynomial $P_\Gamma$. On the other side, the generating function of spin network evaluations is, by Westbury theorem,  the inverse squared  $1/P_\Gamma^2$. It is thus  given by a bosonic Gaussian integral to get the inverse determinant, together with a doubling of the variables and an anti-symmetrization, in order to reach the correct power of the determinant. This is exactly what we have achieved in the previous section.

Let us now start with the following ``meta-theory'', with both fermionic and bosonic degrees of freedom, with a copy of $\R^2\oplus (\bigwedge\R)^2$ on each half-edge,
$\cZ(\Gamma,\{Y_e\}) \equiv Z_{f}(\G,\{y_e\})^2\, Z^{\C}(\G,Y_{e})$:
\bes
\cZ(\Gamma,\{Y_e\}) &= \int_{\mathcal{S}} \prod_{\text{half-edges $h$}} dz_h dw_h d\psi_h d\eta_h\,&
\exp\left( \sum_e (z_{s(e)} w_{t(e)} - w_{s(e)} z_{t(e)}) + \psi_{s(e)}\psi_{t(e)} + \eta_{s(e)}\eta_{t(e)} \right)\nn\\
&&\exp \left(\sum_\alpha X_\alpha \Bigl(z_{s(\alpha)} w_{t(\alpha)} - w_{s(\alpha)} z_{t(\alpha)} + \psi_{s(\alpha)}\psi_{t(\alpha)} + \eta_{s(\alpha)}\eta_{t(\alpha)}\Bigr)\right)
\,,
\ees
with $X_{\alpha}=\sqrt{Y_{s(\alpha)}Y_{t(\alpha)}}$ as before.
The variables $z,w$ are bosonic and are the spin network degrees of freedom, while the odd-Grassmaniann variables $\eta,\psi$ are fermionic and are the Ising degrees of freedom. Obviously the two families of integral factorize. One recognizes the Grassmannian integral of Proposition \ref{prop:GrassmannIsing} doubled (integration over $\psi$ and $\eta$). As for the bosonic part, it is a real version of the formula of Theorem \ref{teo:complexgaussian}, which is however not well-defined (its quadratic form has negative eigenvalues). This Gaussian integral should in principle be constant, thus leading to the result \eqref{point}. Although the above integral is ill-defined, it is a good starting point to define the supersymmetry transformations and later extend them to the well-defined complexified action.

Notice that the argument of the exponential splits into contributions labeled by pairs of adjacent half-edges $i,j$, either on the same edge $e$ or on the same angle $\alpha$; in the former case we suppose $i=s(e), j=t(e)$ and in the latter $i=s(\alpha),j=t(\alpha)$.  The total action  thus reads $S = \sum_{<i,j>} X_{ij} S_{ij}$ where $<i,j>$ means that the half-edges $i$ and $j$ are adjacent, and
\begin{equation}
S_{ij} = (z_i w_j- w_i z_j) + \psi_i \psi_j + \eta_i \eta_j,
\end{equation}
with the couplings $X_{ij}=1$ on an edge and $X_{ij}=X_\alpha$ as before for an angle. 

We define the supersymmetric operator $Q$ of odd parity (i.e. an odd derivation) acting on every half-edge $i$ on $\mathcal{S}_{i} = \R^2_i\oplus (\bigwedge\R)^2_i=\R^{2|2}_i$  as
\begin{equation}
Q z_i = \psi_i,  \quad Q w_i = \eta_i,\quad Q\psi_i = w_i, \quad  Q \eta_i = -z_i.
\end{equation}
We naturally extend the action of $Q$ to the algebra of functions on $\oplus_i\mathcal{S}_i$ via the graded Leibniz rule. It then follows that for any adjacent half-edge $j$, we have:
\begin{equation}
Q S_{ij} = \psi_i w_j + z_i \eta_j - \psi_j w_i - z_j \eta_i + w_i \psi_j - \psi_i w_j - z_i \eta_j + \eta_i z_j
=0,
\end{equation}
thus the action $S = \sum_{<h,g>} X_{hg} S_{hg}$ is $Q$-closed for each half-edge $i$.

Moreover, the action is actually $Q$-exact (but $Q^{2}$ does not vanish),
\begin{equation}
S_{ij} = Q \Phi_{ij} = -Q \Phi_{ji},
\qquad \text{with}\quad
%Q_{ij} =Q^{(i)}+Q^{(j)}
%\quad \text{and}\quad
\Phi_{ij} = z_i\psi_j + w_i \eta_j.
\end{equation}
Viewing the integral $\cZ(\Gamma)=\int \exp(\sum_{<h,g>} X_{hg} S_{hg})$ as a function of the couplings $\{X_{ij}\}$ associated to pairs of adjacent half-edges, we write
\begin{equation}
\frac{\partial\mathcal{I}(\Gamma)}{\partial X_{ij}} = \int_{\mathcal{S}} \prod_h dz_h dw_h d\psi_h d\eta_h\ \exp\Bigl(\sum_{<h,g>} X_{hg} S_{hg} \Bigr)\ S_{ij}
= \int_{\mathcal{S}} \prod_h dz_h dw_h d\psi_h d\eta_h\ Q \biggl( \Phi_{ij}\,\exp\Bigl(\sum_{<h,g>} X_{hg} S_{hg} \Bigr) \biggr).
\end{equation}
Using the lemma \ref{Qexact}, proven below, stating that the integral of any $Q$-exact function vanishes, we can deduce that 
\begin{equation}
\frac{\partial\mathcal{Z}(\Gamma)}{\partial X_{ij}} = 0\,,
\end{equation}
This proves without explicit evaluation that the integral $\mathcal{Z}(\Gamma)$ is a constant. Since it is independent of the variables $\{X_{ij}\}$, one can then evaluate it on any set of values making the integral simple. 

\begin{lemma}
\label{Qexact}
Let $f\in \C^{\infty}(\R^{2|2}_i)$ be a function admitting a Berezin integral. 
Then it holds:
\begin{equation}
\int_{\mathcal{S}}dz_h dw_h d\psi_h d\eta_h\ Q f = 0.
\end{equation}
\end{lemma}

\begin{proof}
Let $\epsilon$ be an odd-Grassmannian  variable, $\epsilon^2=0$ and perform the parity-preserving change of variables: $z\mapsto (1+\epsilon Q) z$, and similarly on $w, \psi, \eta$. 
Ordering the variables as $z,w,\psi,\eta$, the Jacobian matrix of the transformation reads as:
$$
J=\mat{cc|cc}{1 & 0 & 0 & -\eps \\ 0 & 1 & \eps & 0 \\ \hline \eps & 0 & 1 & 0 \\ 0 & \eps & 0 & 1}\,.
$$
Its Berezinian (or super determinant) is simply evaluated
$$
\textrm{Sdet} J
\,=\,
\det\left[
\id- \mat{cc}{0 & -\eps \\ \eps & 0}\mat{cc}{\eps & 0 \\ 0 & \eps}
\right]
\,=\,
1\,.
$$
Applying this change of variables on  the integral of $f$ gives:
\begin{equation}
\int_{\mathcal{S}} \prod_h dz_h dw_h d\psi_h d\eta_h\ f = \int_{\mathcal{S}} \prod_h dz_h dw_h d\psi_h d\eta_h\ \left(f + \sum_{h} \epsilon (Qz_h \partial_{z_h}f + Qw_h \partial_{w_h}f + Q\psi_h \partial_{\psi_h}f + Q\eta_h \partial_{\eta_h}f) \right).
\end{equation}
To get this equation, one has to be careful with the parities of $\epsilon, Qz, Qw$ and similarly that the operators $\partial_\psi, \partial_\eta$ are graded. Then observing that
\begin{equation}
Qz_h \partial_{z_h}f + Qw_h \partial_{w_h}f + Q\psi_h \partial_{\psi_h}f + Q\eta_h \partial_{\eta_h}f = Q(f),
\end{equation}
the desired result follows.
\end{proof}

We now would like to follow this procedure for the well-defined complex Gaussian integral for the generating function of the spin network evaluation given in Theorem \ref{teo:complexgaussian}. To allow to match bosonic and fermionic variables, we also use the formulation of the Ising model in terms of complex fermions and we then use Formula \eqref{eq:complexising}.

Thus we consider the partition function:
\begin{multline}
\cZ^{\C}(\Gamma,\{Y_e\}) = \frac{1}{\pi^{4\#\Edges}}\int_{\mathcal{S}} \prod_{\text{half-edges $h$}}
dz_hd\overline{z}_h dw_h d\overline{w}_h d\psi_hd\overline{\psi}_h d\eta_hd\overline{\eta}_h\,
e^{ -\sum_{e,v} \bigl(|z_{e}^{v}|^{2}+|w_{e}^{v}|^{2}\bigr)-\sum_{e}\bigl(\bz_{s(e)} \bw_{t(e)} - \bw_{s(e)} \bz_{t(e)}\bigr)}\\
e^{  \sum_{e,v}\bigl(\psi_{e}^{v}\ceta_{e}^{v}+\bpsi_{e}^{v}\eta_{e}^{v}\bigr)- \sum_e \bigl(\bpsi_{s(e)}\bpsi_{t(e)} + \ceta_{s(e)}\ceta_{t(e)}\bigr)}
e^{\sum_\alpha X_\alpha \bigl(z_{s(\alpha)} w_{t(\alpha)} - w_{s(\alpha)} z_{t(\alpha)} + \psi_{s(\alpha)}\psi_{t(\alpha)} + \eta_{s(\alpha)}\eta_{t(\alpha)}\bigr)}
%= \text{Constant}
\,,
\end{multline}
We consider the same supersymmetry generator $Q$ as before and we extend its action on the complex variables, both bosonic and fermionic by assuming its compatibility with the complex conjugation, $Q\bar{f}=\overline{Qf}$.
Then we decompose the complex path integral distinguishing the half-edge amplitude, the edge action and the angle action:
\be\label{eq:supersymmetry}
\cZ^{\C}(\Gamma,\{\lambda_{e,v},\mu_{e},Y_e\}) =  \frac{1}{\pi^{4\#\Edges}}\int_{\mathcal{S}} \prod_{\text{half-edges $h$}}
dz_hd\overline{z}_h dw_h d\overline{w}_h d\psi_hd\overline{\psi}_h d\eta_hd\overline{\eta}_h\,
e^{-\sum_{e,v}\lambda_{e,v}K_{e,v}}e^{-\sum_{e}\mu_{e}S_{e}}e^{\sum_{\alpha}X_{\alpha}S_{\alpha}}\,,
\ee
where we have added the extra couplings $\lambda_{e,v}$ and $\mu_{e}$, which can be set to 1, and
\bes
&&K_{e,v}\,\equiv\,
|z_{e}^{v}|^{2}+|w_{e}^{v}|^{2}-\psi_{e}^{v}\ceta_{e}^{v}-\bpsi_{e}^{v}\eta_{e}^{v},\\
&&S_{e}\,\equiv\,
\bz_{s(e)} \bw_{t(e)} - \bw_{s(e)} \bz_{t(e)}+\bpsi_{s(e)}\bpsi_{t(e)} + \ceta_{s(e)}\ceta_{t(e)},\nn\\
&&S_{\alpha}\,\equiv\,
z_{s(\alpha)} w_{t(\alpha)} - w_{s(\alpha)} z_{t(\alpha)} + \psi_{s(\alpha)}\psi_{t(\alpha)} + \eta_{s(\alpha)}\eta_{t(\alpha)}\,.\nn
\ees
A quick calculation shows that all three types of terms are invariant under supersymmetry. They are $Q$-closed and $Q$ exact:
\be
QK_{e,v}=QS_{e}=QS_{\alpha}=0,
\qquad
\left|
\begin{array}{lll}
K_{e,v}&=Q \,\left(\psi_{e,v}\bw_{e,v}-\eta_{e,v}\bz_{e,v}\right)&\\
S_{e}&=Q\,\left(\overline{z}_{s(e)}\bpsi_{t(e)}+\overline{w}_{s(e)}\overline{\eta}_{t(e)}\right)&\\
S_{\alpha}&=Q\,\left(z_{s(\alpha)}\psi_{t(\alpha)}+w_{s(\alpha)}\eta_{t(\alpha)}\right)&
\end{array}
\right.
\ee 
Using Lemma \ref{Qexact}, the integral \eqref{eq:supersymmetry} does not depend on any of the coupling constants $\lambda$, $\mu$ or $X$:
$$
\f{\pp \cZ^{\C}(\Gamma)}{\pp \lambda_{e,v}}
=
\f{\pp \cZ^{\C}(\Gamma)}{\pp \mu_{e}}
=
\f{\pp \cZ^{\C}(\Gamma)}{\pp X_{\alpha}}
=0\,.
$$
For instance, we can set $\lambda=\mu=1$ as in the original integral and then set all $X_{\alpha}=0$ (or equivalently all $Y_{e}=0$). This shows that the path integral is equal to the case where the angle interactions are killed and we are left with a free system with decoupled edges. We can also send the couplings $\lambda_{e,v}$ to $+\infty$.  This is the idea of \emph{localization}: the integral localizes on $z=w=0$ and $\psi=\eta=0$ and has a trivial evaluation.

This opens an interesting direction for future investigation, on how this localization due to the supersymmetry constrains the value of the observables and correlation functions of both the Ising model and the spin network evaluation.

%%%%%%%%
\subsection{Towards Non-trivial Coupled Supersymmetric Theories}\label{sec:nontrivialcoupled}
%%%%%%%%

Of course, the above integrals representing the Ising partition function and the spin network generating function are Gaussian, and therefore localization is not really necessary since we know how to compute these integrals explicitly.  We can nevertheless push the logic further and propose non-linear extensions, beyond the quadratic action, which would still be supersymmetric and thus localizable. This provides a rather large class of non-Gaussian integrals, which we can now compute exactly using our new tools, coupling non-trivially the spin network evaluations and the 2D Ising models on a planar graph.

\smallskip

Since each terms of the action, $K_{e,v}$, $S_{e}$ and $S_{\alpha}$ are $Q$-closed and $Q$-exact, so are arbitrary powers of these terms. This means that we can add arbitrary powers of each of them to the action, keeping the integral invariant under the same supersymmetry as above and therefore localizable.
This leads to a whole class of explicitly computable integrals:
\be
\cZ^{\C}_{n,p,q}(\Gamma,\{\lambda_{e,v},\mu_{e},Y_e\})
\,=\,
\int_{\mathcal{S}}d\mu\,
e^{-\sum_{e,v}\lambda_{e,v}K_{e,v}-\sum_{e}\mu_{e}S_{e}+\sum_{\alpha}X_{\alpha}S_{\alpha}}
e^{-\sum_{e,v}\lambda_{e,v}^{(p)}K_{e,v}^p-\sum_{e}\mu_{e}^{(q)}S_{e}^q+\sum_{\alpha}X_{\alpha}^{(n)}S_{\alpha}^n}\,,
\ee
where $n,p,q$ are arbitrary integers, larger than or equal to 2, and the measure $d\mu$ is the product measure over all  half-edges $h$ of all the even and odd Grassmann variables $d\mu\,=\frac{1}{\pi^{4\#\Edges}}\prod_{h}dz_hd\overline{z}_h dw_h d\overline{w}_h d\psi_hd\overline{\psi}_h d\eta_hd\overline{\eta}_h$.

As before these new integrals are constant and do not depend on the specific couplings $\lambda_{e,v}^{(p)}$, $\mu_{e}^{(q)}$ and $X_{\alpha}^{(n)}$.
%
%Obviously, all powers of each term of the action are $Q$-closed. It means
%\begin{equation}
%\int_{\mathcal{S}} \prod_h dz_h dw_h d\psi_h d\eta_h\ \exp\Bigl(\sum_{<h,g>} X_{hg} (S_{hg})^{k_{hg}} \Bigr) = \text{Constant}.
%\end{equation}
%
In the original Gaussian integral without the new higher order terms, the half-edge terms $K_{e,v}$ were the Gaussian measure factors, the edge terms $S_{e}$ were the anti-holomorphic part of the action, while the angle terms $S_{\alpha}$ were the holomorphic part of the action. In a sense, we were taking the scalar product between the edge action and the angle action, thus gluing the angles with the edges using the Gaussian measure. In this context, it does not seem useful to add non-linear terms to all three part of the action and we focus on adding higher order terms to, say, the angle action, thus considering:
\be
\cZ^{\C}_{n}(\Gamma,\{\lambda_{e,v},\mu_{e},Y_e\})
\,=\,
\int_{\mathcal{S}}d\mu\,
e^{-\sum_{e,v}\lambda_{e,v}K_{e,v}-\sum_{e}\mu_{e}S_{e}+\sum_{\alpha}X_{\alpha}S_{\alpha}}
e^{\sum_{\alpha}X_{\alpha}^{(n)}S_{\alpha}^n}\,,
\ee
where we can compute the arbitrary power $S_{\alpha}^n$ for $n\ge 2$ taking into account that the $\psi$'s and $\eta$'s are odd-Grassmannians:
\beq
S_{\alpha}^{n}
&=&
\left(z_{s(\alpha)} w_{t(\alpha)} - w_{s(\alpha)} z_{t(\alpha)} + \psi_{s(\alpha)}\psi_{t(\alpha)} + \eta_{s(\alpha)}\eta_{t(\alpha)}\right)^{n}
\nn\\
&=&
\left(z_{s(\alpha)} w_{t(\alpha)} - w_{s(\alpha)} z_{t(\alpha)}\right)^{n}
+ n\left(z_{s(\alpha)} w_{t(\alpha)} - w_{s(\alpha)} z_{t(\alpha)}\right)^{n-1}(\psi_{s(\alpha)}\psi_{t(\alpha)} + \eta_{s(\alpha)}\eta_{t(\alpha)}) \nn\\
&&+\f{n(n-1)}2\left(z_{s(\alpha)} w_{t(\alpha)} - w_{s(\alpha)} z_{t(\alpha)}\right)^{n-2}(\psi_{s(\alpha)}\psi_{t(\alpha)}\eta_{s(\alpha)}\eta_{t(\alpha)})\,.
\eeq
The first term corrects the weights of the generating function of the spin network evaluation. This will modify its stationary points, as explored in Section \ref{coherent}, and thus the geometrical background. The  second term then produces a geometry-dependent coupling for the Ising models, while the third term creates a coupling between the two Ising models.

Let us focus on the quartic case, $n=2$. Then the  coupling between Ising models is not geometry-dependent. We can write this integral explicitly:
\beq
\cZ^{\C}_{2}(\Gamma,\{\lambda_{e,v},\mu_{e},Y_e\})
&\,=\,
\int_{\mathcal{S}}d\mu\,&
e^{-\sum_{e,v}\lambda_{e,v}K_{e,v}-\sum_{e}\mu_{e}S_{e}+\sum_{\alpha}X_{\alpha}S_{\alpha}}
%e^{\sum_{\alpha}X_{\alpha}^{(n)}S_{\alpha}^n}
\\
&& e^{\sum_{\alpha}X_{\alpha}^{(2)}\big{[}
\left(z_{s(\alpha)} w_{t(\alpha)} - w_{s(\alpha)} z_{t(\alpha)}\right)^{2}
+ 2\left(z_{s(\alpha)} w_{t(\alpha)} - w_{s(\alpha)} z_{t(\alpha)}\right)(\psi_{s(\alpha)}\psi_{t(\alpha)} + \eta_{s(\alpha)}\eta_{t(\alpha)}) 
+\psi_{s(\alpha)}\psi_{t(\alpha)}\eta_{s(\alpha)}\eta_{t(\alpha)}
\big{]}}
\,.\nn
\eeq
Despite the new coupling $X_{\alpha}^{(2)}$, we know that this integral still evaluates to 1 as before. It would be enlightening to investigate how the generating function for spin networks is modified and what would be the physical effect of the new couplings for the Ising models, which seem to allow us to study the Ising models with dynamical, geometry-dependent couplings.

%%%%%%
\section{The Interplay between Spin Networks and Ising Correlations}\label{sec:isingcorrelations}
%\section{Spin Network Functionals and Geometric Interpretation}
%%%%%%

In this section, we investigate the link induced by the duality developed above between the correlations of the 2D Ising model and the probability distribution defined by the spin network evaluations.

%%%
\subsection{Coherent Spin Network States, Geometric Interpretation and Criticality}
\label{state}
\label{stationary}
%%%

In the context of quantum gravity and quantum geometry, following the logic developed in \cite{bonzom}, coherent spin network states with good semi-classical properties can be interpreted and effectively seen as generating functions for spin networks in the spin basis. It is similar to the way coherent states for the harmonic oscillator define the exponential generating function for the eigenvectors of the number of quanta operator.

We introduce the wave-functions on the graph $\G$ depending on group elements $g_{e}\in\SU(2)$ living on the edges $e\in\Edges$:
\be\label{spin:coherent1}
\phi^{coh}_{\{X_{\alpha}\}}(g_{e})
\,\equiv\,
 \int_{\C^{4E}}\prod_{e,v}\f{e^{-\la z_{e}^{v}|z_{e}^{v}\ra\ }d^{4}z_{e}^{v}}{\pi^{2}}\,
e^{\sum_{e} \la z_{e}^{s}|g_{e}|z_{e}^{t}]}\,e^{\sum_{\alpha}X_{\alpha}[z_{s(\alpha)}|z_{t(\alpha)}\ra}\,,\ 
\ee
where the use the following convention for spinors $|z\ra$ and their duals $|z]$, as introduced in \cite{spinortwisted,spinor}:
$$
|z\ra =\mat{c}{z \\ w} \in \C^{2},\qquad
|z]=\mat{cc}{0 &1 \\-1 & 0}\mat{c}{\bz \\ \bw}=-\mat{c}{-\bw \\ \bz}\,\ {\rm and}  \ d^4z:=dzd\overline{z}dwd\overline{w}.
$$
This allows to encode both complex variables $z$ and $w$ in a single complex 2-vector $|z\ra$. The $\SU(2)$ group acts naturally on $\C^{2}$ via its fundamental representation, and the dual spinor transforms exactly the same way as the original one:
$$
g\,:\,|z\ra\,\mapsto\,
g \vartriangleright\,|z\ra
=g \,|z\ra
=\mat{cc}{\alpha & \beta \\ -\bar{\beta} &\bar{\alpha}}\,\mat{c}{z \\ w}\,,
\qquad
|g\,z]=g\,|z]\,.
$$
The scalar products between states and with their dual provide two bilinear forms invariant under the action of $\SU(2)$:
$$
\la z_{1}|z_{2}\ra = (\bz_{1}z_{2}-\bw_{1}w_{2})\,,
\qquad
[z_{1}|z_{2}\ra = (z_{1}w_{2}-z_{2}w_{1})\,.
$$
This implies that the wave-functions $\phi^{coh}_{\{X_{\alpha}\}}(g_{e})$ introduced above are gauge-invariant as expected under $\SU(2)$ transformations acting at the vertices of the graph:
$$
\forall \{h_{v}\}_{v\in\Ver}\in\SU(2)^{\#\Ver}\,,\quad
\phi^{coh}_{\{X_{\alpha}\}}(g_{e})\,=\,
\phi^{coh}_{\{X_{\alpha}\}}(h_{s(e)}^{-1}g_{e}h_{t(e)})
\,.
$$

This set of coherent states is specially interesting for our perspective because their evaluation at the identity $g_{e}=\id,\,\forall e$, or equivalently its projection on the flat connection state \eqref{FlatState} gives the generating function of spin network evaluations:
$$
\phi^{coh}_{\{X_{\alpha}\}}(\id)
=\la\Omega|\phi^{coh}_{\{X_{\alpha}\}}\ra
%=Z^{\C}(\{X_{\alpha}\})
=Z^{Spin}(\{Y_{e}\})
\,=\,
\sum_{\{j_e\}} \sqrt{\frac{\prod_v (J_v+1)!}{\prod_{ev} (J_v-2j_e)!}} s(\{j_e\}) \prod_e Y_e^{2j_e}\,,
$$
still assuming the matching $X_{\alpha}=\sqrt{Y_{s(\alpha)}Y_{t(\alpha)}}$.
One can further get the whole decomposition of the coherent states $\phi^{coh}_{\{X_{\alpha}\}}$ in the spin network basis $ \varphi_{\{j_e\}}^{\Gamma} $ defined in \eqref{spinbasis}:
\be
\phi^{coh}_{\{X_{\alpha}\}}
\,=\,
\sum_{\{j_e\}} \sqrt{\frac{\prod_v (J_v+1)!}{\prod_{ev} (J_v-2j_e)!}}\,\prod_e Y_e^{2j_e}\,
\varphi_{\{j_e\}}^{\Gamma} \,.
\ee

For here, we can consider two types of probability distribution and averages  for the spins $\{j_{e}\}$. We can look at the averages defined by the amplitude $\phi^{coh}_{\{X_{\alpha}\}}(\id)$, that is the averages weighted by the spin network evaluations:
\be
\la j_{e_{1}}^{n_{1}}..j_{e_{p}}^{n_{p}}\ra
\,\equiv\,
\f1{Z^{Spin}(\{Y_{e}\})}\,
\sum_{\{j_e\}} j_{e_{1}}^{n_{1}}..j_{e_{p}}^{n_{p}}\,\sqrt{\frac{\prod_v (J_v+1)!}{\prod_{ev} (J_v-2j_e)!}} s(\{j_e\}) \prod_e Y_e^{2j_e}\,.
\ee
This is not strictly speaking a true mean value and a probability distribution since the spin network evaluation $s(\{j_e\})$ can be negative. Rather it should be interpreted as an operator insertion in the projection over the flat state:
\be
\la j_{e_{1}}^{n_{1}}..j_{e_{p}}^{n_{p}}\ra
\,=\,
\f{\la\Omega|\, \hat{\j}_{e_{1}}^{n_{1}}\,..\,\hat{\j}_{e_{p}}^{n_{p}}\,|\phi^{coh}_{\{X_{\alpha}\}}\ra}{\la\Omega|\phi^{coh}_{\{X_{\alpha}\}}\ra}\,,
\ee
where the operators $\hat{\j}$ acts by multiplication by the spin $j$ in the spin network basis.
These averages will be directly related to the Ising correlations, as we will investigate in more details in the next section. They are functions of the couplings, $Y_e$ which can be thought of as a background geometry in which quantum fluctuations of the spins $j_{e_i}$ take place.

On the other hand, we can consider the true expectation values for the spin operators in the quantum state $|\phi^{coh}_{\{X_{\alpha}\}}\ra$:
\be
\la j_{e_{1}}^{n_{1}}..j_{e_{p}}^{n_{p}}\ra_{coh}
\,\equiv\,
\f{\la\phi^{coh}_{\{X_{\alpha}\}}|\, \hat{\j}_{e_{1}}^{n_{1}}\,..\,\hat{\j}_{e_{p}}^{n_{p}}\,|\phi^{coh}_{\{X_{\alpha}\}}\ra}
{\la\phi^{coh}_{\{X_{\alpha}\}}|\phi^{coh}_{\{X_{\alpha}\}}\ra}
\,=\,
\f{\sum_{\{j_e\}}j_{e_{1}}^{n_{1}}..j_{e_{p}}^{n_{p}}\, {\frac{\prod_v (J_v+1)!}{\prod_{ev} (J_v-2j_e)!}} \prod_e \f{Y_e^{4j_e}}{(2j_{e}+1)}}
{\sum_{\{j_e\}} {\frac{\prod_v (J_v+1)!}{\prod_{ev} (J_v-2j_e)!}} \prod_e \f{Y_e^{4j_e}}{(2j_{e}+1)}}\,.
\ee
In practice, this is the expectation of the product of half-integers $j_{e_{1}}^{n_{1}}\cdots j_{e_{p}}^{n_{p}}$ with respect to the following probability distribution:
\be
\label{rho}
\hat{\rho}(\{j_e\}) = \frac{\rho(\{j_e\})}{\sum_{\{j_e\}} \rho(\{j_e\})},\qquad \text{with}\qquad
\rho(\{j_e\})
\,=\,
\prod_v{\frac{ (J_v+1)!}{\prod_{e\ni v} (J_v-2j_e)!}} \prod_e \f{Y_e^{4j_e}}{(2j_{e}+1)}\,.
\ee
Notice that it does not feature the spin network evaluation. Let us investigate in more details the shape of this probability profile. It will shed light on the properties of the chosen weights for the generating function of spin network evaluations and on the geometrical interpretation of these coherent states.
To this purpose, we will compute a limit by letting $j^{(n)}:\Edges\to \N$ be a sequence of colorings whose spins go to infinity linearly in $n$, i.e. such that there exists a real valued coloring $j^{\infty}:\Edges\to \N$ such that $\lim_{n\to \infty} \frac{j^{(n)}(e)}{nj^{(\infty)}(e)}=1,\ \forall e\in \Edges$. For simplicity,  in the following computations we will suppress the indices $^{(n)}$ and $^{(\infty)}$ and study the dominant regime for these weights for large values of $n$
%it is peaked on high values for the spins $j_{e}$
%
using the Stirling approximation for the factorials:
$$
\rho(\{j_e\})
\,\approx\,
\prod_e \f{1}{(2j_{e}+1)}\,
\prod_{v} \f{\sqrt{J_v+1}}{2\pi\prod_{e\ni v}\sqrt{J_v-2j_e}}\,
\prod_e Y_e^{4j_e}
\prod_{v} \f{(J_v+1)^{(J_v+1)}}{\prod_{e\ni v}(J_v-2j_e)^{(J_v-2j_e)}}\,,
$$
where we have separated the exponential  contribution from the polynomial pre-factors and by $\approx$ we mean that the limit of the ratio of the two sides when $n$ tends to $\infty$ is $1$. Assuming that the exponential factor controls the main behavior of this distribution, similarly to a Poisson distribution, we focus on the exponent:
\be
\Phi=\sum_{e}4j_{e}\ln Y_{e}+ \sum_{v}\left[(J_v+1)\ln (J_v+1) -\sum_{e\ni v}(J_v-2j_e)\ln(J_v-2j_e)\right]\,.
\ee
Let us do a stationary phase approximation, leading to a Gaussian approximation for the distribution $\rho$ peaked around maxima given by the stationary points of $\Phi$. Still assuming that the $j$'s are large and neglecting sub-leading contributions, we get:
\be\label{eq:criticalY}
\f{\pp\Phi}{\pp j_{e}}=0
\Leftrightarrow
\qquad
Y_{e}^{4}\approx
\f{(J_{s(e)}-2j_{e_{1}})(J_{s(e)}-2j_{e_{2}})}{J_{s(e)}(J_{s(e)}-2j_{e})}\,
\f{(J_{t(e)}-2j_{\te_{1}})(J_{t(e)}-2j_{\te_{2}})}{J_{t(e)}(J_{t(e)}-2j_{e})}\,,
\ee
where $e_{1},e_2$ are the edges touching $e$ in $s(e)$ and $\te_{1},\te_{2}$ are those touching it in $t(e)$, as depicted in Figure \ref{fig:triangles}.
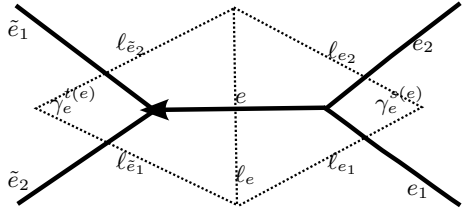
\begin{figure}[h!]
\raisebox{-0.4cm}{\put(-1,8){$\te_2$}\put(-1,65){$\te_1$}\put(152,60){$e_2$}\put(85,40){$e$}\put(85,10){$\ell_e$}\put(120,15){$\ell_{e_{1}}$}\put(120,57){$\ell_{e_{2}}$}\put(150,5){$e_{1}$}\put(40,15){$\ell_{\te_{1}}$}\put(40,60){$\ell_{\te_{2}}$}\put(138,37){$\gamma^{s(e)}_{e}$}\put(15,37){$\gamma^{t(e)}_{e}$} \includegraphics[width=6cm]{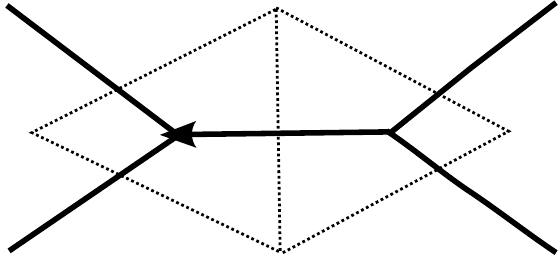}}\ .
\caption{The notation used on dual triangles to $\Gamma$.} \label{fig:triangles}\end{figure}
At this point let us do a little triangle geometry. Assuming that we have three edges of length $l$, $l_{1}$ and $l_{2}$, with total half-perimeter $L$, we have Heron formula for the area, a similar formula for the inner circle radius and expressions for the sine and cosine of the opposite angle $\gamma$ to $l$: 
$$
A^{2}=L(L-l)(L-l_{1})(L-l_{2}),
\quad
r^{2}=\f{(L-l)(L-l_{1})(L-l_{2})}{L},
\quad
%\cos\gamma=\f{l_{1}^{2}+l_{2}^{2}-l^{2}}{2l_{1}l_{2}},
1+\cos\gamma=\f{2L(L-l)}{l_{1}l_{2}}
\quad
\sin\gamma=\f{2A}{l_{1}l_{2}}\,.
$$
Combining all this, in order to describe the stationary point, it seems natural to define dual triangles to the (planar) graph $\Gamma$ by considering triangles around every vertex with edge lengths given by $l_{e}=2j_{e}$ (cf. Figure \ref{fig:triangles}). Then the right hand side of \eqref{eq:criticalY} can be re-interpreted in terms of the opposite angles to the edge dual to $e$ in both triangles dual to the source and target vertices $s(e)$ and $t(e)$. This gives the critical values of the couplings: 
\be
\label{Ygeom}
Y_{e}^{2}=
\left(\f{\sin\gamma_{e}^{s(e)}}{1+\cos\gamma_{e}^{s(e)}}\right)\,
\left(\f{\sin\gamma_{e}^{t(e)}}{1+\cos\gamma_{e}^{t(e)}}\right)
=
\tan\f{\gamma_{e}^{s(e)}}{2}\,
\tan\f{\gamma_{e}^{t(e)}}{2}
\,.
\ee
They correspond to a Euclidean structure on the plane obtained by gluing Euclidean triangles dual to the edges of $\Gamma$ whose edge lengths are $2j_{e}$, up to an arbitrary global rescaling.

\smallskip

We will not perform here the details of the resulting stationary phase approximation, but we will focus on the meaning of these ``geometric'' couplings from the point of view of the Ising model. In the special case of a honeycomb lattice (regular hexagonal), all the angles are $\gamma=\f\pi6$ and  this value of the geometric coupling matches exactly the value of the critical value of the Ising model, $Y^{c}=\f1{\sqrt{3}}$.

This actually holds for the much larger class of isoradial graphs, i.e. such that all faces are inscribable in a circle of a given radius and whose center is in the face. Let us consider such an isoradial embedding of our planar graph $\Gamma$. The dual vertex of a face, i.e the corresponding point of the triangulation, is chosen as the center of its circumscribed circle, as illustrated on Figure \ref{fig:isoradial}.
\begin{figure}
\includegraphics[width=10cm]{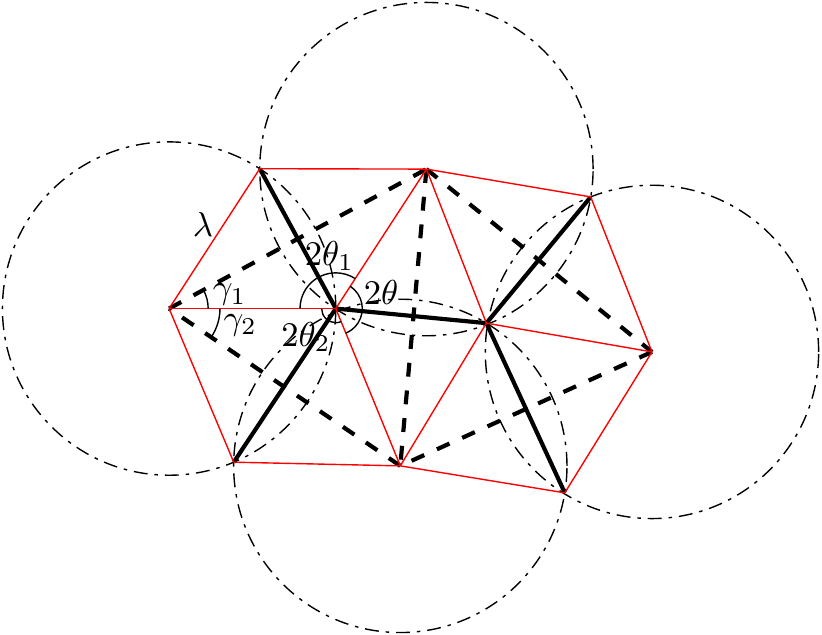}
\caption{In an isoradial graph, the half-rhombus angle $\theta_e$ is equal to both $\gamma^{s(e)}_e$ and $\gamma^{t(e)}_e$. Indeed in the above drawing it holds $\gamma_e=\gamma_1+\gamma_2$, the red triangles are isosceles and the dotted lines are perpendicular to the edges of $\Gamma$.}\label{fig:isoradial}
\end{figure}
 Each edge of the graph $\Gamma$ is at the intersection of two circles, while the dual triangulation edge links the centers of those two circles and is orthogonal to the graph edge. In this special setting, it turns out that the half-rhombus angle $\theta_{e}$ associated to the graph edge $e$ (or equivalently its dual triangulation edge) is equal to both opposite triangle angles $\gamma_{e}=\gamma_{e}^{s(e)}=\gamma_{e}^{t(e)}$ (which are then equal to each other). 
Then the critical Ising couplings read (see \cite{Ke},\cite{isoradial} and the survey \cite{BoutillierDeTiliereSurvey} for additional material):
\be
e^{2y_{e}^{c}}=\f{1+\sin\theta_{e}}{\cos\theta_{e}}
\qquad
\Longrightarrow\quad
Y_{e}^{c}=\tanh y_{e}^{c}=\f{e^{2y_{e}^{c}}-1}{e^{2y_{e}^{c}}+1}
=\f{1+\sin\theta_{e}-\cos\theta_{e}}{1+\sin\theta_{e}+\cos\theta_{e}}
=\tan\f{\theta_{e}}{2}
=Y_{e}^{geom}\,.
\ee

This provides a neat geometrical interpretation of both the generating function of spin network evaluations and of the critical regime of the 2D Ising model, at least in the context of isoradial graphs and their dual triangulation. It would be very interesting to investigate if this correspondence between stationary points in the spin networks and critical couplings of the Ising model is more general than this setting and can be generalized by our formula \eqref{Ygeom}, for instance for Delaunay triangulations (the graph then being its dual Vorono\"\i{} diagram).

\smallskip

This result further leads to a few very interesting questions:

\begin{enumerate}

\item What is the behavior of the generating function of spin network evaluations when the edge couplings do not admit stationary points, and the stationary approximation fails? Normally we would expect an exponential behavior. Then, is it related to some exponential decay (e.g. of the 2-point function) in the non-critical Ising model in low or high temperature regimes?

\item As we pointed out above, the conditions \eqref{Ygeom} relating the dual triangulation to the couplings are scale-independent. They only depend on the angles and we can rescale all the edge lengths by an arbitrary factor without affecting the couplings $Y_{e}$. Thus we actually have a whole \emph{line} of stationary points by rescaling arbitrarily the spins $j_{e}$'s. %This should be revised for small spins where our  high spin approximation breaks down.
In \cite{bonzom} in the case of two-vertex graphs, it was found that this line of stationary points is related to the singularities of the generating function of spin network evaluations, i.e. to the set of zeroes of $P_\G = \sum_{c\subset \Gamma}\prod_{e\in c}Y_{e}$. Can we expect also such a relation here? This requires a stationary phase evaluation of the generating function $Z^{spin}$. Note that the zeroes of $P_\G$ are those of $Z^{Ising}$. Since $P_\G$ is polynomial (in $Y_e$) for all finite graphs $\G$, the zeroes completely determine the partition function. In the context of statistical mechanics and in particular for the Ising model, they are called Fisher's zeroes and their distribution in the thermodynamic limit determines the critical properties \cite{Fisher}. Could it then be that critical properties of the Ising model can be extracted from the stationary phase approximation of the spin network generating function?

%Following the logic proposed in \cite{bonzom}, we expect  a relation between the existence of these stationary points and the poles of the generating function, i.e. when the polynomial $P_{\Gamma}=\sum_{c\subset \Gamma}\prod_{e\in c}Y_{e}$ vanishes.
%%We should check whether we have a pole when we have an admissible set of couplings $Y_{e}$ allowing for a dual triangulation. 
%In that case, can we extract the critical exponent of the Ising model from the stationary phase approximation of the spin network generating function?

\item We could modify the generating function of spin network evaluations in order to get a single non-scalable stationary point, which would depend on the edge coupling $Y_{e}$. This is actually the behavior of coherent spin network states introduced in \cite{coherent}  in the context of loop quantum gravity. As we will describe in section \ref{coherent} below, these define a new generating function for spin network evaluations with a slightly different statistical weight depending on the spins $j_{e}$. Although these other coherent states have a very nice geometric interpretation as semi-classical geometries, what would be their counter-part in terms of Ising model?

\item Finally there is a clash with the correspondence with the Ising model. Indeed here the couplings $Y_{e}$ are allowed to run over all real (positive) values, while the correspondence with the Ising model requires bounded values for the $Y$'s given by $Y_{e}=\tanh y_{e}$. This will be addressed in  section \ref{O(n)} below by generalizing the Ising model to O$(n)$ models.

\end{enumerate}

%%%
\subsection{Mapping Ising Correlations to Spin Averages}\label{sec:correlations}
%%%

The relationship between the partition function of the 2D Ising model and the generating function of spin networks induces a correspondence between the observables of both models. Those observables are on the one hand the Ising spin correlation functions,
\begin{equation}
\langle \sigma_{v_1}\,\sigma_{v_2} \dotsm \sigma_{v_n} \rangle = \frac1{Z^{Ising}} \sum_\sigma \sigma_{v_1}\,\sigma_{v_2} \dotsm \sigma_{v_n}\ e^{\sum_e y_e \sigma_s(e) \sigma_t(e)},
\end{equation}
and on the other hand the expectations of products of colors,
\begin{equation}
\langle j_{e_1}^{n_1} j_{e_2}^{n_2} \dotsm j_{e_k}^{n_k} \rangle
=
\frac1{Z^{Spin}}
\sum_{col} j_{e_1}^{n_1} j_{e_2}^{n_2} \dotsm j_{e_k}^{n_k}
\,s(\G,\{j_{e}\})\,
%\langle \Gamma,col\rangle^{int}
\prod_e (\tanh y_e)^{2j_e},
\end{equation}
where $2j_e\in\mathbb{N}$ is the color of the edge $e$. The correspondence is obtained by taking derivatives of the fundamental equality,
\begin{equation}
Z^{Spin}\, (Z^{Ising})^2 = 4^{\#\Ver}\,\prod_e \cosh^2 y_e,
\qquad
\ln Z^{Spin} +2 \ln Z^{Ising} =\sum_{e} 2\ln\cosh y_{e} + 2(\#\Ver)\ln 2\,.
\end{equation}
The two subtleties are the non-zero term on the right hand side, but it will disappear upon differentiating twice with respect two different edge couplings $y_{e}$, and the dependence of the generating function $Z^{Spin}$ on $\tanh y_{e}$ instead of simply $y_{e}$ like the Ising partition function, but this is easily accounted for by:
$$
Y=\tanh y \quad\Rightarrow\quad
\pp_{Y}f=\cosh^{2}y\pp_{y} f\,.
$$

\smallskip

Taking a first derivative with respect to a given edge coupling $y_e$ gives:
\begin{equation}
\label{FirstDerivative}
\frac1{2\,\cosh^2 y_e} \partial_{\tanh y_e} \ln Z^{spin} = \tanh y_e - \partial_{y_e} \ln Z^{Ising}.
\end{equation}
This is easily interpreted by introducing the nearest-neighbor correlations, i.e. correlations between the two Ising spins incident to the edge $e$,
\begin{equation}\label{eq:isingcorr}
g_e\equiv \langle \sigma_{s(e)} \sigma_{t(e)} \rangle = \partial_{y_e} \ln Z^{Ising},
\end{equation}
and the expectation of the color on the edge $e$,
\begin{equation} \label{DefMeanSpin}
\langle 2j_e\rangle = \tanh y_e\ \partial_{\tanh y_e} \ln Z^{spin}.
\end{equation}
Then equation \eqref{FirstDerivative} simply relates $g_e$ to $\langle j_e\rangle$,
\begin{equation} \label{MeanSpin}
\langle j_e \rangle = \sinh y_e \bigl(\sinh y_e - \cosh y_e\,g_e\bigr),
\qquad
\langle \sigma_{s(e)} \sigma_{t(e)} \rangle
\,=\,
\tanh y_{e} - \f1{\sinh 2 y_{e}}\,\langle 2j_e\rangle\,.
\end{equation}
providing an expression for the mean color of an edge at fixed couplings, in terms of the spin-spin correlation of this edge in the Ising model.

We then differentiate successively with respect to different edge couplings $y_{e}$ along a path of edges between two vertices in order to obtain longer range Ising correlations. For instance, let us start with three vertices $v_{0,1,2}$ linked successively by the edges $e_{1}$ and $e_{2}$. We differentiate with respect to $y_{1}$ and $y_{2}$ and get:
\be
\partial_{y_e} \partial_{y_f} \left(\ln(Z^{Ising})\right)=\frac{\partial_{y_e}\partial_{y_f}Z^{Ising}}{Z^{Ising}}-\frac{\partial_{y_e}Z^{Ising}\partial_{y_f}Z^{Ising}}{(Z^{Ising})^2}=\langle \sigma_{0} \sigma_{2} \rangle -\langle \sigma_{0} \sigma_{1} \rangle \langle \sigma_{1} \sigma_{2} \rangle 
\ee
\be
Y_eY_f\partial_{Y_e} \partial_{Y_f} \left(\ln(Z^{spin})\right)=\frac{Y_eY_f\partial_{Y_e}\partial_{Y_f}Z^{spin}}{Z^{spin}}-\frac{Y_e\partial_{Y_e}Z^{spin}Y_f\partial_{Y_f}Z^{spin}}{(Z^{spin})^2}=\langle (2j_{1})( 2j_{2})\rangle-\langle 2j_{1}\rangle\langle 2j_{2}\rangle
\ee

\be
\langle \sigma_{0} \sigma_{2} \rangle -\langle \sigma_{0} \sigma_{1} \rangle \langle \sigma_{1} \sigma_{2} \rangle 
\,=\,
\f{-2}{\sinh 2y_{1}\sinh 2y_{2}}\,\Big{(}
\langle (2j_{1})( 2j_{2})\rangle-\langle 2j_{1}\rangle\langle 2j_{2}\rangle
\Big{)}\,.
\ee
More generally, we consider two vertices on the graph $\G$,  that is an initial vertex $v_{0}=v$ and a final vertex $v_{n}=w$ linked by a path $\cP$ consisting in $n$ edges, $e_{1}$ to $e_{n}$.
We consider cuts of this path: the cut $P$ of the path $\cP$ is defined by  $p$ intermediate vertices on $\cP$ numbered $v^{i}_{P}$, $i=1..(p-1)$, and divides the path $\cP$ into $(p+1)$ smaller paths enumerated as $\cP_{P}^{I}$ with $I=1..(p+1)$.
Differentiating the logarithm of the partition function and generating function  with respect to  the $n$ variables $y_{1}$ to $y_{n}$ yields a sum over all such cuts $P$ of the path $\cP$:
\be
\langle \sigma_{v} \sigma_{w} \rangle^{(\cP)}_{c}
\,=\,
\f{-2^{n-1}}{\prod_{e\in\cP}\sinh(2j_{e})}\,\langle \prod_{e\in\cP}(2j_{e})\rangle^{(\cP)}_{c}\,,
\ee
where we define the ``connected'' correlations as:
\be
\langle \sigma_{v} \sigma_{w} \rangle^{(\cP)}_{c}
\,\equiv\,
\sum_{P\,|\,\cP} (-1)^{p} \langle \sigma_{v}\sigma_{v^{1}_{P}} \rangle
\,..\,
\langle \sigma_{v^{p}_{P}} \sigma_{w} \rangle\,,
\qquad
\langle \prod_{e\in\cP}(2j_{e})\rangle^{(\cP)}_{c}
\,\equiv\,
\sum_{P\,|\,\cP} (-1)^{p}\,\prod_{I=1}^{p+1}\langle \prod_{e\in\cP_{P}^{I}}(2j_{e})\rangle\,.
\ee

%%%
\subsection{Distribution of the Edge Color}\label{sec:distributionspin}
%%%

Now we restrict our attention to the observables on a single edge, namely $\langle (2j_e)^n\rangle$, for which the correspondence with the Ising model leads to explicit expressions.

\begin{theo}\label{teo:genserexpect}
The exponential generating function of the spin averages $\langle (2j_e)^n\rangle$ is
\begin{equation}
j_e(t) = \sum_{n\geq0} \langle (2j_e)^n\rangle \frac{t^n}{n!} = \frac1{\left(1-\langle j_e\rangle (e^t-1)\right)^2}.
\end{equation}
It can be interpreted as the moment generating function $j_e(t) = \sum_{n\geq0} P(2j_e=n)\,e^{nt}$ for the following distribution,
\begin{equation}
P(2j_e=n) = \frac{n+1}{(1+\langle j_e\rangle)^2}\,\left(\frac{\langle j_e\rangle}{1+\langle j_e\rangle}\right)^n,
\end{equation}
where $\langle j_e\rangle$ is given in \eqref{MeanSpin}.
\end{theo}

Before proceeding to the proof, we emphasize that this is not strictly speaking a probability distribution, since $\langle j_e\rangle$ can be negative as we will see in the next section. This can be traced back to the fact that we are considering spin averages weighted by the spin network evaluations (which can be negative), which is not the expectation values of spin operators on coherent states, as underlined earlier in section \ref{state}.

\begin{proof}
Before beginning the proof {\it per se}, we want to emphasize that the simplicity of the result relies on fact that the derivatives of the Ising free energy $\ln Z^{Ising}$ with respect to a given coupling $y_e$ are all simple functions of the nearest-neighbor correlation. This is due to the fact that $\sigma_v^2 =1$ for all Ising spins. In particular, this implies $\partial_{y_e}^2 Z^{Ising}/Z^{Ising}=1$, so that the second derivative of the free energy, which also is the derivative of the nearest-neighbor correlation, reads
\begin{equation} \label{PartialG_e}
\partial_{y_e} g_e = \partial_{y_e}^2 \ln Z^{Ising} = \frac{\partial_{y_e}^2 Z^{Ising}}{Z^{Ising}} - \left(\frac{\partial_{y_e} Z^{Ising}}{Z^{Ising}} \right)^2 = 1-g_e^2.
\end{equation}

The first thing to do is to relate the expectation $\langle (2j_e)^n \rangle$ to the derivatives of the spin network free energy $\ln Z^{spin}$. We start with the standard expansion of $(2j)^n$ in terms of falling factorials,
\begin{equation}
\langle (2j_e)^n\rangle = \sum_{k=1}^n S(n,k) \langle (2j_e)^{\underline{n}} \rangle,
\end{equation}
where $x^{\underline{n}}= x(x-1)\dotsm (x-n+1)$ is the falling factorial and $S(n,k)$ are the Stirling numbers of the second kind. Notice that
\begin{equation}
\langle (2j_e)^{\underline{n}} \rangle =  \frac{\tanh^n y_e}{Z^{spin}}\,\partial_{\tanh y_e}^n Z^{spin}.
\end{equation}
The derivatives of $Z^{spin}$ can be related to those of $\ln Z^{spin}$ using Fa\`a Di Bruno's formula\footnote{This gives here
\begin{equation}
\partial_{y_e}^n \ln Z^{spin} = \sum_{\substack{\{\nu_p\geq0\}_{p\geq1} \\ \sum_p p\nu_p=n}} (-1)^{1+\sum \nu_p} \frac{n!\,(\sum_p \nu_p -1)!}{\prod_p \nu_p!\,p!^{\nu_p}} \prod_p \left(\frac{\partial_{y_e}^p Z^{spin}}{Z^{spin}}\right)^{\nu_p}.
\end{equation}}, but it is simpler to use generating functions. Defining the shorthand notations
\begin{equation}
\mu_n = (\partial_{\tanh y_e} Z^{spin})/Z^{spin}\quad \text{and } \kappa_n = \partial_{\tanh y_e}^n \ln Z^{spin},
\end{equation}
it is well-known that if
\begin{equation}
m(t) = \sum_{n\geq 0} \mu_n \frac{t^n}{n!},\quad \text{then }\kappa(t) = \sum_{n\geq 1} \kappa_n \frac{t^n}{n!} = \ln m(t).
\end{equation}
The strategy is therefore to first find $\kappa_n$ using the correspondence with the Ising model, then take the exponential of $\kappa(t)$ to get $m(t)$ and extract its coefficients $\mu_n$. We prove by induction that
\begin{equation}
\kappa_n = 2\,(n-1)!\,\left(\kappa_1/2\right)^n.
\end{equation}
This clearly holds true for $n=1$. In addition, we recall from the definition of $\kappa_n$ and the equations \eqref{DefMeanSpin} and \eqref{MeanSpin}
\begin{equation}
\kappa_1 = \langle 2j_e \rangle/\tanh y_e = 2\,\cosh^2 y_e\, (\tanh y_e - g_e).
\end{equation}
Then the following holds:
\begin{equation} \label{Kappa_{n+1}}
\kappa_{n+1} = \partial_{\tanh y_e} \kappa_n = 2 n!\,(\kappa_1/2)^{n-1} \partial_{\tanh y_e} (\kappa_1/2),
\end{equation}
where the last equality makes use of the induction hypothesis. We therefore have to evaluate $\kappa_2$,
\begin{equation}
\begin{aligned}
\kappa_2 &= \partial_{\tanh y_e} \kappa_1 = 2 \partial_{\tanh y_e}\left( \cosh^2 y_e (\tanh y_e - g_e)\right)\\
&= 4 \cosh^3 y_e \sinh y_e (\tanh y_e - g_e) + 2 \cosh^2 y_e (1-\cosh^2 y_e \partial_{y_e} g_e)\\
&= 2\,\cosh^4 y_e\,(\tanh y_e - g_e)^2 = 2 (\kappa_1/2)^2.
\end{aligned}
\end{equation}
(using $\partial_{\tanh y_e} = \cosh^2 y_e \partial_{y_e}$). From the second line to the third, we have used \eqref{PartialG_e} and some simple algebraic manipulations. Plugging this expression into \eqref{Kappa_{n+1}} proves the formula for $\kappa_n$.
The generating function $\kappa(t)$ is then
\begin{equation}
\kappa(t) = 2\sum_{n\geq 1} \frac1n\,\left(\frac{\kappa_1 t}{2}\right)^n = -2 \ln \left(1-\frac{\kappa_1 t}{2}\right),
\end{equation}
so that
\begin{equation}
m(t) = e^{\kappa(t)} = \frac1{\left(1-\frac{\kappa_1 t}{2}\right)^2}, \quad \text{and}\quad \mu_n = [t^n/n!]m(t) = (n+1)!\,(\kappa_1/2)^n.
\end{equation}
Up to a factor $\tanh^n y_e$, this is the expectation of $(2j_e)^{\underline{n}}$. We can now form the generating function of the expectations $\langle (2j_e)^n\rangle$,
\begin{equation}
\begin{aligned}
j(t) = \sum_{n\geq 0} \langle (2j_e)^n\rangle \frac{t^n}{n!} &= 1+\sum_{n\geq 1} \sum_{k=1}^n S(n,k) (k+1)! (\tanh y_e \kappa_1/2)^k\,\frac{t^n}{n!}\\
&= 1 + \sum_{k\geq 1} \Bigl(\sum_{n\geq k}^n S(n,k)\frac{t^n}{n!}\Bigr) (k+1)! (\tanh y_e \kappa_1/2)^k
\end{aligned}
\end{equation}
One recognizes the exponential generating function of the Stirling numbers of the second kind, $\sum_{n\geq k} S(n,k)t^n/n! = (e^t -1)^k/k!$. Therefore, together with $\langle 2j_e\rangle = \tanh y_e \kappa_1$,
\begin{equation}
j(t) = \sum_{k\geq 0} (k+1) \left(\langle j_e\rangle (e^t-1)\right)^k = \frac1{\left(1-\langle j_e\rangle (e^t-1)\right)^2}.
\end{equation}
In order to identify a discrete probability distribution, this expression has to be expanded onto powers of $e^t$, which is readily done,
\begin{equation}
j(t) = \frac1{(1+\langle j_e\rangle)^2} \sum_{n\geq0} (n+1)\,\left(\frac{\langle j_e\rangle}{1+\langle j_e\rangle}\right)^n\,e^{nt},
\end{equation}
and the coefficient of $e^{nt}$ is interpreted as the probability $P(2j_e=n)$.
\end{proof}

We can furthermore write the expectations as polynomials of order $n$ in the mean color $\langle j_e\rangle$,
\begin{equation}
\langle (2j_e)^n\rangle = \sum_{k=1}^n S(n,k)\,(k+1)!\,\langle j_e\rangle^k,
\end{equation}
or as
\begin{equation}
\begin{aligned}
\langle (2j_e)^n\rangle &= \sum_{k\geq0} P(2j_e=k) k^n\\
&= \frac1{(1+\langle j_e\rangle)^2} \left(\operatorname{Li}_{-n-1}\left(\frac{\langle j_e\rangle}{1+\langle j_e\rangle}\right) + \operatorname{Li}_{-n}\left(\frac{\langle j_e\rangle}{1+\langle j_e\rangle}\right)\right),
\end{aligned}
\end{equation}
where $\operatorname{Li}_n(z) = \sum_{k\geq1} k^{-n}z^k$ is the polylogarithm.

%%%
\subsection{$O(n)$ Models,  Critical Ising Model and Phase Diagram}
\label{O(n)}
%%%

Let us look at the potential critical behavior of the spin network generating function from the point of view of its duality with the 2D Ising model. For the sake of simplicity,  we restrict our attention to homogeneous couplings, $Y_e=Y,\ \forall e\in \Edges$. While exploring the range of all edge couplings is desirable from the point of view of spin networks and quantum geometry states, we see that it is not possible within the frame of the Ising model since $Y$ is restricted to be smaller than 1, $Y=\tanh y \leq 1$.  
To explore the regime $Y\geq 1$, we match the spin network generating function with the squared inverse partition function of the $O(1)$ model,
\begin{equation}
Z^{spin}(Y) = 1/(Z_{O(1)}(Y))^2.
\end{equation}
The $O(n)$ model, with $n$ an integer, is defined as follows. Some $n$-component spins $\vec{S}_i$ sit on the vertices of the graph and each equipped with a normalized measure $d\vec{S}$ such that $\int d\vec{S}\,S^\mu S^\nu = \delta^{\mu\nu}$. The partition function reads
\begin{equation}
Z_{O(n)}(Y) = \int \prod_v d\vec{S}_v\ \prod_e \bigl(1 + Y\,\vec{S}_{s(e)}\cdot \vec{S}_{t(e)}\bigr)
\end{equation}
By expanding the product of the edge weights and performing the integral, a formulation as a sum over loop configurations is obtained
\begin{equation}
Z_{O(n)}(Y) = \sum_{\gamma\in \mathcal{G}} n^{C(\gamma)} \, Y^{\#\Edges(\gamma)}
\end{equation}
where $C(\gamma)$ is the number of connected components of the loop configuration $\gamma$ and $\#\Edges(\gamma)$ the total number of edges it covers.

In the scaling limit, the $O(n)$ model gives rise for $-2<n\leq 2$ to the following the phase diagram (adapted from \cite{Dubail}),
\begin{equation*}
\begin{array}{c} \includegraphics[scale=.75]{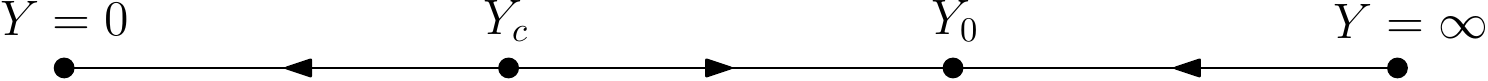} \end{array}
\end{equation*}
Introducing $\gamma\in [0,\pi)$ such that $n=2\cos \pi\gamma$, it is found that \cite{NienhuisPRL, NienhuisCG, JacobsenCG}
\begin{itemize}
\item[--] $Y=Y_c$ is critical, with central charge $c=1-6\frac{(g-1)^2}{g}$, for $g=1+\gamma$. Moreover $Y_c = 1/\sqrt{2+\sqrt{2-n}}$ on the hexagonal lattice. This is a second order phase transition with dilute loops.
\item[--] The region above $Y_c$, i.e. $Y>Y_c$, is called the dense phase and flows towards $Y_0$ whose location is $Y_0=1/\sqrt{2-\sqrt{2-n}}$ on the hexagonal lattice.
\item[--] At $Y=\infty$, one gets the fully packed $O(n)$ loop model (whose universality class is lattice dependent).
\end{itemize}
Our interest is the case $n=1$, $\gamma = 1/3$. Then, $Y=0$ is the infinite temperature Ising model, with no loops at all. $Y=Y_c = 1/\sqrt{3}$ is the Ising critical point with central charge $c=1/2$, while $Y_0=Y$ corresponds to the Ising model at zero temperature.

We recall the expression which was found for the average spin of an edge in the Ising case, i.e. $Y=\tanh y$,
\begin{equation*}
\langle j_e\rangle = \sinh^2 y - g_e(y) \sinh y \cosh y,
\end{equation*}
where $g_e(y)$ is the nearest-neighbor correlation at coupling $y$. Notice that the expectation of $j_e+1/2$ takes the simple form
\begin{equation} \label{j+1/2}
\langle j_e + \frac12 \rangle = \frac12 \bigl[\cosh 2y - g_e(y) \sinh 2y\bigr] = \frac14 \Bigl[e^{2y}(1-g_e(y)) + e^{-2y}(1+g_e(y))\Bigr]
\end{equation}
At small coupling $Y=\tanh y\sim \sinh y$, $g_e(y)$ goes to zero as a power law (at least like $y^2$ if there is no 2-cycle), and the expected spin behaves like
\begin{equation}
\langle j_e\rangle \sim Y(Y-g_e(Y)),
\end{equation}
which obviously goes to zero, and thus $\langle j_e+1/2\rangle$ goes to $1/2$.

When $Y$ goes to $Y_0=1$, i.e. $y\to\infty$, the correlation $g_e(y)$ goes to 1, and therefore the contribution $e^{-2y}(1+g_e(y))$ goes to $2e^{-2y}$. To evaluate the behavior of $1-g_e(y)$, one uses the low temperature expansion of the Ising model\footnote{It is an expansion with respect to the two configurations where all Ising spins are aligned, $Z^{Ising} = 2e^{y\#\Edges} (1 + C_2 e^{-4y} + \mathcal{O}(e^{-6y}))$. Corrections like $e^{-2py}$ come from flipping some Ising spins such that $p$ edges have opposite spins at their ends. There is no $p=0$ correction since that would require a single edge to have opposite spins on its vertices and this is impossible if $\G$ is bridgeless. Then, $C_2$ is the number of pairs of 2-cut edges (i.e. such that cutting 2 edges disconnects $\G$). Similarly, $\sum_{\{\sigma_v\}} \sigma_1 \sigma_2 e^{\sum_e y_e \sigma_{s(e)} \sigma_{t(e)}} = 2 e^{y\#\Edges} (1 + C_2' e^{-4y} + \mathcal{O}(e^{-6y}))$ with $C_2'\leq C_2$.}: it is easily checked that $1-g_e(y)$ decays at least like $e^{-4y}$. Therefore,
\begin{equation}
\langle j_e + \frac12 \rangle \sim \frac12\,e^{-2y} \sim \frac{1-Y}{4},
\end{equation}
which also goes to zero.

Furthermore, explicit formula exist for the the nearest-neighbor correlations of the Ising model. For instance, on the hexagonal lattice with isotropic coupling $Y$, one has (suppressing the edge dependence) \cite{Baxter399, BaxterBook}
\begin{equation}
g(y) = \coth \bigl(2L(y)\bigr)\,\bigl(a(k(y)) A(y) - b(k(y)) B(y)\bigr),
\end{equation}
with $L(y) = \frac14 \ln \frac{\cosh(3y)}{\cosh(y)}$ and $k(y)=1/(\sinh 2L(y)\,\sinh 2y)$, and
\begin{align}
a(k) &= \frac1{\pi}\bigl( (1+k) E(k_1) + (1-k)K(k_1)\bigr)\\
b(k) &= \frac2{\pi} (1-k) K(k_1).
\end{align}
Here $K(k), E(k)$ are respectively the complete elliptic integrals of the first and second kinds of modulus $k$, and $k_1 = 2\sqrt{k}/(1+k)$. Finally,
\begin{align}
A(y) &= F(\arctan \sinh 2L(y), 1-k^2(y)) \\
B(y) &= \frac1{1-k^2(y)}\Bigl( F(\arctan \sinh 2L(y), 1-k^2(y)) - E(\arctan \sinh 2L(y), 1-k^2(y))\Bigr),
\end{align}
where $F(\phi,k), E(\phi, k)$ are the incomplete elliptic integrals of the first and second kinds of modulus $k$.

\begin{figure}
\includegraphics[scale=.45]{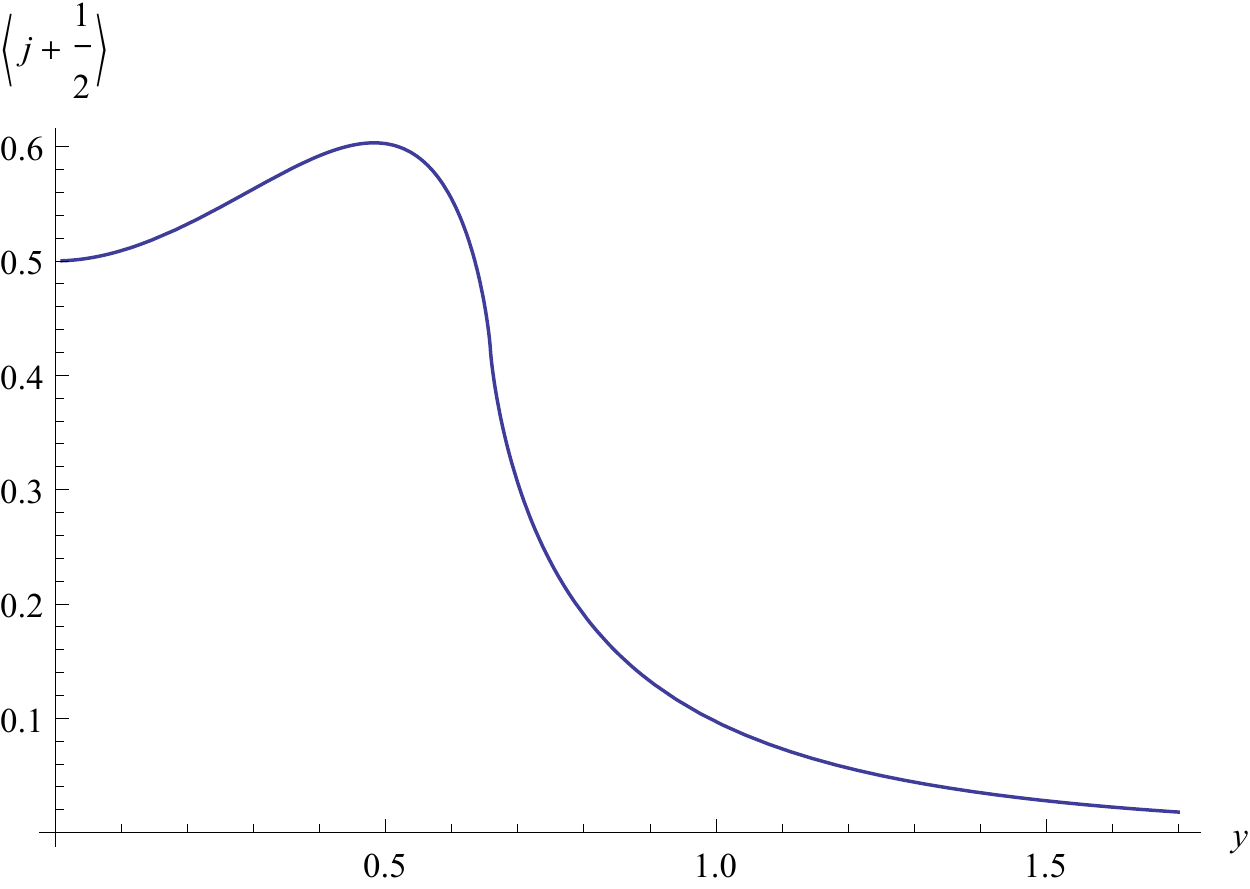} %\hspace{1.5cm} 
\includegraphics[scale=.45]{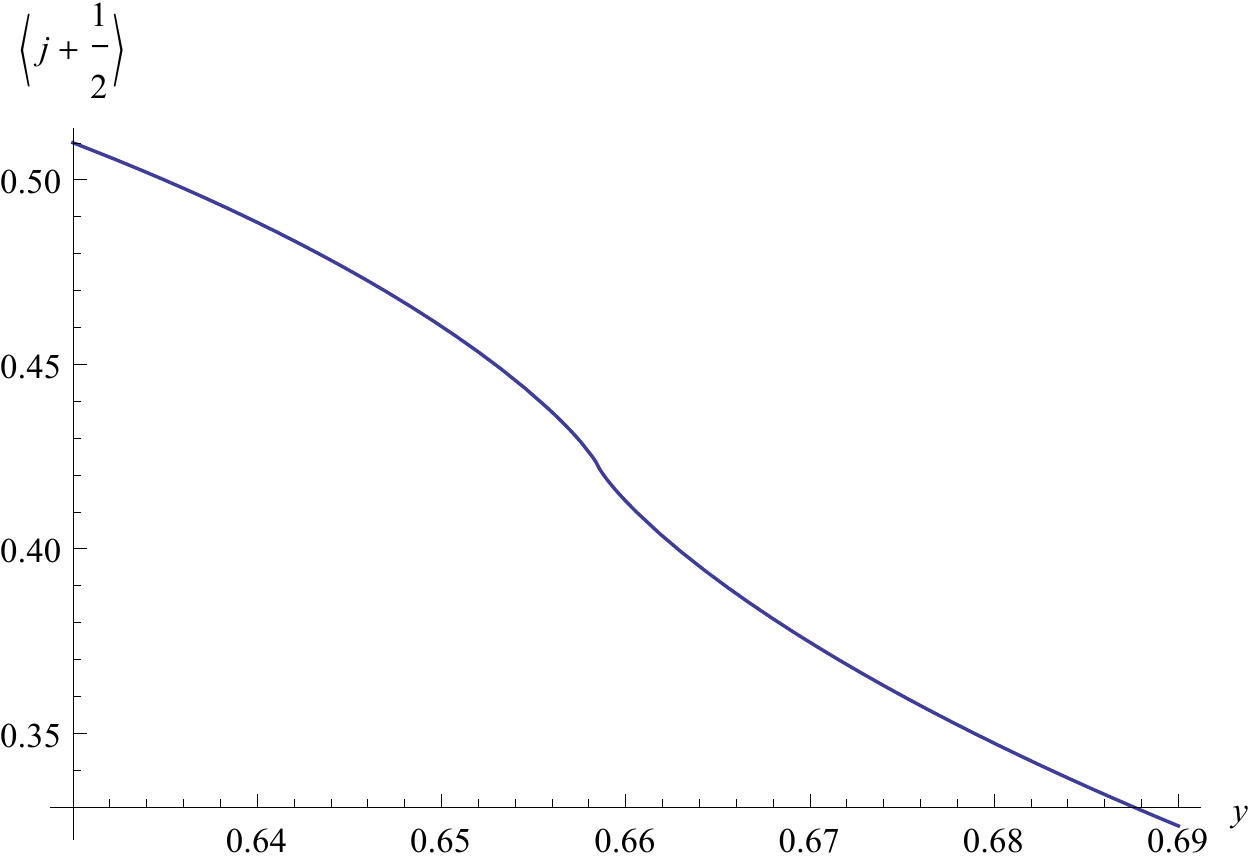}
\includegraphics[scale=.45]{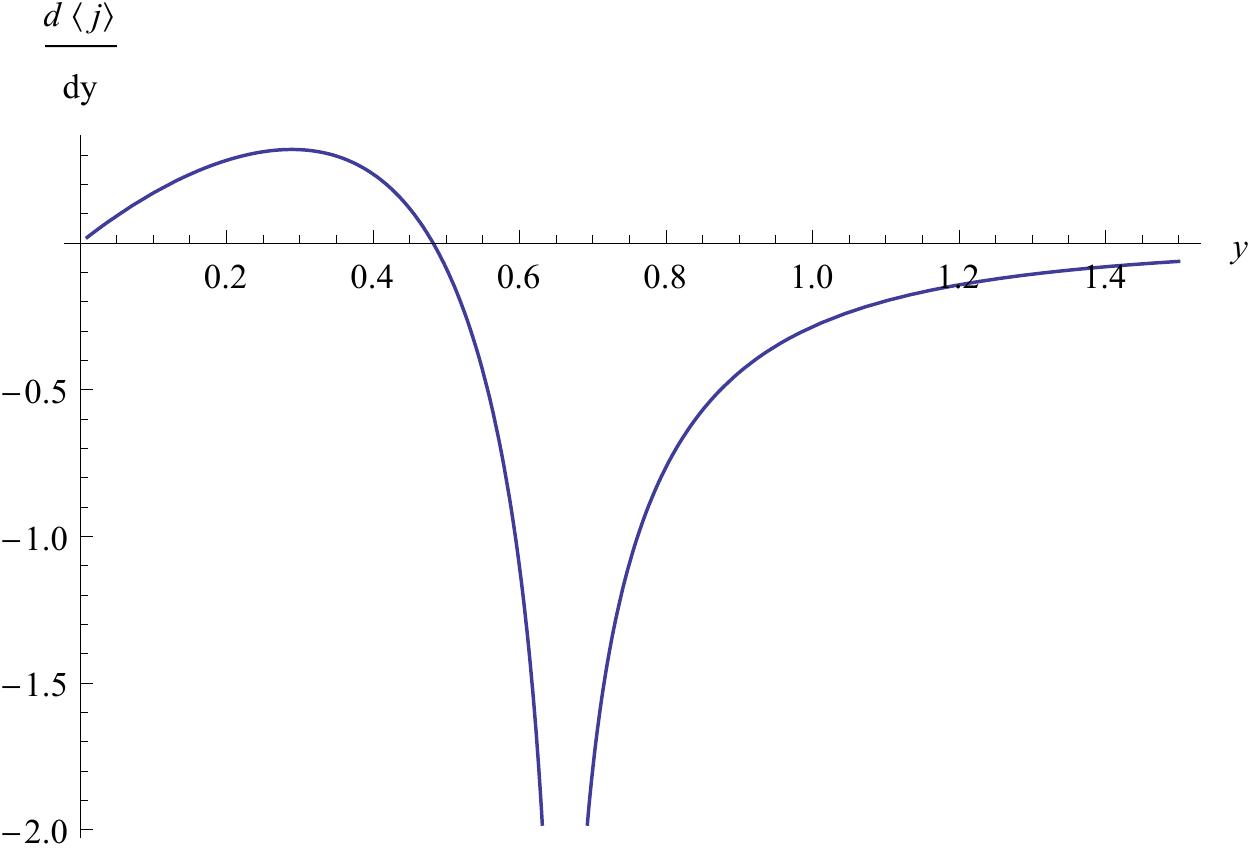}
\caption{\label{fig:PlotMeanJ} The two plots on the left are plots of $\langle j+1/2\rangle$ as a function of $y$ on the hexagonal lattice. The leftmost plot shows the behavior from $y=0$ to $1.7$ (it goes to 0 at infinity). The phase transition takes place at $y_c\sim .658$ around which the middle plot is centered. The rightmost plot shows the logarithmic singularity of the derivative of $\langle j\rangle$ at $y_c$.}
\end{figure}

The quantity $\langle 2j\rangle$ as well as its derivative with respect to the temperature can be rewritten in terms of standard thermodynamical quantities. Let us rescale $y$ by the inverse temperature $\beta$ explicitly. Then the internal energy $U(\beta y)$ is such that
\begin{equation}
g(\beta y) = - \frac{U(\beta y)}{y\ \#\Edges}.
\end{equation}
Further, the derivative of $\langle 2j\rangle$ reads
\begin{equation}
\partial_\beta \langle 2j\rangle = 2y \biggl(\sinh (2\beta y) - \cosh(2\beta y)\,\frac{\partial_\beta \ln Z^{Ising}}{y\ \#\Edges} \biggr) - \sinh (2\beta y) \frac{\partial^2_{\beta}\ln Z^{Ising}}{y\ \#\Edges}.
\end{equation}
and introducing the heat capacity $C(\beta y) = \beta^2 \partial^2_\beta \ln Z^{Ising}$, we get
\begin{equation}
\partial_\beta \langle 2j\rangle = 2y \biggl(\sinh (2\beta y) + \cosh(2\beta y)\,\frac{U(\beta y)}{y\ \#\Edges} \biggr) - \sinh (2\beta y) \frac{C(\beta y)}{\beta^2\,y\ \#\Edges}.
\end{equation}

Plots of $\langle j+1/2\rangle$ are given in the figure \ref{fig:PlotMeanJ}, using the above formula for $g(y)$. The Ising phase transition occurs at $\tanh y_c = 1/\sqrt{3}$ (hence $y_c\sim .658$) which corresponds to $k(y_c)=1$. The singularity comes from the $k\sim 1$ behavior of $a(k), b(k)$, as for instance
\begin{equation*}
b(k) \underset{k\sim 1}{\sim} \frac{1-k^2}{\pi} \ln \frac{16}{|1-k^2|}.
\end{equation*}
In particular \cite{BaxterBook}, one finds around $t = 1-T/T_c$ close to zero
\begin{equation}
\langle 2j\rangle \sim \text{Cst}\times t\,\ln |t| + \text{analytic},
\end{equation}
which gives a singularity of the type $\ln |t|$ for $\partial_\beta \langle 2j\rangle$.

To conclude this section, the duality of the spin network evaluations with the 2D Ising model allows to identify a phase transition in the behavior of their generating  function and corresponding coherent spin network state. This is clearly visible in the behavior of the average of the edge spin $\la j_{e}\ra$, which we computed from the 2-point correlation of the Ising model between nearest neighbors, which shows a discontinuity in its derivative. Let us point out that this average can be negative: as explained earlier in section \ref{state}, here $\la j_{e}\ra$ is not the expectation value of the spin operator on the coherent spin network state but the average of the spin $j_{e}$ weighted by the spin network evaluations (which can be negative and do not define a true probability distribution). In this context,  $\la j_{e}\ra$ can not actually be interpreted as the average edge length of a dual triangulation to our graph, as suggested earlier in section \ref{state} by the analysis of the stationary point approximation of the expectation value $\la j_{e}\ra_{coh}$. However, it does plays the same role as the nearest-neighbors correlations in the Ising model. It would nevertheless be enlightening to understand the geometrical interpretation of $\la j_{e}\ra$, in particular from the  perspective of this stationary point approximation of the spin network generating function and its geometrical interpretation; doing this could shed light on a potential relation to a phase transition and continuum limit of discrete geometry in quantum gravity.

To this purpose, we conclude this paper by the introduction of slightly different coherent spin network states, and thus of generating function for spin network evaluations, as defined in the context of quantum gravity \cite{coherent}, which admit a different (and maybe better behaved) stationary point with a clear geometric interpretation, and which could be relevant for future investigation of this phase transition.

%%%%%%%%
\section{Vertex Integrals and Coherent Spin Network States}
\label{vertexint}
%%%%%%%%

We introduce a reformulation of the Gaussian integral for the spin network generating function over half-edge complex variables as a Gaussian integral over complex variables living at the vertices of the graph $\G$. This will allow us to define new coherent spin network states with an improved behavior for the stationary point analysis. But it should also allow us in future investigations to tackle the case of graphs with nodes of arbitrary valency.

%%%
\subsection{The Generating Function as a Vertex Integral}
%%%

We use the spinor notations introduced in \cite{spinor} and defined earlier in section \ref{state}.
We recall the expression of the generating function for spin network evaluations in terms of the angle couplings $X_{\alpha}$ (Theorem \ref{teo:complexgaussian}):
\be
Z^{Spin}(\{X_{\alpha}\})
\,=\,
\int_{\C^{4\#\Edges}}\prod_{e,v}\f{e^{-\la z_{e}^{v}|z_{e}^{v}\ra}\,d^{4}z_{e}^{v}}{\pi^{2}}\,
e^{\sum_{e} \la z_{e}^{s(e)}|z_{e}^{t(e)}]}\,e^{\sum_{\alpha}X_{\alpha}[z_{s(\alpha)}|z_{t(\alpha)}\ra}
\ee
where we changed the sign in front of $ \sum_{e}\la z_{e}^{s(e)}|z_{e}^{t(e)}]$ by a linear symmetry sending $\overline{z}_e^v\to -\overline{z}^v_z$ and having a trivial Jacobian equal to $1$. 
This formula suggests considering angle couplings of the type $X_{\alpha}=[\zeta_{s(\alpha)}|\zeta_{t(\alpha)}\ra$ where $\zeta_{e}^{v}$ are arbitrary fixed spinors living on the half-edges similarly to the integration variables $z_{e}^{v}$. In that case, we can apply the following proposition, developed in the spinor formulation of loop quantum gravity:
\begin{prop}
Let $\Om =(1,\,0)$, $\zeta_i,z_i, i\in I$ be vectors in $\C^2$ and $I$ be a finite set; then the following holds:
\be
\int_{\C^{2}} \f{e^{-\la z|z\ra}\,d^{4}z}{\pi^{2}}\,
e^{\sum_{i\in I}[\zeta_{i}|z\ra\la \Om|z_{i}\ra+[\zeta_{i}|z][ \Om|z_{i}\ra}
\,=\,
e^{\sum_{i< j}[\zeta_{i}|\zeta_{j}\ra[z_{i}|z_{j}\ra}.
\ee
\end{prop}
Let us point out that due to the anti-symmetry of the scalar products $[\zeta_{i}|\zeta_{j}\ra$ and $[z_{i}|z_{j}\ra$, the condition $i<j$ is not particularly relevant. It only serves to avoid over-counting and does not reflect the necessity of a linear ordering.
This proposition can be proved by putting together the following beautiful integral over $\SU(2)$,
\begin{lemma}
\be
\int_{\SU(2)}dg\,
e^{\sum_{i}[\zeta_{i}|g|z_{i}\ra}
\,=\,
\sum_{J\in\N}\f1{J!(J+1)!}\left(
\sum_{i< j}[\zeta_{i}|\zeta_{j}\ra[z_{i}|z_{j}\ra
\right)^{J}
\ee
\end{lemma}
with the following relation between the Haar measure on $\SU(2)$ and the Gaussian measure over spinors:
\begin{lemma}
From \cite{spinor,bonzom}, the integral of a homogeneous polynomial $P(g)$ in $g\in\SU(2)$ of even degree $2J $ can be expressed as a Gaussian integral over $\C^{2}$:
\be
\int_{\SU(2)}dg\,P(g)
\,=\,
\f1{(J+1)!}\int_{\C^{2}} \f{e^{-\la z|z\ra}\,d^{4}z}{\pi^{2}}\,P(z)\,,
\ee
with $g=|z\ra\la \Om|+|z][ \Om|$, with unit spinor $\Om =(1\,\,0)$.
\end{lemma}

The above results allow to reformulate the spin network generating function, introducing intermediate spinor variables $\xi_{v}\in\C^{2}$, living at the vertices $v$:
\be
Z^{Spin}(\G,\{X_{\alpha}
=[\zeta_{s(\alpha)}|\zeta_{t(\alpha)}\ra\})
\,=\,
\int_{\C^{4\#\Edges}}\int_{\C^{2\#\Ver}}
\prod_{e,v}\f{e^{-\la z_{e}^{v}|z_{e}^{v}\ra}\,d^{4}z_{e}^{v}}{\pi^{2}}\,
\prod_{v}\f{e^{-\la \xi_{v}|\xi_{v}\ra}\,d^{4}\xi_{v}}{\pi^{2}}\,
\,
e^{\sum_{e} \la z_{e}^{s(e)}|z_{e}^{t(e)}]}\,
e^{\sum_{v,e\ni v}[\zeta_{e}^{v}|\xi_{v}\ra\la \Om|z_{e}^{v}\ra+[\zeta_{e}^{v}|\xi_{v}][ \Om|z_{e}^{v}\ra}
%e^{\sum_{\alpha}X_{\alpha}[z_{s(\alpha)}|z_{t(\alpha)}\ra}
\ee
We now use the following beautifully simple lemma of complex integration:
\begin{lemma}
Considering an arbitrary holomorphic function $f$ of spinors $z\in\C^2$, the following integration identity using the Gaussian measure over the space of spinors holds:
$$
\int_{\C^{2}} \f{e^{-\la z|z\ra}\,d^{4}z}{\pi^{2}}\,
f(z)e^{\la z|\xi\ra}=f(\xi)\,.
$$
\end{lemma}
This allows to integrate over all the half-edge spinors $z_{e}^{v}$, leaving the generating function as an integral over vertex spinors only:
\be\label{snvertexspinors}
Z^{Spin}(\G,\{X_{\alpha}
=[\zeta_{s(\alpha)}|\zeta_{t(\alpha)}\ra\})
\,=\,
\int_{\C^{2\#\Ver}}
\prod_{v}\f{e^{-\la \xi_{v}|\xi_{v}\ra}\,d^{4}\xi_{v}}{\pi^{2}}\,
\,
e^{-\sum_{e} [ \zeta_{e}^{s(e)}|\, \left(
|\xi_{s(e)}\ra\la\xi_{t(e)}|+|\xi_{z(e)}][\xi_{t(e)}|
\right)
\,|\zeta_{e}^{t(e)}\ra}\,.
%e^{\sum_{e} \la z_{e}^{s}|z_{e}^{t}]}\,
%e^{\sum_{v,e\ni v}[\zeta_{e}^{v}|\xi_{v}\ra\la \Om|z_{e}^{v}\ra+[\zeta_{e}^{v}|\xi_{v}][ \Om|z_{e}^{v}\ra}
\ee
This is the spinfoam amplitude, recently considered in \cite{laurent,bonzom}. These few steps allow us to bridge between the formulas for the spin network generating function as an integral over variables living on vertices as in \cite{laurent, bonzom} and on half-edges as in \cite{CoMa}.

\smallskip

This naturally leads to two questions. First, at  a technical level, this vertex integral reformulation assumes the definition of the angle couplings from half-edge spinors as $X_{\alpha}=[\zeta_{s(\alpha)}|\zeta_{t(\alpha)}\ra$. This is different from the couplings defined from the Ising edge couplings,  $X_{\alpha}=\sqrt{Y_{s(\alpha)}Y_{t(\alpha)}}$. How compatible are those two definitions? Since we already know that edge couplings are equivalent to angle couplings from the perspective of loop observables as explained earlier in Section \ref{sec:mappings}, we can focus on asking how generic is the form $X_{\alpha}=[\zeta_{s(\alpha)}|\zeta_{t(\alpha)}\ra$ for angle couplings.
We claim that it is completely general. Indeed around a 3-valent vertex, with angle couplings $X_{12},X_{23},X_{31}$ in $\C$, there always exist three spinors $|\zeta_{1}\ra, |\zeta_{2}\ra, |\zeta_{3}\ra\in(\C^{2})^{3}$ such that $X_{i,i+1}=[\zeta_{i}|\zeta_{i+1}\ra$. Indeed possible choices are:
$$
|\zeta_{1}\ra=
\mat{c}{1 \\ 0},\qquad
|\zeta_{2}\ra=
\mat{c}{0 \\ X_{12}},\qquad
|\zeta_{3}\ra
=-\mat{c}{X_{23}/X_{12} \\ X_{31}}.
$$
Then we can act on these three spinors with an arbitrary $\SL(2,\C)$ matrix without changing their products $[\zeta_i,\zeta_j\ra$.

The second question is more intricate. This expression of the spin network generating function as an integral over vertex variables hints towards a similar reformulation of the Ising model possibly as an integral over odd-Grassmann variables living on vertices. We can also wonder if there exists a supersymmetry relating the new vertex variables $\xi_{v}$ to the Ising integration variables $\psi_{e}^v, \eta_{e}^v$. We leave these very interesting issues for future investigation.

%%%
\subsection{Another Class of Coherent States and Generating Function}
\label{coherent}
%%%

We have reformulated the spin network generating function alternatively as a complex Gaussian integral over spinors associated to half-edges of the graph (Theorem \ref{teo:complexgaussian}) or associated to vertices (Formula \eqref{snvertexspinors}), for angle couplings defined as $X_{\alpha}=[\zeta_{s(\alpha)}|\zeta_{t(\alpha)}\ra$ in terms of fixed half-edge spinors $\zeta_{e}^{v}$:
\beq
Z^{Spin}(\G,\{X_{\alpha}\})
&=&
\int_{\C^{4\#\Edges}}\prod_{e,v}\f{e^{-\la z_{e}^{v}|z_{e}^{v}\ra}\,d^{4}z_{e}^{v}}{\pi^{2}}\,
e^{\sum_{e} \la z_{e}^{s(e)}|z_{e}^{t(e)}]}\,e^{\sum_{\alpha}X_{\alpha}[z_{s(\alpha)}|z_{t(\alpha)}\ra}
\nn\\
&=&
\int_{\C^{2\#\Ver}}
\prod_{v}\f{e^{-\la \xi_{v}|\xi_{v}\ra}\,d^{4}\xi_{v}}{\pi^{2}}\,
\,
e^{-\sum_{e} [ \zeta_{e}^{s(e)}|\, \left(
|\xi_{s(e)}\ra\la\xi_{t(e)}|+|\xi_{z(e)}][\xi_{t(e)}|
\right)
\,|\zeta_{e}^{t(e)}\ra}\,.\nn
\eeq
We can generalize this re-expression to the spin network state defined as a gauge-invariant function on $\SU(2)^{\#\Edges}/\SU(2)^{\#\Ver}$:
\beq
\phi_{\G,\{X_{\alpha}\}}(g_{e})
&=&
\int_{\C^{4\#\Edges}}\prod_{e,v}\f{e^{-\la z_{e}^{v}|z_{e}^{v}\ra}\,d^{4}z_{e}^{v}}{\pi^{2}}\,
e^{\sum_{e} \la z_{e}^{s(e)}|g_{e}|z_{e}^{t(e)}]}\,e^{\sum_{\alpha}X_{\alpha}[z_{s(\alpha)}|z_{t(\alpha)}\ra}
\nn\\
&=&
\int_{\C^{2\#\Ver}}
\prod_{v}\f{e^{-\la \xi_{v}|\xi_{v}\ra}\,d^{4}\xi_{v}}{\pi^{2}}\,
\,
e^{-\sum_{e} [ \zeta_{e}^{s(e)}|\, 
H_{s(e)}\,g_{e}H_{t(e)}^\dagger
\,|\zeta_{e}^{t(e)}\ra}\,,\nn
\eeq
where we have defined the $2\times 2$ matrices $H_{v}$'s in terms of the integration spinors $\xi_{v}$'s as:
\be
H_v=|\xi_{v}\ra\la \Om|+|\xi_{v}][ \Om|\,.
\ee
This expression naturally suggests to switch to $\SU(2)$ integrations instead of Gaussian integrals over the spinor space $\C^2$ and we introduce the other class of gauge-invariant coherent states:
\beq
\phi^{cl}_{\G,\{X_{\alpha}\}}(g_{e})
&\equiv&
\int_{\SU(2)^{\#\Ver}}
\prod_{v}dh_{v}\,
\,
e^{-\sum_{e} [ \zeta_{e}^{s(e)}|\,h_{s(e)}^{-1}g_{e}h_{t(e)}\,|\zeta_{e}^{t(e)}\ra} \\
&=&
 \int_{\C^{4\#\Edges}}\prod_{e,v}\f{e^{-\la z_{e}^{v}|z_{e}^{v}\ra}\,d^{4}z_{e}^{v}}{\pi^{2}}\,
e^{\sum_{e} \la z_{e}^{s(e)}|g_{e}|z_{e}^{t(e)}]}\,
\prod_{v}\sum_{J\in\N}\f1{J!(J+1)!}\left(\sum_{\alpha \ni v}
[\zeta_{s(\alpha)}|\zeta_{t(\alpha)}\ra[z_{s(\alpha)}|z_{t(\alpha)}\ra
\right)^{J}\,,\nn\\
&=&
 \int_{\C^{4\#\Edges}}\prod_{e,v}\f{e^{-\la z_{e}^{v}|z_{e}^{v}\ra}\,d^{4}z_{e}^{v}}{\pi^{2}}\,
e^{\sum_{e} \la z_{e}^{s(e)}|g_{e}|z_{e}^{t(e)}]}\,
\prod_{v}\sum_{J\in\N}\f1{J!(J+1)!}\left(\sum_{\alpha \ni v}
X_{\alpha}[z_{s(\alpha)}|z_{t(\alpha)}\ra
\right)^{J}\,,
%e^{\sum_{\alpha}X_{\alpha}[z_{s(\alpha)}|z_{t(\alpha)}\ra}
\nn
\eeq
These are actually the coherent spin network states for loop quantum gravity and spinfoam models, introduced earlier in \cite{coherent}, further developed in \cite{bonzom} and recently used in studying the coarse-graining of the spinfoam path integral in \cite{dupuislivine}.
As one can see, the difference with the coherent states defined in Equation \eqref{spin:coherent1} and the above is that we have been using is the extra-factors $(J+1)!^{-1}$, turning the exponentials into Bessel functions. The two coherent states are actually related to each other by a Borel transform or reversely by a inverse Laplace transform. 

These new coherent states were introduced in the context of loop quantum gravity (in $3+1$ dimensions) as describing good semi-classical states of 3d geometry. Here the geometrical meaning of those states is different since we seek for an interpretation in terms of 2D triangulations dual to our planar graph $\Gamma$. We will show below that they define slightly different weights for the generating function of spin network evaluations, which leads to a non-scalable stationary point uniquely fixed by the angle couplings $X_{\alpha}$, or equivalently the fixed half-edge spinor data $\zeta_{e}^v$.
The introduction of the $(J+1)!^{-1}$ factor for every vertex leads to  a modified generating functional, to be compared with \eqref{eq:genser}:
\be\label{eq:newspin}
\tZ^{Spin}
%(\{Y_{e}\})
\,=\,
\sum_{\{j_e\}} \sqrt{\frac{1}{\prod_v (J_v+1)!\prod_{e\ni v} (J_v-2j_e)!}} s(\{j_e\}) \prod_e Y_e^{2j_e}\,,
\ee
and the associated modified probability distribution for the edge spins, to be compared with \eqref{rho}:
\be
\trho(\{j_e\})
\,=\,
\prod_v{\frac{ 1}{(J_v+1)!\prod_{e\ni v} (J_v-2j_e)!}} \prod_e \f{(|Y_e|^{2})^{2j_e}}{(2j_{e}+1)}\,.
\ee
Here the edge couplings $Y_{e}$ are given in terms of the angle couplings $X_{\alpha}$ according to Equation \eqref{angletoedge}:
\be
\label{XtoY}
Y_{e}^{2}=\f{X_{ee_{1}}X_{ee_{2}}X_{e\te_{1}}X_{e\te_{2}}}{X_{e_{1}e_{2}}X_{\te_{1}\te_{2}}}\,,
\ee
where  $e_{1}$ and $e_{2}$ are the other two edges attached to the source vertex $s(e)$, while $\te_{1}$ and $\te_{2}$ are the other two edges attached to the target vertex $t(e)$.
If we were to plug back the relation $X_{\alpha}=\sqrt{Y_{s(\alpha)}Y_{t(\alpha)}}$ in that definition of the $Y$'s, then it would be simply a consistency check. But here the $X$'s are defined instead in terms of the half-edge spinor variables $\zeta_{e}^{v}$.

The stationary point analysis can be done exactly as earlier in Section \ref{stationary}. Assuming that the spins $j$'s are large and using the Stirling approximation for the factorials, we get the conditions at leading order in the spins:
\be
\label{newJe}
|Y_{e}|^{4}\approx
\f{J_{s(e)}(J_{s(e)}-2j_{e_{1}})(J_{s(e)}-2j_{e_{2}})}{(J_{s(e)}-2j_{e})}\,
\f{J_{t(e)}(J_{t(e)}-2j_{\te_{1}})(J_{t(e)}-2j_{\te_{2}})}{(J_{t(e)}-2j_{e})}\,.
\ee
These equations are solved again if the spins $j_{e}$ are related to the edge lengths $l_{e}=2j_{e}$ of a dual triangulation to our planar graph $\Gamma$  with the condition that the edge couplings $Y_{e}$ are related to the geometric data  by:
\be
Y_{e}^{2}=L_{s(e)}\tan\f{\gamma_{e}^{s(e)}}2\,\,L_{t(e)}\tan\f{\gamma_{e}^{t(e)}}2\,,
\ee
in terms of the half-perimeters $L_{v}$ and the opposite angles $\gamma_{e}^{v}$. Unlike before, the $Y_{e}$ scale with the triangulation lengths and are not scale-invariant anymore. Given admissible couplings admitting a stationary point, we can not rescale the spins of that stationary point: we thus have an actual stationary point and not a stationary line as in Section \ref{stationary}. 

\smallskip

Having started with the spinor data $\zeta_{e}^{v}$ allows a finer analysis of the geometry of the fixed point. Indeed, comparing the two equations above \eqref{newJe} and \eqref{XtoY}, the conditions on the edge couplings for the stationary point are solved if for every angle $\alpha$ around every vertex $v$, we have the matching:
\be
J_{v}(J_{v}-2j_{\hat{\alpha}})
\,=\,
|X_{\alpha}|^{2}
\,=\,
\Big{|}[\zeta_{s(\alpha)}|\zeta_{t(\alpha)}\ra\Big{|}^{2}\,,
\ee
where  we have written $\hat{\alpha}$ for the edge opposite to the angle $\alpha$ around the vertex $v$.
More explicitly, around every (3-valent) vertex $v$, the spins $j_{e}$ at the fixed point are given in terms of the spinor data $\zeta_{e}^{v}$ by the relations:
\be
J=j_{1}+j_{2}+j_{3},\quad
J(J-2j_{1})=|[\zeta_{2}|\zeta_{3}\ra|^{2}, \quad \dots
\ee
Following the work on coherent intertwiners in the context of loop quantum gravity and its discrete geometry \cite{spinor,sfspinor}, we can easily solve such equations to get the spins $j_{e}$ by introducing the 3-vectors $\vV\in\R^{3}$ corresponding to the spinors $\zeta\in\C^{2}$ by projecting them onto the Pauli matrices:
\be
\vV=\f12\,\la \zeta|\vsigma|\zeta\ra,\qquad
|\vV|=\la \zeta|\zeta\ra\,,
\ee
where the Pauli matrices $\sigma^{a}$, $a=1..3$, are normalized such that they square to the identity. The previous conditions can be entirely re-written in terms of those 3-vectors:
$$
J(J-2j_{1})=\f12(|\vV_{2}||\vV_{3}|-\vV_{2}\cdot\vV_{3})\,.
$$
Then, letting $V_i:=|\vec{V}_i |$ and summing over all 3 such equations around the vertex, we define the total norm $V=\sum_{i }V_{i}$ and the closure vector $\vV=\sum_{i }\vV_{i}$, to get:
\be
4J^{2}= (V^{2}-\vV\cdot \vV)\,,
\qquad
2j_{i}=\f1{\sqrt{V^{2}-\vV\cdot \vV}}\,(VV_{i}-\vV\cdot\vV_{i})\,,
\ee
where the spins at the stationary point are indeed fixed entirely by the value of the generating function couplings $\zeta_{e}^{v}$. If we further assume closure constraints on the initial spinor data, as introduced in \cite{coherent,spinortwisted},
\be
\sum_{i}\vV_{i}=0,\qquad
\textrm{or equivalently}\qquad
\sum_{i} |\zeta_{i}\ra\la\zeta_{i}|\propto \id
\ee
the expression of the spins at the stationary point further simplifies to:
$$
2J=V,\quad 2j_{i}=V_{i}\,.
$$
We obviously have to require norm-matching of the spinors on both ends of every edge $e$, in order to get the same spin $j_{e}$ from the equations at both it source and target vertices.

\smallskip

To summarize the previous analysis, we start with spinor data living on the half-edges $\zeta_{e}^{v}$ and satisfying closure constraints at the vertices and norm-matching on the edges:
\be
\forall v,\quad
\sum_{e\ni v}  |\zeta_{e}^{v}\ra\la\zeta_{e}^{v}|\propto \id,
\qquad
\forall e,\quad
\la \zeta_{s(e)}|\zeta_{s(e)}\ra
=\la \zeta_{t(e)}|\zeta_{t(e)}\ra\,\,.
\ee
In this twisted geometry setting \cite{twisted, spinor, spinortwisted, sfspinor}, there exists a triangulation dual to our planar 3-valent graph such that the edge lengths $l_i$ are the norm squared $V_i$ of the corresponding spinors. Then the new spin probability distribution $\trho$ is peaked on a stationary point, such that the spins $j_i$ are half those edge lengths.

\smallskip

This modified spin probability distribution has a very nice geometrical interpretation, which made it very useful to construct semi-classical states peaked on classical discrete geometries in the framework of loop quantum gravity. From our perspective here, the existence of the stationary point instead of the scale-invariant stationary line suggests that the corresponding spin network generating function may not have poles and therefore may not exhibit critical behavior. 
It would be very  interesting to investigate this further and analyze in details the behavior of both spin network generating functions by a careful stationary point analysis involving not only the spin statistical weight but also the asymptotical behavior of the spin network evaluations (as studied for instance in e.g. \cite{CoMa}), and study the counterpart of those from the viewpoint of the 2D Ising model.

%%%%%%
\section{Conclusion \& Outlook}\label{sec:conclusions}
%%%%%%

In the present paper we related the partition function of the Ising model on a planar trivalent graph $\Gamma$ to the generating series of the spin network evaluations on the same graph. Although we explored some of the first consequences of this duality, many open questions remain which we will explore in the future. 
\begin{itemize}
\item Is it possible to extend the correspondence to non-planar graphs and/or graphs with vertices of arbitrary degrees? In \cite{CoMa} a general formula was obtained for the trivalent non-planar case. In loop quantum gravity, coherent states on graphs with vertices of degrees greater than 3 also exist. We expect that those formulae to be related on the Ising side to formulae like the Kac-Ward formula (see \cite{KacWard}), where Kasteleyn orientations are extended to spin structures. 
\item In \cite{CoMa} a Gaussian integral formulation was provided also for spin network evaluations of graphs equipped with $\mathrm{SL}_2(\C)$ holonomies on the edges. From the point of view of Loop Quantum Gravity, they are natural objects to study. Is there a counterpart of holonomies on the side of the Ising model? What is the link with the spin network evaluations?
\item When the graph is embedded in the 3-sphere or 3-hyperboloid, spin network evaluations have famous deformations known as \emph{quantum spin networks}, based on the q-deformed $\cU_{q}(\SU(2))$. Is there a statistical model which maps to the quantum evaluations via supersymmetry and which deforms the Ising model? In the positive case, what is the interpretation in terms of this model of the conjectures on the asymptotic behavior of quantum spin networks (\cite{CM},\cite{coCV}, \cite{CGV})? 
\item Our analysis in Section \ref{state} is a first analysis of the relations between stationary points of the distribution of the spins induced by spin network coherent states and critical values of the couplings for the Ising model. Can we better understand this relation, which at this stage might just look like a coincidence? Our analysis gives a relation with the critical temperatures for the Ising models only if the graph is isoradial. Can one push further the analysis to non-isoradial graphs?
\end{itemize}

Then there is much to do on the interplay between spin network evaluations and the Ising model in the thermodynamic limit. In particular, we might be able to use methods developed for the Ising model to study the generating function of spin network evaluations. Also comparing the equations which appear in the resolution of both models (Biedenharn-Elliott identity and Yang-Baxter equations for instance) should lead to a fruitful cross-fertilization.
This opens new perspectives on the coarse-graining of spin networks and their spin foam transition amplitudes for quantum gravity, which might meet the intertwiner dynamics models developed in \cite{bianca2} and their interplay with condensed matter models. Here is an outlook for potential future investigation:
%on their critical behavior and properties under coarse-graining.

\begin{itemize}
\item The fact that the stationary points found in Section \ref{state} for the distribution induced by the coherent spin network state on the length operator are the same as the critical values for the Ising model in the thermodynamic limit seems a bit mysterious. The stationary points we found should be compared with those of the stationary phase approximation of the generating function of spin network evaluations $Z^{spin}$ itself. As the latter takes the form $Z^{spin} = 1/P_\Gamma^2$ where $P_\Gamma$ is a polynomial, we expect the stationary phase method to reveal the singularities of $Z^{spin}$. They obviously are the zeroes of $P_\G$, hence the Fisher's zeroes of the corresponding Ising partition function, as $Z^{Ising} \sim P_\Gamma$. This is a first way to study the thermodynamic limit, as the distribution of Fisher's zeroes in this limit carries information on the critical behavior \cite{Fisher}.

\item Coarse-graining techniques are well-developed for the Ising model. The duality we exhibited suggests that they could be applied to spin network evaluations. We expect the star-triangle transformation to be directly the Ising equivalent of the behavior of spin network evaluations under Pachner moves. For instance, is there an equivalence between the Yang-Baxter equation for the Ising model and the Biedenharn-Elliott identity for Wigner symbols?

\item The duality of the 2D Ising model between low and high temperatures is expected to carry over to the spin network side. Does it then correspond to some UV-IR duality in 3D quantum gravity amplitudes? This typically requires to use non-trivalent graphs, as the duality pushes a model from the primal graph to its dual.

\item Finally, properties of spin network evaluations might be described in the continuum as a conformal field theory, dual to that of the continuum Ising model.
Indeed, on the one hand, the critical 2D Ising model is conformally invariant and can be described as a Wess-Zumino-Witten model for the coset $\SU(2)_{1}\times\SU(2)_{1}/\SU(2)_{2}$ (were the indices denotes the level of the theory) (see e.g. \cite{Ising CFT}). One can consider its ``inverse'' by going to negative level or equivalently replacing the gauge group $\SU(2)$ by the coset $\SL(2,\C)/\SU(2)$  \cite{gawedzki}.
On the other hand, the spin network evaluation live on the boundary of the spinfoam model for 3D quantum gravity, which is described in the continuum by a Chern-Simons theory in the bulk for the Poincar\'e group inducing a Wess-Zumino-Witten theory on its boundary. It would be enlightening if the two sides of the story would meet.
Could this fit in the paradigm of gravity/CFT correspondence?

\end{itemize}

\end{document}